\documentclass[10pt,a4paper]{article}

\usepackage{amsmath,amsgen,latexsym}
\usepackage{amstext,amssymb,amsfonts,latexsym}
\usepackage{theorem}
\usepackage{pifont}
\usepackage[hyphens]{url}

 \newcommand{\bs}{\bigskip}
 \newcommand{\ms}{\medskip}
 \newcommand{\n}{\noindent}
 \newcommand{\s}{\smallskip}
 \newcommand{\hs}[1]{\hspace*{ #1 mm}}
 \newcommand{\vs}[1]{\vspace*{ #1 mm}}

 \newcommand{\setempty}{\varnothing}
 \newcommand{\real}{\mathbb{R}}
 \newcommand{\nat}{\mathbb{N}}
 
 \newcommand{\integer}{\mathbb{Z}}
 
 \newcommand{\complex}{\mathbb{C}}

 \newcommand{\etalc}{\textrm{et al.}}

 \newcommand{\AAA}{{\cal A}}

 \newcommand{\FF}{{\cal F}}

 \newcommand{\KK}{{\cal K}}
 \newcommand{\LL}{{\cal L}}
 \newcommand{\NN}{{\cal N}}
 \newcommand{\GG}{{\cal G}}
 \newcommand{\MM}{{\cal M}}
 \newcommand{\TT}{{\cal T}}
 \newcommand{\SSS}{{\cal S}}
 \newcommand{\PP}{{\cal P}}

 \newcommand{\np}{\mathrm{NP}}

 \newcommand{\poly}{\mathrm{poly}}

 \newcommand{\reg}{\mathrm{REG}}
 \newcommand{\cfl}{\mathrm{CFL}}
 \newcommand{\ucfl}{\mathrm{UCFL}}

 \newcommand{\matrices}[4]{\left(\:\begin{array}{cc} #1 & #2 \\%
      #3 & #4   \end{array} \right)}

\theoremstyle{plain}
\theoremheaderfont{\bfseries}
 \newtheorem{theorem}{Theorem}[section]
 \newtheorem{lemma}[theorem]{Lemma}
 \newtheorem{proposition}[theorem]{{\bf Proposition}}

 {\theorembodyfont{\rmfamily}
  \newtheorem{definition}[theorem]{Definition}}
 {\theorembodyfont{\rmfamily} \newtheorem{example}[theorem]{Example}}
 {\theorembodyfont{\rmfamily} }

 \newenvironment{proofof}[1]{\vspace*{5mm} \par \noindent
         {\bf Proof of #1.\hs{2}}}{\hfill$\Box$ \vspace*{3mm}}

 \newenvironment{proof}{\par \noindent
            {\bf Proof. \hs{2}}}{\hfill$\Box$ \vspace*{3mm}}

 \newcommand{\ceilings}[1]{\lceil #1 \rceil}
 
 \newcommand{\pair}[1]{\langle #1 \rangle}

 \newcommand{\qubit}[1]{| #1 \rangle}
 \newcommand{\bra}[1]{\langle #1 |}
 \newcommand{\ket}[1]{| #1 \rangle}
 \newcommand{\measure}[2]{\langle #1 | #2 \rangle}
 \newcommand{\density}[2]{| #1 \rangle \! \langle #2 |}

\newcommand{\ignore}[1]{}

\newcommand{\track}[2]{[\:\begin{subarray}{c} #1 \\%
      #2 \end{subarray} ]}

\newcommand{\cent}{|\!\! \mathrm{c}}
\newcommand{\dollar}{\$}

 \newcommand{\bql}{\mathrm{BQL}}

 \newcommand{\aeqs}[1]{\mathrm{AEQS}( \mathrm{ #1 } )}

 \newcommand{\diag}{\mathrm{diag}}
 \newcommand{\trace}{\mathrm{tr}}
 \newcommand{\ilog}{\mathrm{ilog}\:}

\begin{document}

\pagestyle{plain}
\pagenumbering{arabic}
\setcounter{page}{1}
\setcounter{footnote}{0}

\begin{center}
{\Large {\bf How Does Adiabatic Quantum Computation \s\\ Fit into Quantum Automata Theory?}}\footnote{This current article extends and corrects its preliminary report that has appeared in the Proceedings of the
21st IFIP WG 1.02 International Conference on  Descriptional Complexity of Formal Systems (DCFS 2019), Ko\v{s}ice, Slovakia, July 17--19, 2019, Springer-Verlag,
Lecture Notes in Computer Science, vol. 11612, pp. 285--297, 2019.}
\bs\s\\

{\sc Tomoyuki Yamakami}\footnote{Present Affiliation: Faculty of Engineering, University of Fukui, 3-9-1 Bunkyo, Fukui 910-8507, Japan}
\end{center}
\ms


\begin{abstract}
Quantum computation has emerged as a powerful computational medium of our time, having demonstrated the remarkable efficiency in factoring a positive integer and searching databases faster than any currently known classical computing algorithm. Adiabatic evolution of quantum systems have been studied as a potential means that physically realizes quantum computation. Up to now, all the research on adiabatic quantum systems has dealt with polynomial time-bounded computation and little attention has been paid to, for instance, adiabatic quantum systems consuming only constant memory space. Such quantum systems can be modeled in a form similar to quantum finite automata. This exposition dares to ask a bold question of how to make adiabatic quantum computation fit into the rapidly progressing framework of quantum automata theory. As our answer to this eminent but profound question, we first lay out a fundamental platform to carry out adiabatic evolutionary quantum systems (AEQSs) with limited computational resources (in size, energy, spectral gap, etc.) and then establish how to construct such AEQSs by operating suitable families of quantum finite automata.
We further explore fundamental structural properties of decision problems (as well as promise problems) solved quickly by the appropriately constructed AEQSs.

\ms
\n{\bf Key words.} {adiabatic quantum computation, quantum finite automata, Hamiltonian, Schr\"{o}dinger equation, decision problem, promise problem}
\end{abstract}

\sloppy
\section{Background, Motivations, and Challenges}\label{sec:introduction}

We will explain historical background, motivational discussions, and challenging open questions, which trigger our intensive research of this exposition.

\subsection{Adiabatic Quantum Computation}\label{sec:quant-oracle}

The primary purpose of \emph{computation} is to solve given computational problems efficiently with as fewer resources as possible. Most of the computing devices at present time are in fact \emph{programmable machines}  that perform basic operations mechanically in enormous speed and precision.
As a new paradigm founded solely on quantum mechanics, \emph{quantum computation} has gained large popularity over the past few decades through numerous physical experiments. There are already several important milestones in our time that indicate the supremacy of quantum computation over the existing computers. Shor \cite{Sho97}
discovered a polynomial-time quantum algorithm to factor any positive integer and compute discrete logarithms  whereas Grover
\cite{Gro96,Gro97}
presented a quantum way to search for a given key in an unstructured database quadratically faster than the traditional search.

Early quantum-mechanical models of computation were proposed in the 1980s by Benioff \cite{Ben80} and Deutsch
\cite{Deu85} and these computational models were later refined by Yao \cite{Yao93} and Bernstein and Vazirani \cite{BV97} as \emph{quantum circuit} and \emph{quantum Turing machine} (QTM), which respectively  extend the classical models of Boolean circuit and Turing machine.
Beyond a single-tape QTM model of Bernstein and Vazirani,
a multi-tape QTM model was further studied in \cite{Yam99,Yam03} as well as \cite{ON00}.
A recursion-theoretic formulation  was also proposed in \cite{Yam17} with no use of machinery to capture quantum computation.
Those models are powerful enough to implement the quantum algorithms of Shor and Grover but are hard to realize physically as real-life computing devices.
Quantum Turing machines seem to be slightly more contrive than quantum circuits but they are a better manifestation of ``programmable'' computers because they are made of input/work tapes, tape heads, and finite-control units.
\emph{Quantum finite(-state) automata} can be seen as quantum Turing machines operating with  only constant work space.
Early restrictive models of quantum finite automata were studied by Moore and Crutchfield \cite{MC00} and by Kondacs and Watrous \cite{KW97}.

There also exists another model for classical computation, known as \emph{simulated annealing} (or thermal annealing), which has been implemented as a physical system performing various computations based on thermodynamics. \emph{Quantum annealing} was proposed based on quantum mechanics to extend simulated annealing and a tunneling effect in quantum mechanics makes quantum annealing outperform simulated annealing (see, e.g., \cite{SMTC02}).
In quantum annealing, computation is viewed as a process of evolution of a quantum state $\qubit{\psi(t)}$ at time $t$ in a given quantum system    according to the Schr\"{o}dinger equation $\imath\hbar \frac{d}{dt}\qubit{\psi(t)} = H(t)\qubit{\psi(t)}$ using a specified time-dependent \emph{Hamiltonian} $H(t)$ (which is a Hermitian matrix), where $\hbar$ is the \emph{reduced Planck constant} (i.e., Plank's constant $\approx 6.63\times 10^{-34}$ joule/second divided by $2\pi$).
In early 2000s,  Farhi, Goldstone, Gutmann, and Sipser \cite{FGGS00} and Farhi, Goldstone, Gutmann, Lapan, Lundgren, and Preda \cite{FGG+01} developed quantum algorithms based on a variant of quantum annealing, called \emph{adiabatic quantum computation}, in which an initial quantum system evolves to find a unique solution, which is represented by the  ground state (i.e., an eigenstate of the smallest eigenvalue) of a final quantum system. The adiabatic quantum algorithm of Farhi \etalc~\cite{FGGS00}, for instance, solves Search-2SAT (i.e., a search version of the satisfiability problem for 2CNF Boolean formulas). Another paper by Farhi \etalc~\cite{FGG+01} demonstrated how to solve an $\np$-complete problem, known as the Exact Cover Problem.
Later, van Dam, Mosca, and Vazirani \cite{DMV01} showed an exponential lower-bound for adiabatic quantum computation to solve the Minimum
Hamming Weight Problem.

To be more precise, \emph{adiabatic quantum computation} is dictated by  two Hamiltonians $H_{ini}$ and $H_{fin}$ of dimension $2^n$ and a closeness bound $\varepsilon$ such that the start of the system is $H_{ini}$'s unique ground state and the outcome of the system becomes the unique ground state of $H_{fin}$, provided that the ``uniqueness'' condition is guaranteed for $H_{ini}$ and $H_{fin}$.
The time-dependent Hamiltonian $H(t)$ is given as a linear combination $(1-\frac{t}{T})H_{ini}+\frac{t}{T}H_{fin}$ of $H_{ini}$ and $H_{fin}$.
Such a quantum system starts with the ground state $\qubit{\psi_g(0)}$ of the initial Hamiltonian $H(0)=H_{ini}$ at time $t=0$. If $H(t)$ changes sufficiently slowly, the evolving quantum state $\qubit{\psi(t)}$ stays close to the ground state $\qubit{\psi_g(t)}$ of $H(t)$. At time $t=T$, the ground state of the quantum system becomes $\qubit{\psi_g(T)}$ of the final Hamiltonian $H(T)=H_{fin}$. We demand that this ground state is sufficiently close to the quantum state encoding the desired solution of a target computational problem.

For the efficiency of adiabatic quantum computation, we are primarily concerned with the \emph{evolution time} $T$ of the quantum system and the \emph{structural complexity} of two Hamiltonians $H_{ini}$ and $H_{fin}$ of the system. The \emph{running time} of the system, which is determined by the minimal evolution time of the system as well as the maximal matrix norm of $H(t)$, is basically proportional to the reciprocal of the \emph{spectral gap} of $H_{ini}$ and $H_{fin}$ according to the so-called \emph{adiabatic theorem} \cite{Kat51,Mes58}. For the details, see Section \ref{sec:system-evolution}. A crucial point is how fast the evolution of the quantum system takes place. Unfortunately, it turns out that the aforementioned algorithm of Farhi $\etalc$ requires exponential time to execute \cite{DMV01}. Nevertheless, Aharonov, van Dam, Kemp, Landau, Lloyd, and Regev \cite{ADK+07} demonstrated how adiabatic quantum computation can simulate quantum-circuit computation efficiently.
In addition, van Dam, \etalc~\cite{DMV01} gave a detailed analysis of adiabatic quantum computation and presented how to simulate adiabatic quantum computation on quantum circuits.  The simulations of  Aharonov $\etalc$ and van Dam $\etalc$ together establish the (polynomial) equivalence between adiabatic quantum computation and machine-based  quantum computation.
Although adiabatic quantum computation is no more powerful than standard quantum computation, it seems to remain as significant potentials to realize restricted variants of quantum computation. With the current technology, however, it still seems to be difficult to build a large-scale adiabatic quantum computing device since making local evolution in a large system is quite sensitive to \emph{decoherence}. It is rather better to make global evolution in a small-scale quantum system.

\subsection{Complexity of Hamiltonians}

In an early stage, adiabatic quantum computing was used to solve optimization problems; in contrast, this exposition seeks to solve \emph{decision problems}  using adiabatic quantum computation. Decision problems can be treated as (formal) languages by identifying yes/no answers for the decision problems with inclusion/exclusion of inputs to languages. For such languages, we are more concerned with adiabatic quantum computation that can determine the acceptance/rejection of inputs rather than finding solutions.

Since the eigenvalue properties of two Hamiltonians dictate the performance of the adiabatic quantum systems, the key to adiabatic computation is how well we can prepare two essential Hamiltonians before the start of the computation.
Our task is therefore to find out how to encode appropriate witnesses of  answers to each membership question of a target language into Hamiltonians so that the constructed Hamiltonians automatically lead to the desired answers by adiabatic quantum evolution.

From a practical viewpoint, we are more interested in generating
Hamiltonians with fewer resources in the encoding process of the actual adiabatic quantum computation. The generation of Hamiltonians can be done in several different ways. A quantum Ising model, for example, provides a particular framework for constructing such Hamiltonians in terms of linear forms of Pauli matrices (see, e.g., \cite{Sac11}).
Unlike the quantum Ising model, Farhi \etalc~\cite{FGG+01} and van Dam \etalc~\cite{DMV01} presented a natural way to construct Hamiltonians  using quantum circuits as well as QTMs.
Since QTMs are a quite powerful computational model, it is desirable to place reasonable restrictions on their resources needed for the construction of the Hamiltonians.
It therefore remains more realistic to build Hamiltonians using ``resource-bounded'' quantum-mechanical devices.
In particular, we pay our attention to a \emph{constant-memory model} of quantum Turing machine, which is conceptually realized by \emph{quantum finite(-state) automata}, because, in general, quantum finite automata  have been considered as one of the most fundamental machine models of algorithmic computation.
To seek for future potentials of adiabatic quantum computation in such a realistic setting, it is worth considering quantum finite automata as an algorithmic tool in generating the desired Hamiltonians.

\subsection{Models of Quantum Finite Automata}\label{sec:automata-theory}

The \emph{theory of quantum finite automata} has made a remarkable progress  since the first installment of quantum finite automata in late 1990s
(see, e.g.,  \cite{AY15,Gru00}). The early models of quantum finite automata were rather simple in their mechanism.
In general, a quantum finite automaton takes an input string written on its read-only input tape and, as reading input symbols one by one, it changes its inner states in a quantum fashion until it finally terminates. This entire process can be described as a series of quantum transitions of inner states determined by scanned input symbols.
\emph{One-way measure-once quantum finite automata} of Moore and Crutchfield \cite{MC00} operate by applying predetermined sets of unitary transforms to superpositions of inner states as the machine scans input symbols one by one until reaching the endmarker, and then perform projective  measurements to determine the outcome of the quantum computation.
Since a tape head always moves to the next tape cell after reading each  tape symbol, the tape head can be viewed as a classical device.
In contrast, Kondacs and Watrous \cite{KW97} studied \emph{two-way measure-many quantum finite automata}, which make quantum moves and take projective measurements at  each application step of unitary transformations.
As variants and extensions of those one-way and two-way quantum finite automata, numerous models have been proposed in the literature.
To empower the early models of quantum finite automata, a more general model, known as \emph{two-way quantum finite automata with mixed states and quantum operations}, for example, was studied under various names \cite{ABG+06,FOM09,YS11}. This general model is (computationally) equivalent to a garbage-tape model of two-way quantum finite automata \cite{Yam19b}.
Recently, a ``nonuniform'' analogue of a quantum automata family has been discussed in \cite{VY15,Yam19b,Yam19c}. In particular, a nonuniform family of polynomial-size two-way quantum finite automata with garbage tapes nicely captures nonuniform logarithmic-space quantum computation \cite{Yam19b}.

\subsection{Our Challenges in This Exposition}\label{sec:challenge}

Since a new paradigm of adiabatic quantum computation looks quite different from the standard framework of quantum finite automata described in Section \ref{sec:automata-theory}, we face the following challenging question. Is it possible for us to make adiabatic quantum computation fit into the framework of quantum automata theory?
This exposition attempts to answer this question affirmatively by proposing a reasonable platform for an automata-theoretic discussion on the efficient construction of Hamiltonians that are necessary to carry out the desired  adiabatic quantum  computation.

For a further discussion, nevertheless, it is quite useful to set up a formal quantum system that realizes adiabatic quantum computation in such a way that we can handle it using even memory-restricted computing devices, in particular,
quantum finite automata.
Aiming at capturing an essence of adiabatic quantum computation in terms of quantum finite automata, we first lay out a scaled-down model of adiabatic quantum computation, which we call an \emph{adiabatic evolutionary quantum system} (AEQS, pronounced as ``eeh-ks''). An AEQS $\SSS$ is composed of an input alphabet $\Sigma$, a system size parameter $\mu$ mapping $\Sigma^*$ to natural numbers, an accuracy bound $\varepsilon\in[0,1]$, two $2^{m(x)}$-dimensional Hamiltonians $H^{(x)}_{ini}$ and $H^{(x)}_{fin}$ for every input $x$, and an acceptance/rejection criteria pair $(S^{(n)}_{acc},S^{(n)}_{rej})$ for each system size $n$.
Given every input $x\in\Sigma^*$, $\SSS$ uses $H^{(x)}_{ini}$ and $H^{(x)}_{fin}$ to perform adiabatic quantum computation. When the computation terminates, $\SSS$ accepts or rejects the input if the ground state of $H^{(x)}_{fin}$ is close enough to a Hilbert space spanned by the vectors in  $S^{(m(x))}_{acc}$ or $S^{(m(x))}_{rej}$, respectively.
The precise definition of AEQS will be given in Section \ref{sec:AEQS}.

This exposition proposes the use of quantum finite automata as a mechanical tool to construct the Hamiltonians of AEQSs step by step as we read input symbols one by one.
For this purpose, we need to modify the original form of quantum finite automata by removing initial inner states and accepting/rejecing inner states.
The modified automata are conventionally dubbed as \emph{quantum quasi-automata} and we need to discuss how to design (or program) such machines to produce the desired Hamiltonians.

The rest of this exposition will be organized as follows.
After giving in Section \ref{sec:preparation} the basic notions and notation necessary for later discussions,  we will explain the mechanism of an adiabatic evolution of quantum systems in Section \ref{sec:system-evolution}. As a model of such adiabatic evolutionary system, we will describe AEQSs in details and show  in Section \ref{sec:AEQS} that AEQSs are powerful enough to solve all decision problems.
Our tools of quantum quasi-automata will be explained in Section \ref{sec:quasi-automata}.
In Section \ref{sec:example}, we will demonstrate how to design (or program) AEQSs for six simple examples of decision problems (some of which are actually promise problems), indicating the adequacy of the formulation of our AEQSs. In addition, basic structural properties of AEQSs will be briefly discussed.
Section \ref{sec:simulation} will present more general simulation processes of four types of finite automata on appropriate AEQSs. In contrast, we will show upper bounds of certain conditional AEQSs in Section \ref{sec:complexity-AEQS}.

This exposition is merely an initial attempt to expand the scope of adiabatic quantum computability of the past literature and to relate it to quantum finite automata through the fundamental framework of AEQSs.
Our new approach is likely to open a door to a further exploration of the ``practical'' use of adiabatic quantum computation under various natural restrictions imposed by, for instance, quantum finite automata.
We strongly expect our work to mark the beginning of a series of exciting research, aiming at the deeper understanding of adiabatic quantum computation in theory and in practice.

\section{Preparation: Notions and Notation}\label{sec:preparation}

We will provide fundamental notions and notation necessary to read through the subsequent sections.
Some notation slightly differs from the standard one but will prove itself to be more convenient for our arguments.

\subsection{Numbers, Vectors, and Matrices}\label{sec:numbers}

The notation $\nat$ expresses the set of all \emph{natural numbers} (that is, nonnegative integers) and we denote $\nat-\{0\}$ by $\nat^{+}$. Given two integers $m$ and $n$ with $m\leq n$, the \emph{integer interval} $[m,n]_{\integer}$ is the set $\{m,m+1,m+2,\ldots,n\}$, which is compared to a real interval $[a,b]$. For simplicity, we abbreviate $[1,n]_{\integer}$ as $[n]$ if $n\geq1$.
In contrast, $\complex$ denotes the set of all \emph{complex numbers}; in particular, we set $\imath = \sqrt{-1}$.
The \emph{complex conjugate} of a complex number $\alpha$ is expressed as $\alpha^*$.
Throughout this exposition, all \emph{polynomials} are assumed to have nonnegative integer coefficients and all \emph{logarithms} are taken to the base $2$. For convenience, we set $\log{0}=0$ and $\ilog{x} = \ceilings{\log{x}}$ for any $x\geq0$.
Given a finite set $S$, the notation $\PP(S)$ expresses the \emph{power set} of $S$, namely, the set of all subsets of $S$, and $|S|$ denotes the \emph{cardinality} of $S$.

For the sake of convenience, a function $f$ on $\nat$ (i.e., from $\nat$ to $\nat$) is said to be \emph{inverse-polynomially large} if $f(n)$ is at least the reciprocal of a certain polynomial $p$, namely, $f(n)\geq 1/p(n)$ for every number $n\in\nat$.
Similarly, $f$ is \emph{inverse-exponentially large} if there is a  polynomial $p$ satisfying $f(n)\geq 1/2^{p(n)}$ for all $n\in\nat$,  whereas $f$ is \emph{inverse-constantly large} if an appropriate constant $c\geq1$ ensures $f(n)\geq 1/c$ for any $n\in\nat$.
We assume the existence of an efficient bijection $\pair{\cdot,\cdot}$ from $\nat\times\nat$ to $\nat$ so that (i) we can easily encode $x$ and $y$ to $\pair{x,y}$ and (ii) we can easily decode $x$ and $y$ from $\pair{x,y}$. Such a function is known as a \emph{pairing function}. We can easily expand it to a bijection from $\nat^k$ to $\nat$ by setting $\pair{x_1,x_2,\ldots,x_k} = \pair{\pair{\pair{\pair{x_1,x_2},x_3},\ldots},x_k}$ for a fixed constant $k\geq2$.

In this exposition, we deal only with \emph{finite-dimensional Hilbert spaces}. To express (column) vectors of such a Hilbert space, we use Dirac's ``ket'' notation $\qubit{\cdot}$. The \emph{dual vector} of $\qubit{\phi}$ is denoted by $\bra{\phi}$ (the ``bra'' notation). A \emph{density operator} (or a \emph{density matrix}), which  is used to express a \emph{mixed quantum state}, is a positive operator whose trace equals $1$.
Given a real number $\varepsilon\in[0,1]$ and two vectors $\qubit{\phi}$ and $\qubit{\psi}$ in the same Hilbert space, we say that $\qubit{\phi}$ is \emph{$\varepsilon$-close to} $\qubit{\psi}$ if $\|\qubit{\phi}-\qubit{\psi}\|_2\leq\varepsilon$, where $\|\cdot\|_2$ indicates the $\ell_2$-norm, i.e., $\|\qubit{\phi}\|_2 = \sqrt{\measure{\phi}{\phi}}$.
For readability, we often abbreviate, e.g., the tensor product $\qubit{p}\otimes \qubit{q}$ of two basis vectors $\qubit{p}$ and $\qubit{q}$ as $\qubit{p}\qubit{q}$, $\qubit{p,q}$, or even $\qubit{pq}$.

The special notation $O$ denotes the all-zero square matrix of an arbitrary dimension and $I$ denotes the \emph{identity matrix} of an arbitrary dimension.  The \emph{commutator} $[A,B]$ of two square matrices $A$ and $B$ is defined as $AB-BA$. Given a complex matrix $A$, the notation $A^T$ denotes the \emph{transpose} of $A$, and $A^{\dagger}$ expresses the \emph{(Hermitian) adjoint} (i.e., the  complex conjugate transpose of $A$). A complex matrix $A$ is \emph{Hermitian} if $A$ is equal to its adjoint, namely, $A^{\dagger} = A$. For any  matrix $A$ and its index pair $(q,r)$, the notation $A[q,r]$ indicates the \emph{$(q,r)$-entry} of $A$. Similarly, for a vector $v$, $v[i]$ denotes the $i$th entry of $v$.

The notation $diag(a_1,a_2,\ldots,a_n)$ denotes an $n\times n$ matrix whose diagonal entries are $a_1,a_2,\ldots,a_n$ and the other entries are all $0$. The \emph{trace} $\trace(A)$ of an $n\times n$ matrix $A=(a_{ij})_{i,j\in[n]}$ is $\sum_{i=1}^{n}a_{ii}$.
Given any square complex matrix $A$, the notation $e^{A}$ expresses a \emph{matrix exponential}, defined by $e^A = \sum_{k=0}^{\infty} \frac{1}{k!}A^k$ (where $0!=1$ and $A^0 = I$) and  the \emph{spectral norm} $\|A\|$ is a matrix norm defined by $\|A\|= \max_{\qubit{\phi}\neq0} \{\frac{\|A\qubit{\phi}\|_2}{\|\qubit{\phi}\|_2}\}$.
A matrix $A$ is \emph{positive semidefinite} if  $\bra{\phi}A\ket{\phi}\geq0$ holds for any nonzero vector $\qubit{\phi}$. In this exposition, whenever we discuss eigenvectors of square matrices, we implicitly assume that all eigenvectors are \emph{normalized} (i.e., taken to have $\ell_2$-norm $1$).
Given two matrices $A$ and $B$ of the same dimension, we say that $A$ is \emph{approximated by $B$ to within $\varepsilon$} if $\|A-B\|\leq \varepsilon$.

A \emph{quantum bit} (or a \emph{qubit}) is a normalized quantum state in the $2$-dimensional Hilbert space, expressed as a linear combination of two designated basis vectors $\qubit{0}=(1\;\; 0)^T$ and $\qubit{1}=(0\;\; 1)^T$.
For those qubits $\qubit{0}$ and $\qubit{1}$, let $\qubit{\hat{0}}=\frac{1}{\sqrt{2}}(\qubit{0}+\qubit{1})$ and $\qubit{\hat{1}}=\frac{1}{\sqrt{2}}(\qubit{0}-\qubit{1})$.
The sets $\{\qubit{\hat{0}},\qubit{\hat{1}}\}$  and $\{\qubit{0},\qubit{1}\}$ are respectively called the \emph{Hadamard basis} and the \emph{computational basis}.
We use the notation $W$ for the \emph{Walsh-Hadamard transform} $\frac{1}{\sqrt{2}} \matrices{1}{1}{1}{-1}$. Notice that $W$ transforms the computational basis to the Hadamard basis.
For any fixed number $N\in\nat^{+}$, the \emph{quantum Fourier transform} $F_N$ changes $\qubit{j}$ to $\frac{1}{\sqrt{N}}\sum_{k=0}^{N-1}\omega_N^{jk}\qubit{k}$, where $j\in[0,N-1]_{\integer}$ and $\omega_N= e^{\imath\frac{2\pi}{N}}$.
In particular, $F_2$ coincides with $W$.

A \emph{Hamiltonian} is a complex Hermitian matrix. In this exposition,  eigenvectors are called \emph{eigenstates}. For any  Hamiltonian $H$, we set $\Delta(H)$ to be the \emph{spectral gap} of $H$, which is the difference between the lowest eigenvalue and the second lowest
eigenvalue of $H$. An eigenvalue of $H$ is also called an \emph{energy} of $H$. The lowest eigenvalue is particularly called the \emph{ground energy} of $H$ and its associated eigenstate is called the \emph{ground state} of $H$.

\subsection{Languages and Quantum Finite Automata}\label{sec:QFA}

An \emph{alphabet} is a finite nonempty set of ``symbols'' or ``letters''.  A \emph{string} over an alphabet $\Sigma$ is a finite sequence of symbols in $\Sigma$. The \emph{length} of a string $x$ is the total number of symbols appearing in $x$ and is denoted by $|x|$. In particular, the \emph{empty string}, denoted by $\lambda$, is a unique string of length $0$.
The notation $\Sigma^*$ stands for the set of all strings over $\Sigma$ and $\Sigma^{+}$ expresses $\Sigma^*-\{\lambda\}$.
Given a string $x$ and a symbol $a$, $\#_a(x)$ expresses the total number of $a$ in $x$.
In addition, for any number $i\in [|x|]$, $x_{(i)}$ denotes the $i$th symbol of $x$.
For any number $n\in\nat$ and any index $i\in[0,\bar{n}]_{\integer}$ with $\bar{n} = 2^{{\ilog{n}}}$ (i.e., $2$ to the power of  ${\ilog{n}}$), the notation $s_{n,i}$ denotes the lexicographically $i$th string of length ${\ilog{n}}$ over the binary alphabet $\{0,1\}$; in particular, $s_{n,0}=0^{{\ilog{n}}}$ and $s_{n,\bar{n}} = 1^{{\ilog{n}}}$. Each string $s_{n,i}$ can be viewed as the binary  encoding of number $i$ expressed by exactly ${\ilog{n}}$ bits. The aforementioned notion $\pair{\cdot,\cdot}$ for the pairing function is naturally adapted to a bijection from $\Sigma^*\times\Sigma^*$ to $\Sigma^*$.

A \emph{language} over $\Sigma$ is a subset of $\Sigma^*$. Hereafter, we freely identify a \emph{decision problem} with its associated language.
Given a number $n\in\nat$, $\Sigma^n$ (resp., $\Sigma^{\leq n}$) denotes the set of all strings of length exactly $n$ (resp., at most $n$) over $\Sigma$. Notice that $\Sigma^{\leq n}= \bigcup_{k\leq n} \Sigma^k$ and $\Sigma^*=\bigcup_{n\geq0} \Sigma^n$.
The set $\Sigma^*-L$ is the \emph{complement} of $L$ and it is often written as $\overline{L}$ as long as $\Sigma$ is clear from the context. Given a string $x=x_1x_2\cdots x_{n-1}x_n$ with $x_i\in\Sigma$ for any $i\in[n]$, $x^R$ denotes the \emph{reversal} of $x$, i.e., $x^R = x_nx_{n-1}\cdots x_2x_1$.
If $L$ is a language over $\Sigma$, then we also use the same symbol $L$ to denote its \emph{characteristic function}; that is, $L(x)=1$ for any $x\in L$ and $L(x)=0$ for any $x\in \Sigma^*-L$.
A function $f$ on $\Sigma^*$ (i.e., from $\Sigma^*$ to $\Sigma^*$) is \emph{length preserving} if $|f(x)|=|x|$ holds for any string $x\in\Sigma^*$, and $f$ is \emph{polynomially bounded} if there exists a polynomial $p$ satisfying  $|f(x)|\leq p(|x|)$  for any string $x\in\Sigma^*$. This last notion is easily expanded to functions $f$ from $\Sigma^*$ to $\nat$ by requiring $f(x)\leq p(|x|)$ in place of $|f(x)|\leq p(|x|)$.
For any two languages $A,B\subseteq \Sigma^*$, the notation $AB$ denotes the language $\{xy\mid x\in A,y\in B\}$.
Given $k$ symbols $y_1,y_2,\cdots, y_k$ in $\Sigma$, we customarily abbreviate a sequential multiplication $A_{y_k}\cdot A_{y_{k-1}} \cdots A_{y_2}\cdot A_{y_1}$
of square matrices as $A_{y_1y_2\cdots y_{k-1}y_k}$.

A \emph{promise problem} is a pair $(A,B)$ of disjoint sets over alphabet $\Sigma$ such that $A$ and $B$ consist of \emph{accepting instances} and \emph{rejecting instances}, respectively, and all instances $x$ given to this problem  are always ``promised'' to be in $A\cup B$. The promise problem $(A,B)$ asks to determine whether $x\in A$ or $x\in B$ for any promised instance $x$. In particular, when $B=\Sigma^*-A$, $(A,B)$ coincides with $(A,\overline{A})$, which is customarily identified with the language $A$.

We assume the reader's familiarity with one-way deterministic finite(-state) automata (or 1dfa's), one-way deterministic pushdown automata (or 1dpda's), and one-way nondeterministic pushdown automata (or 1npda's). Refer to, e.g., \cite{HU79}.
The \emph{state complexity} of a finite automaton is the total number of the inner states of the finite automaton.
In what follows, we explain the models of one-way and two-way quantum finite automata because they are the core notions of this exposition.

A \emph{two-way quantum finite automaton} (or a 2qfa, for short) is expressed as a septuple $(Q,\Sigma,\{\cent,\dollar\}, \delta,q_0,Q_{acc},Q_{rej})$, where $Q$ is a finite set of inner states, $\Sigma$ is an (input) alphabet, $\cent$ and $\dollar$ are respectively the \emph{left-endmarker} and the \emph{right-endmarker}, $\delta$ is a (quantum) transition function from $Q\times\check{\Sigma}\times Q\times D$ to $\complex$ with $\check{\Sigma}=\Sigma\cup\{\cent,\dollar\}$ and $D=\{-1,0,+1\}$, $q_0$ is the initial state in $Q$, and $Q_{acc}$ and $Q_{rej}$ are subsets of $Q$, consisting of \emph{accepting states} and \emph{rejecting states}, respectively. The values $-1$, $0$, and $+1$ in $D$ respectively indicate ``to the left,'' ``staying still,'' and ``to the right.''
For convenience, we write $Q^{(-)}$ for $Q-\{q_0\}$ and $Q_{halt}$ for $Q_{acc}\cup Q_{rej}$.
When the direction set $D$ is restricted to $\{0,+1\}$, $M$ is particularly called a \emph{1.5-way quantum finite automaton} (or a 1.5qfa).
A read-only input tape is indexed by natural numbers from left to right.
An input string $x=x_1x_2\cdots x_n$ of length $n$ is written on this tape, surrounded by $\cent$ and $\dollar$, so that $\cent$ is in cell $0$, each $x_i$ is in cell $i$, and $\dollar$ is in cell $n+1$. For convenience, we set $x_0=\cent$ and $x_{n+1}=\dollar$ and we call $x_0x_1\cdots x_nx_{n+1}$ ($=\cent{x}\dollar$) an \emph{extended input}.
As customary, the input tape is assumed to be \emph{circular}; that is, the both ends of the tape are glued together so that the right of $\dollar$ is $\cent$ and the left of $\cent$ is $\dollar$.

A \emph{configuration} of a 2qfa $M$ is a triplet $(q,x,i)\in Q\times\Sigma^*\times\nat$, which expresses a circumstance that $M$ is in inner state $q$, scanning the $i$th location of an extended input $\cent{x}\dollar$, whereas a \emph{surface configuration} is a pair $(q,i)$, excluding the input $x$.
The \emph{configuration space} is the Hilbert space spanned by $\{\qubit{q,x,i}\mid (q,x,i)\;\; \text{is a configuration of $M$}\:\}$; in contrast, the \emph{surface configuration space} is spanned by the vectors $\qubit{q,i}$ with $(q,i)\in Q\times [0,|x|+1]_{\integer}$.
The complex value $\delta(q,x_i,p,d)$ is called a \emph{transition amplitude}.
When we apply a transition $\delta(q,x_i,p,d)=\alpha$ to a configuration $(q,x,i)$, in a single step, the configuration is changed to $(p,x,i+d\:(\mathrm{mod}\:|x|+2))$ with transition amplitude $\alpha$.
The transition function $\delta$ induces a linear operator (called a \emph{time-evolution operator}) $U^{(x)}_{\delta}$ defined by $U^{(x)}_{\delta}\qubit{q,i} = \sum_{(p,d)\in Q\times D} \delta(q,x_i,p,d) \qubit{p,i+d\:\mathrm{mod}\:(|x|+2)}$ for every surface configuration $(q,i)\in Q\times [0,|x|+1]_{\integer}$. Occasionally, we restrict these transition amplitudes on a certain nonempty subset of $\complex$.
Finally, we always demand that $U^{(x)}_{\delta}$ is \emph{unitary} (i.e., $U^{(x)}_{\sigma}(U^{(x)}_{\sigma})^{\dagger}=I$) for every input string $x\in\Sigma^*$.

The special operators $\Pi_{acc}$, $\Pi_{rej}$, and $\Pi_{non}$ are \emph{projective measurements} onto the Hilbert spaces spanned by $\{\qubit{q,i}\mid q\in Q_{acc},i\in[0,n+1]_{\integer}\}$, by $\{\qubit{q,i}\mid q\in Q_{rej},i\in[0,n+1]_{\integer}\}$, and by $\{\qubit{q,i}\mid q\in Q-Q_{halt},i\in[0,n+1]_{\integer}\}$, respectively. We define $\qubit{\phi_0}=\qubit{q_0,0}$ and $\qubit{\phi_{i+1}} = \Pi_{non}U^{(n)}_{\delta}\qubit{\phi_i}$ for every index $i\in\nat$.
We say that $M$ \emph{accepts} (resp., \emph{rejects}) $x$ with probability $\gamma$ if $\sum_{i\in\nat}\| \Pi_{acc}\qubit{\phi_i} \|^2 =\gamma$ (resp., $\sum_{i\in\nat}\| \Pi_{rej}\qubit{\phi_i} \|^2 =\gamma$).

The computational power of 2qfa's can be enhanced if we equip them with \emph{(flexible) garbage tapes}, which are write-once\footnote{A tape is \emph{write once} if its tape head never moves to the left and it moves to the right whenever it writes down a non-blank symbol.} tapes \cite{Yam19b}. We call them \emph{2qfa's with garbage tapes} (or garbage-tape 2qfa's) and each of them has the form $(Q,\Sigma,\{\cent,\dollar\},\Xi,\delta,q_0,Q_{acc},Q_{rej})$ with a \emph{garbage alphabet} $\Xi$. The (quantum) transition function $\delta$ is now a map from $Q\times\check{\Sigma}\times Q\times D\times \Xi_{\lambda}$ to $\complex$,
where $\Xi_{\lambda} = \Xi\cup\{\lambda\}$.
A 2qfa writes down a garbage symbol $\xi$ in $\Xi_{\lambda}$ and, whenever $\xi\neq\lambda$, it moves the tape head to the right.
With the good use of garbage tapes, we can postpone any intermediate measurement until the end of a computation
if the computation terminates.

Let $L$ denote any language over alphabet $\Sigma$ and let $\varepsilon$ be any error bound in $[0,1/2]$. We say that $M$ \emph{recognizes} $L$ with probability at least $1-\varepsilon$ if (i) for any $x\in L$, $M$ accepts $x$ with probability at least $1-\varepsilon$ and (ii) for any $x\in \Sigma^*-L$, $M$ \emph{rejects} $x$ with probability at least $1-\varepsilon$, except that, whenever both the accepting probability and the rejecting probability are exactly $1/2$, $M$ is considered to \emph{reject} the input. In the case where $\varepsilon$ falls into the range  $[0,1/2)$, $M$ is particularly said to recognize $L$ with \emph{bounded error probability}.

A \emph{one-way quantum finite automaton} (abbreviated as a 1qfa)  \emph{with quantum operations}\footnote{This model is called \emph{general quantum finite automata} in a survey \cite{AY15}.}  is a septuple $(Q,\Sigma,\{\cent,\dollar\}, \{A_{\sigma}\}_{\sigma\in\check{\Sigma}}, q_0,Q_{acc},Q_{rej})$, where each $A_{\sigma}$ is a \emph{quantum operation}\footnote{This is a completely positive, trace preserving map and is also called a \emph{superoperator}.} acting on the Hilbert space of linear operators on the \emph{configuration space} spanned by the vectors in  $\{\qubit{q}\mid q\in Q\}$ \cite{ABG+06,FOM09,YS11}.
Such a quantum operation $A_{\sigma}$ has a \emph{Kraus representation}  $\KK_{\sigma} = \{K_{\sigma,j}\}_{j\in[k]}$, composed of $k$ \emph{Kraus operators} (or \emph{operation elements}) $K_{\sigma,j}$ for a certain constant $k\in\nat^{+}$, provided that all entries of each Kraus operator $K_{\sigma,j}$ are indexed by the elements of $Q\times Q$.
More precisely, $A_{\sigma}$ has the form $A_{\sigma}(H)= \sum_{j=1}^{k}K_{\sigma,j} H (K_{\sigma,j})^{\dagger}$ for any given linear operator $H$. In this exposition, we always demand that $\{K_{\sigma,j}\}_{j\in[k]}$ satisfies the \emph{completeness relation}
$\sum_{j=1}^{k}(K_{\sigma,j})^{\dagger} K_{\sigma,j}= I$. This implies that $A_{\sigma}$ is \emph{trace preserving}, that is, $\trace{A_{\sigma}(H)} = \trace{H}$.
We further expand the notation $A_{\sigma}$ to $A_{\cent x\dollar}$ for every string $x\in\Sigma^*$ by defining $A_{z\sigma}(H) = A_{\sigma}(A_{z}(H))$ inductively for any $\sigma\in\Sigma\cup\{\dollar\}$ and any $z\in\{\cent,\lambda\}\Sigma^*-\{\lambda\}$.
Note that, if $H$ is a Hermitian operator, then
so is $A^{(n)}_{\cent x \dollar}(H)$.
Given a language $L$ over $\Sigma$ and a constant $\varepsilon\in[0,1]$, we say that $M$ \emph{recognizes} $L$ with error probability at most $\varepsilon$ if (i) for any $x\in L$, $\trace(\Pi_{acc}A_{\cent x\dollar}(\rho_0)) \geq 1-\varepsilon$ and (ii) for any $x\in \Sigma^*-L$, $\trace(\Pi_{rej}A_{\cent x\dollar}(\rho_0)) \geq 1-\varepsilon$, where $\rho_0=\density{q_0}{q_0}$, and $\Pi_{acc}$ and $\Pi_{rej}$ are projections similar to those of 2qfa's.
This model is computationally equivalent to a \emph{garbage-tape model} of 1qfa's used in \cite{Yam19b} (implicitly in \cite{Yam14b}).
In the case of $k=1$, since $\KK_{\sigma}$ is a singleton $\{K_{\sigma,1}\}$, we briefly write $K_{\sigma}$ for $K_{\sigma,1}$, and thus we obtain $A_{\sigma}(H) = K_{\sigma}H (K_{\sigma})^{\dagger}$. This case precisely coincides with a \emph{one-way measure-once quantum finite automaton} (or a 1moqfa) of Moore and Crutchfield \cite{MC00}.
For later use, we write  $\mathrm{1MOQFA}$  to denote the collection of all languages recognized by bounded-error 1moqfa's.

\subsection{Hamiltonians and Quantum Quasi-Automata}\label{sec:quasi-automata}

Since the construction of Hamiltonians is a key to producing adiabatic quantum computations, this exposition attempts to operate finite automata to ``generate'' (or ``produce'') such Hermitian matrices. For this purpose, we first modify the model of 1qfa's explained in Section \ref{sec:QFA} so that they can ``produce'' square matrices instead of ``recognizing'' languages. Such modified 1qfa's are succinctly called \emph{quantum quasi-automata} in this exposition.
In short, a family of quantum quasi-automata generates a series of complex square matrices, in particular, Hamiltonians. This exposition aims at introducing three different types of quantum quasi-automata.
Notice that a family of such quantum quasi-automata $M_n$ is inherently \emph{nonuniform}; that is, no  algorithm  is required to exist for  producing (an encoding of) $M_n$ from each unary input $1^n$.

A \emph{one-way quantum quasi-automata family} (abbreviated as 1qqaf) is a family of 1qfa's using quantum operations with \emph{no use} of initial state and halting state.
Formally, a 1qqaf $\MM=\{M_n\}_{n\in\nat}$ is a family of machines $M_n$ of the form $(Q^{(n)},\Sigma, \{\cent,\dollar\}, \{A^{(n)}_{\sigma}\}_{\sigma\in\check{\Sigma}}, \Lambda^{(n)}_0, Q^{(n)}_{0})$,
where $Q^{(n)}$ is a finite set of inner states, $A^{(n)}_{\sigma}$ is a quantum operation on the Hilbert space of linear operators acting on the configuration space spanned by the vectors in $\{\qubit{q}\mid q\in Q^{(n)}\}$, $\Lambda^{(n)}_0$ is an \emph{initial mixture} of the form $\sum_{u\in Q^{(n)}}\gamma_u\density{u}{u}$ with $\gamma_u\geq0$ (which is a positive-semidefinite Hermitian operator) acting on the same space (not necessarily limited to $\density{q_0}{q_0}$ as for 1qfa's), and $Q_0^{(n)}$ is a subset of $Q^{(n)}$.
Associated with this last set $Q^{(n)}_0$, we define the \emph{projective measurement} $\Pi^{(n)}_0$ as the projection onto the Hilbert space spanned by the vectors in $\{\qubit{u}\mid u\in Q^{(n)}-Q^{(n)}_0\}$.
Whenever $Q^{(n)}_0=\setempty$, we often omit $Q^{(n)}_0$ from the above definition of $M_n$ since $\Pi^{(n)}_0$ coincides with $I$.
The number $|Q^{(n)}|$ of inner states may vary according to $n$ and this provides the ``size'' of $\MM$.
Although 1qfa's and 1qqaf's are similar in their forms, 1qqafa's have
no acceptance/rejection criteria.
Each 1qqaf acts as a means to produce a series of \emph{Kraus operators},   according to each input symbol, which represent quantum operations.
Notice that the quantum operation $A^{(n)}_{\cent x\dollar}$ takes linear operators acting on the $|Q^{(n)}|$-dimensional Hilbert space.
The machine $M_n$ is said to \emph{generate} (or \emph{produce}) a quantum  operator $E^{(n)}_{\cent{x}\dollar} = \Pi^{(n)}_0A^{(n)}_{\cent x\dollar}(\Lambda^{(n)}_0)\Pi^{(n)}_0$
for any instance $x$. We remark that $E^{(n)}_{\cent{x}\dollar}$ is  positive semidefinite.

Our goal is to generate a family $\{H^{(x)}\}_{x\in\Sigma^*}$ of Hamiltonians. However, unlike Boolean circuits, every machine $M_n$ of $\MM$ takes all strings in $\Sigma^*$ as its inputs.
This fact forces us to specify which machine $M_n$ to run in order to generate the target Hamiltonian operator $H^{(x)}$ for each input $x$.
The necessity of such specification of $n$ for $x$ has been naturally observed in discussing the computability of a nonuniform family of quantum finite automata \cite{Yam18,Yam19b}.
For this purpose, we introduce a notion of \emph{selector} $\mu:\Sigma^*\to\nat$ that bridges between $x$ and $n$ as $\mu(x)=n$. Such a selector $\mu$ induces the \emph{input domain} $\Delta_n = \{x\in \Sigma^*\mid \mu(x)=n\}$ of $M_n$ for each index $n\in\nat$.
The family $\{\Delta_n\}_{n\in\nat}$ of such input domains satisfies that (i)  $\Delta_n\cap \Delta_m=\setempty$ for any two distinct pair $m,n\in\nat$ and (ii) $\bigcup_{n\in\nat}\Delta_n = \Sigma^*$.

The above observation makes it possible to generate  $H^{(x)}$ simply by running an appropriately chosen machine $M_{\mu(x)}$ on input $x$.
We formally say that the family $\{H^{(x)}\}_{x\in\Sigma^*}$ of Hamiltonians is \emph{generated by 1qqaf's} if there exist a polynomially-bounded selector $\mu:\Sigma^*\to\nat$ and a 1qqaf $\MM$ such that,  for every string $x\in\Sigma^*$, $E^{(\mu(x))}_{\cent x \dollar}$
coincides with $H^{(x)}$. Whenever $\mu$ is clear from the context, nonetheless, we often omit any reference to $\mu$. Here, we do not require the computability of $\mu$; however, later in Section \ref{sec:complexity-AEQS}, $\mu$ will be further assumed to be \emph{logarithmic-space computable}.

A more important family of machines in this exposition is, in fact, a \emph{one-way measure-once quantum quasi-automata family} (or a 1moqqaf) $\MM$, which is a special case of a 1qqaf satisfying that its Kraus representation $\KK_{n,\sigma} = \{K^{(n)}_{\sigma,i}\}_{i\in[k]}$ for each $A^{(n)}_{\sigma}$ as in the above definition must have $k=1$. In this case, as noted in Section \ref{sec:QFA}, we write $K^{(n)}_{\sigma}$ in place of $K^{(n)}_{\sigma,1}$.

Next, we wish to expand 1qqaf's to their corresponding two-way machine families. For a practical reason, we wish to curtail the runtime of such two-way machines by setting up an appropriate time-bounding function $t:\Sigma^*\to\nat$. A \emph{$t$-time two-way quantum quasi-automata family} (or a $t$-time 2qqaf) $\MM$ consists of machines $M_n=(Q^{(n)},\Sigma, \{\cent,\dollar\}, \Xi, \delta^{(n)}, \{K^{(n)}_{\cent,i}\}_{i\in\Xi}, \Lambda^{(n)}_0, Q^{(n)}_0)$ for each index $n\in\nat$, where $\Lambda^{(n)}_0$ is an initial mixture and $\delta^{(n)}$ is a (quantum) transition function from $Q^{(n)}\times \check{\Sigma}\times Q^{(n)}\times D\times \Xi$ to $\complex$.
We identify $\Xi$ with $[|\Xi|]$ and assume for simplicity that $\Xi =\{1,2,\ldots,k\}$.

In what follows, let us fix  $x\in\Sigma^*$ arbitrarily and set $n$ to be its associated value $\mu(x)$. The corresponding \emph{surface configuration space} of $M_n$ on input $x$ is spanned by the vectors in $\{\qubit{q,i}\mid q\in Q^{(n)},i\in[0,|x|+1]_{\integer}\}$, and thus it has dimension $|Q^{(n)}|(|x|+2)$.
The initial step of $M_n$ is carried out by a single application of the set of Kraus operators $K^{(n)}_{\cent,i}$. This first move produces a new Hermitian matrix $\tilde{\Lambda}^{(n)}_0$ defined by $\sum_{i\in[k]} K^{(n)}_{\cent,i} \Lambda^{(n)}_0 (K^{(n)}_{\cent,i})^{\dagger}$.
In the case of $k=1$, however, since $\tilde{\Lambda}^{(n)}_0$ can be expressed as $\sum_{u\in Q^{n}\times[0,n]_{\integer}} \tilde{\gamma}_u \density{u}{u}$ for certain values $\tilde{\gamma}_u\geq0$, by setting $\tilde{\Lambda}^{(n)}_0$ as a new initial mixture, we can omit $\{K^{(n)}_{\cent,i}\}_{i\in[k]}$ from the definition of $M_n$.
Furthermore, we introduce a quantum operation $A^{(n,x)}$, which takes linear operators acting on this surface configuration space by defining a Kraus representation $\KK_{n,x} = \{K^{(n,x)}_{j}\}_{j\in\Xi}$ as follows:   each Kraus operator $K^{(n,x)}_j$ for $j\in\Xi$ expresses $\delta^{(n)}$ as $\bra{p,i+d} K^{(n,x)}_{j} \ket{q,i} = \delta^{(n)}(q,x_{(i)},p,d,j)$ for any $p,q\in Q^{(n)}$, $i\in[0,|x|+1]_{\integer}$, and $d\in D$, where $x_{(0)}=\cent$ and $x_{(|x|+1)}=\dollar$.
More precisely, $A^{(n,x)}(H)$ is set to be $\sum_{j\in\Xi} K^{(n,x)}_{j}H (K^{(n,x)}_j)^{\dagger}$ for any linear operator $H$.
To express the repetitive use of an operation $A^{(n,x)}$, we write $(A^{(n,x)})^k(H)$ to mean the result obtained by the $k$ consecutive applications of $A^{(n,x)}(\cdot)$ to $H$.
In this case, the $t$-time 2qqaf $\MM$ is said to \emph{generate} (or \emph{produce}) the matrix  $E^{(n,x)} = \Pi^{(n)}_0(A^{(n,x)})^{t(n,|x|)}(\tilde{\Lambda}^{(n)}_0)\Pi^{(n)}_0$.
The dimension of this matrix $E^{(n,x)}$ is $|Q^{(n)}\times [0,|x|+1]_{\integer}| = |Q^{(n)}|(|x|+2)$.
For each index $j\in\Xi$, we define $K^{(n)}_j = \sum_{x\in\Sigma^*}(\density{x}{x}\otimes K^{(n,x)}_j)$ and write $\KK_n$ for $\{K^{(n)}_j\}_{j\in\Xi}$. Since $\delta^{(n)}$ precisely induces $\KK_n$, we occasionally express $M_n$ as $(Q^{(n)}, \Sigma, \{\cent,\dollar\}, \Xi, \KK_n, \{K^{(n)}_{\cent,i}\}_{i\in[k]}, \Lambda^{(n)}_0, Q^{(n)}_0)$ by including $\KK_n$.
Formally, a family $\{H^{(x)}\}_{x\in\Sigma^*}$ of Hamiltonians is said to be \emph{generated by $t$-time 2qqaf's} if, for an appropriate  polynomially-bounded selector $\mu:\Sigma^*\to\nat$, a certain  $t$-time 2qqaf produces, for every input $x\in\Sigma^*$, a matrix $E^{(\mu(x),x)}$ that coincides with $H^{(x)}$.
Furthermore, when $t(x)= O(|x|)$ and $t(x)=|x|^{O(1)}$, we conveniently call $t$-time 2qqaf's by \emph{linear-time 2qqaf's} and \emph{polynomial-time 2qqaf's}, respectively. As a special case of $t$-time 2qqaf's, we also introduce \emph{linear-time 1.5qqaf's}  by setting $t(x)$ to be $O(|x|)$ as well as replacing the direction set $D$ by $\{0,+1\}$.

\section{Adiabatic Evolutionary Quantum Systems}\label{sec:AEQS-def}

We formally introduce an adiabatic model of AEQS and present how to place practical restrictions on this model in order to use it as a technical tool in  classifying various formal languages according to their  complexities.

\subsection{Adiabatic Evolution of a Quantum System}\label{sec:system-evolution}

Loosely following \cite{FGGS00}, we briefly discuss how a quantum system evolves according to the Schr\"{o}dinger equation of the following general form:  $\imath\hbar \frac{d}{dt}\qubit{\psi(t)} = H(t)\qubit{\psi(t)}$ for a time-dependent Hamiltonian $H(t)$ and a time-dependent quantum state $\qubit{\psi(t)}$.
To carry out adiabatic quantum computation on this quantum system, we prepare two Hamiltonians $H_{ini}$ and $H_{fin}$ acting on the same Hilbert space and, for a sufficiently large constant $T>0$, we define $H(t) = \left(1-\frac{t}{T}\right)H_{ini}+\frac{t}{T}H_{fin}$ using
a time parameter $t\in[0,T]$, provided that $H_{ini}$ and $H_{fin}$ do not commute; that is, $[H_{ini},H_{fin}]\neq O$. Notice that the condition $[H_{ini},H_{fin}]= O$ implies the existence of simultaneous eigenstates, causing the below-mentioned minimal evolution time to approach infinity.
To ensure $[H_{ini},H_{fin}]\neq O$, nonetheless, we often use the \emph{Hadamard basis} for $H_{ini}$ and the \emph{computational basis} for $H_{fin}$ \cite{DMV01,FGG+01}. Furthermore, we require $H_{ini}$ as well as $H_{fin}$ to have a \emph{unique} ground state.

At time $t=0$, we assume that the quantum system is initialized to be the ground state $\qubit{\psi_g(0)}$ of $H_{ini}$, namely,
$\qubit{\psi(0)} = \qubit{\psi_g(0)}$.
We allow the system to gradually evolve by applying $H(t)$ \emph{discretely} from time $t=0$ to $t=T$. Let $\qubit{\psi(t)}$ denote the quantum state at time $t\in[0,T]$.
This quantum state $\qubit{\psi(t)}$ is known to approach slowly to the ground state of $H_{fin}$.
This evolutionary process is referred to as an \emph{adiabatic evolution according to $H(t)$ for $T$ steps}.
We take the smallest value $T$ for which $\qubit{\psi(T)}$ is $\varepsilon$-close to the ground state of $H_{fin}$ and
this particular value $T$ is called the \emph{minimum evolution time}
of the system.
The \emph{runtime} of the system, however, is defined to be $T\cdot \max_{t\in[0,T]}\|H(t)\|$ and the \emph{outcome} of the system is the quantum state $\qubit{\psi(T)}$.
The \emph{adiabatic theorem} \cite{Kat51,Mes58} adequately provides a lower bound on $T$. The following form of the theorem is taken from \cite{ADK+07}: for any two constants $\varepsilon,\delta>0$, if $T\geq \Omega\left( \frac{\|H_{fin}-H_{ini}\|^{1+\delta}} {\varepsilon^{\delta}\min_{t\in[0,T]}  \{\Delta(H(t))^{2+\delta}\}} \right)$, then $\qubit{\psi(T)}$ (with an appropriately chosen \emph{global phase}) is $\varepsilon$-close to the ground state $\qubit{\psi_g(T)}$ of $H_{fin}$, provided that $H(t)$ has a unique ground state for each value $t\in[0,T]$.

The adiabatic evolution can be described by an appropriate unitary matrix $U_T$ satisfying $\qubit{\psi(T)} = U_T\qubit{\psi(0)}$.
We want to approximate  $U_T$ as follows.
Firstly, we make a good refinement of the time intervals. Let $R$ denote a fixed integer satisfying $T\ll R$ and consider refined time intervals  $[\frac{jT}{R},\frac{(j+1)T}{R}]$ for all indices  $j\in[0,R-1]_{\integer}$.
For convenience, let $\real^{\geq0}$ denote the set $\{r\in\real\mid r\geq0\}$.

\begin{lemma}\label{time-evolution}
Consider a quantum system of adiabatic evolution with $H_{ini}$ and  $H_{fin}$. Assume that $H_{ini}$ and $H_{fin}$ are of dimension $2^n$ and that $\max\{\|H_{ini}\|,\|H_{fin}\|\}\leq \nu(n)2^{n}$ for a certain function  $\nu:\nat\to\real^{\geq0}$. Let $T$ be the minimum evolution time and let $U_T$ denote a unitary matrix satisfying $\qubit{\psi(T)} = U_{T}\qubit{\psi(0)}$.
Let $R$ denote an integer with  $T\ll R$ and, for each index $j\in[0,R-1]_{\integer}$, let $\alpha_j= \frac{1}{\hbar}\frac{T}{R}\left(1-\frac{2j+1}{2R}\right)$, $\beta_j= \frac{1}{\hbar}\frac{T}{R}\frac{2j+1}{2R}$, and  $V(j)=  e^{-\imath \alpha_j H^{(x)}_{ini}}\cdot e^{-\imath \beta_j H^{(x)}_{fin}}$.
Denote by $V_R$ the sequential multiplication $V(R) V(R-1)\cdots V(2) V(1)$.
It follows that $U_{T}$  can be approximated by the matrix  $V_R$
to within $O(\frac{2^{2n}T^{2}\nu^2(n)}{R})$.
\end{lemma}

If we take $R$ to satisfy $R\geq 2^{2n}T^3\nu^2(n)$, then we obtain $\|U_T-V_R\|=O(\frac{1}{T})$. In such a case, since $V_R$ is ``close'' enough to $U_T$, we can use it in place of $U_T$ for our later argument in Section \ref{sec:complexity-AEQS}.

\begin{proofof}{Lemma \ref{time-evolution}}
Our argument that follows below refines the proof given in \cite[Section 4]{DMV01}.
Assume that $H_{ini}$, $H_{fin}$, $T$, $R$, and $U_T$ are given as in the premise of the lemma.
We fix a starting time $t_0\in[0,T)$ and consider the time interval $[t_0,t]$ for an arbitrary time $t>t_0$. Recall that $H(t)$ equals $\left(1-\frac{t}{T}\right)H_{ini}+\frac{t}{T}H_{fin}$.
Since $(t-t_0)H(\frac{t+t_0}{2}) = (t-t_0)H_{ini}+\frac{t^2-t_0^2}{2T}(H_{fin}-H_{ini})$, it follows that
$\frac{d}{dt}(t-t_0)H(\frac{t+t_0}{2}) = H_{ini}+\frac{t}{T}(H_{fin}-H_{ini}) = H(t)$.
Here, we claim that the solution of the Schr\"{o}dinger equation   $\imath\hbar \frac{d}{dt}\qubit{\psi(t)} = H(t)\qubit{\psi(t)}$ is given by $\qubit{\psi(t)} = U(t,t_0) \qubit{\psi(t_0)}$, where $U(t,t_0) = e^{-\frac{\imath}{\hbar} (t-t_0) H(\frac{t+t_0}{2})}$. To see this fact, by differentiating $\qubit{\psi(t)}$, we obtain
\[
\frac{d}{dt}\qubit{\psi(t)} = \frac{d}{dt}U(t,t_0)\qubit{\psi(t_0)} =
-\frac{\imath}{\hbar}H(t) e^{-\frac{\imath}{\hbar}(t-t_0) H(\frac{t+t_0}{2})} \qubit{\psi(t_0)} = -\frac{\imath}{\hbar}H(t)\qubit{\psi(t)},
\]
which is obviously equal to the aforementioned Schr\"{o}dinger equation.

Fix an index $j\in[0,R-1]_{\integer}$ arbitrarily and consider the refined interval $[\frac{jT}{R},\frac{(j+1)T}{R}]$.
We conveniently write $\qubit{\phi(j)}$ for the quantum state $\qubit{\psi(t)}$ at time $t=\frac{jT}{R}$.
We also take a unitary matrix $U'(j+1,j)$ satisfying  $\qubit{\phi(j+1)} = U'(j+1,j) \qubit{\phi(j)}$; in other words, $\qubit{\phi(j)}$ evolves to $\qubit{\phi(j+1)}$ by applying $U'(j+1,j)$.
Let $\alpha_j= \frac{1}{\hbar}\frac{T}{R}\left(1-\frac{2j+1}{2R}\right)$ and $\beta_j= \frac{1}{\hbar}\frac{T}{R}\frac{2j+1}{2R}$.
Similar to $U(t,t_0)$, the matrix $U'(j+1,j)$ can be written as
\[
U'(j+1,j) = e^{-\frac{\imath}{\hbar}\frac{T}{R} H(\frac{(2j+1)T}{2R})} = e^{-\frac{\imath}{\hbar}\frac{T}{R}(1-\frac{2j+1}{2R})H_{ini} -  \frac{\imath}{\hbar}\frac{T}{R}\frac{2j+1}{2R}H_{fin}} = e^{-\imath\alpha_j H_{ini} -\imath \beta_j H_{fin}}.
\]

It follows by the Baker-Campbell-Hausdorff Theorem that $e^{-\imath \alpha_j H_{ini} - \imath \beta_j H_{fin}}$ is approximated
by $e^{-\imath \alpha_j H_{ini}} \cdot e^{-\imath \beta_j H_{fin}}$  to within $|\alpha_j||\beta_j| \cdot O(\|H_{ini}\|\|H_{fin}\|) = O(\frac{T^2}{R^2}\|H_{ini}\|\|H_{fin}\|)$ (cited in \cite{DMV01}).
Since $\|H_{ini}\|\|H_{fin}\|=O(\nu^2(n)2^{2n})$, letting $V(j) = e^{-\imath \alpha_j H_{ini}}\cdot e^{-\imath \beta_j H_{fin}}$, the matrix $U'(j+1,j)$ is approximated by $V(j)$ to within $O(\frac{T^2\nu^2(n)2^{2n}}{R^2})$.
Since $\qubit{\psi(0)} = \qubit{\phi(0)}$ and $\qubit{\psi(T)} = \qubit{\phi(R)}$, $U_T$ coincides with the matrix $U'(R,R-1)\cdots U'(2,1) U'(1,0)$. The approximability of $U'(j+1,j)$ by $V(j)$ concludes that $U_T$ is approximated by $V_{R} = V(R)V(R-1) \cdots V(2) V(1)$ to within $R \cdot O(\frac{T^2\nu^2(n)2^{2n}}{R^2}) = O(\frac{T^2\nu^2(n)2^{2n}}{R})$; that is, $\|U_T - V_R\| =  O(\frac{T^2\nu^2(n)2^{2n}}{R})$.
This completes the proof of the lemma.
\end{proofof}

\subsection{Adiabatic Evolutionary Quantum Systems or AEQSs}\label{sec:AEQS}

Adiabatic quantum computing was initially sought to solve \emph{optimization problems}; on the contrary, the major target of this exposition is \emph{decision problems} (or equivalently, \emph{languages}). Instead of searching solutions of a computational problem as in \cite{FGG+01}, we wish to determine ``acceptance'' (yes) or ``rejection'' (no) of each instance given to the problem.

To lay out a suitable platform to carry out adiabatic quantum computation, we loosely adapt the key definition of Aharonov \etalc~\cite{ADK+07} but modify it significantly to fulfil our purpose of implementing the adiabatic quantum computation on a new, generic model, which we call an \emph{adiabatic evolutionary quantum system} (or an AEQS, pronounced as ``eeh-ks'').

\begin{definition}
An AEQS $\SSS$ is a septuple  $(m,\Sigma,\varepsilon, \{H^{(x)}_{ini}\}_{x\in\Sigma^*}, \{H^{(x)}_{fin}\}_{x\in\Sigma^*}, \{S^{(n)}_{acc}\}_{n\in\nat}, \{S^{(n)}_{rej}\}_{n\in\nat})$, where $m:\Sigma^*\to\nat$ is a size function, $\Sigma$ is an (input) alphabet, $\varepsilon$ is an accuracy bound in $[0,1]$,  both $H^{(x)}_{ini}$ and $H^{(x)}_{fin}$ are Hamiltonians acting on the same Hilbert space of $2^{m(x)}$ dimension (where this space is referred to as the system's \emph{evolution space}), and both $S^{(n)}_{acc}$ and $S^{(n)}_{rej}$ are subsets of $\{0,1\}^{n}$ whose elements respectively represent ``acceptance'' and ``rejection'' (where a pair $(S^{(n)}_{acc},S^{(n)}_{rej})$ is called an \emph{acceptance/rejection criteria pair}). The function $m$ indicates  the \emph{(system) size} of $\SSS$. We further demand that $H^{(x)}_{ini}$ and $H^{(x)}_{fin}$ should have \emph{unique} ground states.
\end{definition}

The system size $m$ of $\SSS$ expresses how large the evolution space of $\SSS$ is.
An adiabatic evolution process of an AEQS is similar to the one in Section \ref{sec:system-evolution}. Given an input string $x\in\Sigma^*$, letting $T_{x}$ denote the minimum evolution time of this system, we define $H^{(x)}(t)$ to be  $\left(1-\frac{t}{T_x}\right)H^{(x)}_{ini} + \frac{t}{T_x}H^{(x)}_{fin}$ for any real number $t\in[0,T_x]$.
We express the ground state of $H^{(x)}(t)$ as  $\qubit{\psi^{(x)}_g(t)}$.
At time $t=0$, the AEQS is initialized to be the ground state $\qubit{\psi^{(x)}_g(0)}$ of $H^{(x)}_{ini}$.
The system slowly evolves by applying $H^{(x)}(t)$ discretely from time $t=0$ to $t=T_x$.
This AEQS $\SSS$ is thought to run in time $T_x\cdot \max_{t\in[0,T_x]}\|H^{(x)}(t)\|$.

Let $S^{(x)}_0$ denote a set of basis vectors. We often assume that $H^{(x)}_{ini}$ is of the form $\sum_{\qubit{u}\in S^{(x)}_0} \nu(u) \density{\hat{u}}{\hat{u}}$, where  each $\nu(u)$ is a real eigenvalue associated with an eigenstate $\qubit{u}$ of $H^{(x)}_{ini}$.
The ground state of $H^{(x)}_{ini}$ is therefore of the form $\qubit{u_0}$ for a certain vector $\qubit{u_0}\in S^{(x)}_0$ satisfying   $\nu(u_0)=\min\{\nu(u)\mid \qubit{u}\in S^{(x)}_0\}$.
The adiabatic evolution eventually makes the system approach close enough to the ground state $\qubit{\psi^{(x)}_g(T_x)}$ of $H^{(x)}_{fin}$.
To solve a computational problem using the adiabatic evolution of a quantum system, as noted in \cite{FGG+01}, it suffices to encode a correct solution of the problem into $\qubit{\psi^{(x)}_g(T_x)}$.

To work on decision problems, on the contrary, we need to specify ``accepting'' and ``rejecting'' quantum states in the evolution space on which $H^{(x)}(t)$ acts during the time interval $[0,T_x]$. This task can be done by incorporating the two index sets $S^{(n)}_{acc}$ and $S^{(n)}_{rej}$ and by defining $QS^{(n)}_{acc}$ and $QS^{(n)}_{rej}$ to be the Hilbert spaces spanned respectively by the vectors in $\{\qubit{u}\mid u\in S^{(n)}_{acc}\}$ and $\{\qubit{u}\mid u\in S^{(n)}_{rej}\}$, where $n=m(x)$. 
These spaces $QS^{(n)}_{acc}$ and $QS^{(n)}_{rej}$ are respectively called the \emph{accepting space} and the \emph{rejecting space} and their elements are respectively called \emph{accepting quantum state} and \emph{rejecting quantum states} of $\SSS$.

The uniqueness of the ground state of $H^{(x)}_{fin}$ for every input string $x\in\Sigma^*$ ensures a unique outcome of each adiabatic quantum computation. Recall from Section \ref{sec:numbers} that all eigenstates in this exposition are \emph{normalized}.
When the ground state of $H^{(x)}_{fin}$ is sufficiently close to a certain normalized accepting (resp., rejecting) quantum state in  $QS^{(n)}_{acc}$ (resp., $QS^{(n)}_{rej}$), the AEQS is considered to \emph{accept} (resp., \emph{reject}) $x$.
As customary in computational complexity theory, we also say that the AEQS $\SSS$ \emph{outputs} $1$ (resp., $0$) if it accepts (resp., rejects).
The closeness of the ground state of $H^{(x)}_{fin}$ to either an accepting quantum state or a rejecting quantum state relates to the \emph{accuracy} of the AEQS's answer to the correct solution of the target decision problem.
Two AEQSs $\SSS_1$ and $\SSS_2$ over the same alphabet $\Sigma$ are \emph{(computationally) equivalent} if, for any input $x\in\Sigma^*$, the outcome of $\SSS_1$ on $x$ matches the outcome of $\SSS_2$ on $x$.

\begin{definition}
Given a decision problem $L$ and any constant $\varepsilon\in[0,1]$, we say that an AEQS $\SSS = (m,\Sigma,\varepsilon, \{H^{(x)}_{ini}\}_{x\in\Sigma^*}, \{H^{(x)}_{fin}\}_{x\in\Sigma^*}, \{S^{(n)}_{acc}\}_{n\in\nat}, \{S^{(n)}_{rej}\}_{n\in\nat})$ \emph{solves} (or \emph{recognizes}) $L$ \emph{with accuracy at least} $\varepsilon$ if (i) for each input $x\in\Sigma^*$, there exist two unique ground states $\qubit{\psi_g^{(x)}(0)}$ of $H^{(x)}_{ini}$ and $\qubit{\psi^{(x)}_{g}(T_x)}$ of $H^{(x)}_{fin}$, (ii) for any string $x\in L$, the ground state $\qubit{\psi^{(x)}_{g}(T_x)}$ is $\sqrt{2}(1-\varepsilon)$-close\footnote{It may be possible to use the notion of ``fidelity'' in place of the $\ell_2$-norm.} to a certain normalized accepting quantum state $\qubit{\phi_x}$ in $QS^{(m(x))}_{acc}$, and (iii) for any string $x\in \Sigma^*-L$, the ground state $\qubit{\psi^{(x)}_{g}(T_x)}$ is $\sqrt{2}(1-\varepsilon)$-close to a certain normalized rejecting quantum state $\qubit{\phi_x}$ in $QS^{(m(x))}_{rej}$.  The \emph{adiabatic quantum size complexity} of $L$ is $m(x)$, where ``$x$'' expresses a ``symbolic'' input.
Even for a promise decision problem $\LL=(L,nonL)$, we also say that $\SSS$ \emph{solves $\LL$ with accuracy at least $\varepsilon$} if Condition (i) for $x\in\Sigma^*$, Condition (ii) for $x\in L$, and Condition (iii) for $x\in\Sigma^*-L$ are met only for promised strings $x$. No condition is required for any non-promised inputs $x$.
\end{definition}

The closeness factor $\sqrt{2}(1-\varepsilon)$ of the above definition comes from the following reasoning. Since $\qubit{\phi_x}$ and $\qubit{\psi^{(x)}_{g}(T_x)}$ are normalized, the $\ell_2$-norm of the difference between them,
$\|\qubit{\phi_x}-\qubit{\psi^{(x)}_{g}(T_x)}\|_2$, equals $\sqrt{2(1-\cos\theta)}$ for a certain angle $\theta$. We reformulate  this last formula by setting $\varepsilon = 1- \sqrt{1-\cos\theta}$, which ranges over $[0,1]$, and we then obtain   $\|\qubit{\phi_x}-\qubit{\psi^{(x)}_{g}(T_x)}\|_2
=\sqrt{2}(1-\varepsilon)$, which is a \emph{linear function} in  $\varepsilon$.

Our formalism of AEQSs is inherently ``nonuniform'' in the sense that the construction (or designing) of each AEQS is allowed to vary drastically according to the choice of inputs.
Nonetheless, the usefulness of our AEQSs comes from the fact that they are powerful enough to recognize all possible languages. This will be the basis of our further study in Sections \ref{sec:behaviors}--\ref{sec:simulation} on how various restrictions of AEQSs affect the recognition of formal languages of different complexities.

\begin{lemma}\label{finite-size-AEQS}
For any language $L$ over alphabet $\Sigma$, there is a series $\SSS$ of AEQSs of system size $1$ such that $\SSS$ solves $L$ with accuracy $1$.
\end{lemma}

\begin{proof}
The key of the following proof rests on the appropriate choice of Hamiltonians, which heavily relies on individual input strings. Let $\Sigma$ be any alphabet and let $L$ be any language over $\Sigma$. Here, we use the same notation $L$ to denote its characteristic function.

Let us define the desired AEQS $\SSS = (m,\Sigma,\varepsilon, \{H^{(x)}_{ini}\}_{x\in\Sigma^*}, \{H^{(x)}_{fin}\}_{x\in\Sigma^*}, \{S^{(n)}_{acc}\}_{n\in\nat}, \{S^{(n)}_{rej}\}_{n\in\nat})$ for $L$ in the following way. Fix an arbitrary string $x\in\Sigma^*$. We set $m(x)=1$ and $\varepsilon=1$, and we define   $H^{(x)}_{ini} = \density{\hat{1}}{\hat{1}}$ and $H^{(x)}_{fin} = \density{\overline{L}(x)}{\overline{L}(x)}$, where  $\overline{L}(x)=1-L(x)$.
Moreover, we set $S^{(n)}_{acc}=\{1\}$ and $S^{(n)}_{rej}=\{0\}$ for any index $n\in\nat$. 

Since $H^{(x)}_{ini}\qubit{\hat{0}} = 0$, the ground state of $H^{(x)}_{ini}$ is $\qubit{\hat{0}}$. Similarly, for each $x\in\Sigma^*$, the ground state of $H^{(x)}_{fin}$ is $\qubit{L(x)}$ because $H^{(x)}_{fin}\qubit{L(x)} =0$. It thus follows that $x\in L$ iff $\SSS$ outputs $L(x)$.
Therefore, the accuracy of $\SSS$ must be exactly $1$. Since $x$ is arbitrary, we conclude that $\SSS$ solves $L$ with accuracy $1$.
\end{proof}

\subsection{Conditional AEQSs or AEQS($\FF$)}\label{sec:conditional-AEQS}

Lemma \ref{finite-size-AEQS} guarantees that it is always possible to  construct an appropriate AEQS for any given language. This suggests that we can discuss the computational complexity of languages simply by placing a ``maximal'' amount of conditions (or restrictions) on the behaviors of AEQSs so that  the resulted AEQSs remain sufficiently powerful to recognize the languages.
To describe such conditions and study their direct influence to AEQSs, we consider AEQSs restricted to a set $\FF$ of ``natural conditions''  on two Hamiltonians of the AEQSs. Of those ``conditional'' AEQSs, we are interested only in the ones whose accuracy is relatively high, in particular, at least $\frac{1}{2}+\eta$ for a fixed constant $\eta>0$.
To denote a family of decision problems solved by such highly-accurate  conditional AEQSs under a given conditional set $\FF$, we use the abbreviation  of ``highly-accurate $\mathrm{AEQS}(\FF)$''.

\begin{definition}
Let $\FF$ indicate a set of conditions imposed on Hamiltonians of AEQSs. The complexity class, \emph{highly-accurate} $\mathrm{AEQS}(\FF)$, is the collection of all languages $L$ for which there exist an AEQS $\SSS$ and an accuracy bound $\varepsilon\in(1/2,1]$ satisfying that
$\SSS$ recognizes $L$ with accuracy at least $\varepsilon$ and the  Hamiltonians of the AEQS meet all the conditions specified by $\FF$. Since we discuss only highly-accurate AEQS's in the subsequent sections, we often drop the prefix ``highly-accurate'' and simply call them $\aeqs{\FF}$ unless stated otherwise. The definition of $\aeqs{\FF}$ can be naturally extended into promise decision problems.
\end{definition}

The above notion gives us freedom to discuss various types of conditions, which will play essential roles in determining the computational complexity of languages in later sections.

Of all possible types of conditions, we are primarily interested in the following four types of conditions.

\s
(1)
Firstly, we are interested in how efficiently we can generate two Hamiltonians of AEQSs \emph{in an algorithmic way} since these AEQSs are dictated by such Hamiltonians. In particular, we study the case where these   Hamiltonians are generated by certain forms of quantum quasi-automata.
An AEQS $\SSS$ with $\{H^{(x)}_{ini}\}_{x\in\Sigma^*}$ and $\{H^{(x)}_{fin}\}_{x\in\Sigma^*}$ is said to be \emph{generated by 1moqqaf's} if there exist two 1moqqaf's $\MM_0$ and $\MM_1$ working over $\Sigma$ that respectively generate $\{H^{(x)}_{ini}\}_{x\in\Sigma^*}$ and $\{H^{(x)}_{fin}\}_{x\in\Sigma^*}$. We use the notation $\FF=$``1moqqaf''
to mean the use of 1moqqaf's in order to generate Hamiltonians of AEQSs. To express the use of 1qqaf's, in contrast, we use the notation $\FF=$``1qqaf''.
We further expand these definitions to \emph{time-bounded 2qqaf's} as well as \emph{time-bounded 1.5qqaf's}. For 2qqaf's,
$\FF=$``ltime-2qqaf'' and $\FF=$``ptime-2qqaf'' refer to linear-time 2qqaf's and polynomial-time 2qqaf's, respectively, and for 1.5qqaf's, the notation $\FF=$``ltime-1.5qqaf'' indicates the use of linear-time 1.5qqaf's.

\s
(2) We are mostly concerned with the \emph{(system) size} $m$ of AEQSs. Since  the size of an AEQS equals the logarithm of the dimension of its Hamiltonians, it relates to the size of the evolution space of the AEQS.
We write $\FF=$``constsize'' (constant size) to indicate the case where the size of an AEQS is  $O(1)$. In a similar way, we write  $\FF=$``logsize'' (logarithmic size), $\FF=$``linsize'' (linear size), and $\FF=$``polysize'' (polynomial size) to express that the size of an AEQS is $O(\log{n})$, $O(n)$, and $n^{O(1)}$, respectively.

\s
(3) We further need to pay extra attention to the value of the \emph{spectral gap} of each final Hamiltonian $H^{(x)}_{fin}$ of an AEQS on each input $x$  because, by the adiabatic theorem in Section \ref{sec:system-evolution}, this value provides an upper bound of the runtime of the AEQS.
For instance, if the final Hamiltonians of an AEQS have  an  inverse-polynomially large spectral gap, then the adiabatic evolution of the AEQS takes only polynomially many steps. We remark that, even if Hamiltonians are generated by 1qqaf's, their spectral gaps may not be  guaranteed to be inverse-polynomially large. To express a required lower bound of the spectral gap, we first introduce the notation $\FF=$``polygap'' to mean that the spectral gap is lower-bounded by $1/n^{O(1)}$ (i.e., inverse-polynomially large). A similar notation  $\FF=$``constgap'' indicates that the spectral gap is at least $1/O(1)$ (i.e., inverse-constantly large).

\s
(4) Finally, we look into the \emph{ground energy levels} of final Hamiltonians of an AEQS. In certain cases \cite{FGG+01,FGGS00},  it is possible to set the ground energy of every final Hamiltonian to be $0$. This motivates us to introduce the notation $\FF=$``0-energy'' for the situation where the ground energy of the final Hamiltonian $H^{(x)}_{fin}$ is $0$ for every input $x$.

\s
In the subsequent section, we will demonstrate how to design (or program) AEQSs with various conditions for six simple languages.

\section{Behaviors of AEQSs and Their Designing}\label{sec:behaviors}

We have introduced  in Section \ref{sec:AEQS-def} the basic \emph{adiabatic evolutionary quantum systems}  (AEQSs) and their conditional variants $\aeqs{\FF}$. This section further demonstrates  how to design (or program) such conditional AEQSs for six simple example languages.

\subsection{How to Design (or Program) AEQSs}\label{sec:example}

In this exposition, AEQSs are the basic platform to discuss the computational complexity of any given language in such a way that
the difficulty in constructing Hamiltonians of AEQSs can be viewed as a reasonable complexity measure of the languages.
To help understand this viewpoint, it is beneficial to see how to design (or program) AEQSs for specific languages.
In particular, we intend to present various methods of designing AEQSs for six  simple languages (actually the last two examples are \emph{promise problems}).  Even though we do not attempt to seek for the best possible AEQSs, these examples will serve as bases to more general claims made in the subsequent sections. In what follows,  $\SSS$ denotes an AEQS to be constructed and it is assumed to have the form $(m,\Sigma,\varepsilon, \{H^{(x)}_{ini}\}_{x\in\Sigma^*}, \{H^{(x)}_{fin}\}_{x\in\Sigma^*}, \{S^{(n)}_{acc}\}_{n\in\nat}, \{S^{(n)}_{rej}\}_{n\in\nat})$, and we will describe the desired AEQS simply by specifying each  element of $\SSS$.


\begin{example}\label{example:L_a}
Given a fixed string $a\in\{0,1\}^+$, the language $L_a= \{ax\mid x\in\{0,1\}^*\}$ belongs to $\aeqs{1moqqaf,logsize,constgap,0\mbox{-}energy}$.
\end{example}

This language $L_a=\{ax\mid x\in\{0,1\}^*\}$ is regular for each fixed string $a\in\{0,1\}^+$. Hereafter, we intend to construct the desired AEQS $\SSS$ for $L_a$.
For readability, we consider only the simplest case where $a=0$ since the other cases can be  treated in essentially the same way.

For simplicity, we write $\Sigma$ for $\{0,1\}$ and set $Q= \{q_0,q_1,q_2,q_3\}$. In connection to 1moqqaf's, it is more convenient to identify $00$, $01$, $10$, and $11$ respectively with  $q_0$, $q_1$, $q_2$, and $q_3$ and, as noted in Section \ref{sec:QFA}, we can express each number in $[0,n+1]_{\integer}$ using ${\ilog(n+2)}$ bits.
We then define $IND_n$ to be the index set $Q\times [0,n+1]_{\integer}$. Our selector $\mu$ is defined as $\mu(x)=|x|$ for any string $x\in\Sigma^*$.
Given an input string $x\in\Sigma^*$, we define
$m(x)= {\ilog|IND_{\mu(x)}|}$, which is at most $2+{\ilog(|x|+1)}$. Informally, we treat $IND_{\mu(x)}$ as $\{0,1\}^{m(x)}$. 
To complete the construction of $\SSS$,
we further set $S^{(m(x))}_{acc} = \{(q_1,0)\}$ and $S^{(m(x))}_{rej} = \{(q_2,0)\}$. Notice that $S^{(m(x))}_{acc},S^{(m(x))}_{rej}\subseteq IND_{\mu(x)}$. 
We define $\{H^{(x)}_{ini}\}_{x\in\Sigma^*}$  simply by setting
$H^{(x)}_{ini} = W^{\otimes{m(x)}} \diag(0,1,1,\cdots,1) (W^{\otimes{m(x)}})^{\dagger}$ for each $x\in\Sigma^*$.  
Since it is relatively easy to generate $H^{(x)}_{ini}$ by 1moqqaf's, we hereafter intend to concentrate on  $H^{(x)}_{fin}$ and construct its associated 1moqqaf $\MM=\{M_{n}\}_{n\in\nat}$.
We assume that $M_{n}$ has the form $(Q^{(n)},\Sigma,\{\cent,\dollar\}, \{A^{(n)}_{\sigma}\}_{\sigma\in\check{\Sigma}}, \Lambda^{(n)}_{0})$ and we wish to define its components one by one. Firstly, we define   $Q^{(n)}= IND_n$ and set $\Lambda^{(n)}_{0} = \sum_{u\in IND^{(-)}_n} \density{u}{u}$, where $IND^{(-)}_n = IND_n-\{(q_0,0)\}$. Each $A^{(n)}_{\sigma}$ will be defined below.

Given a parameter $n\in\nat$, the desired machine $M_n$ starts in cell $0$ with the configuration $(q_0,0)$ and moves its tape head to the right. Whenever it scans $0$ (resp., $1$) on cell $1$, it enters $(q_1,2)$ (resp., $(q_2,2)$) from $(q_0,1)$. After leaving cell $1$, $M_n$ increases the second component of $(q,h)$ with preserving
the first component intact.
More formally, we define unitary operators $U^{(n)}_{\sigma}$ as follows.
Let $U^{(n)}_{\cent}\qubit{q,h}=\qubit{q,h+1\;\mathrm{mod}\;N'}$ and  $U^{(n)}_{\dollar} = U^{(n)}_{\cent}$,  where $N'=n+2$.
In addition, let
$U^{(n)}_{0}\qubit{q_0,1} =\qubit{q_1,2}$,
$U^{(n)}_{0}\qubit{q_1,1} =\qubit{q_3,2}$,
$U^{(n)}_{0}\qubit{q_2,1} =\qubit{q_2,2}$,
$U^{(n)}_{1}\qubit{q_0,1} =\qubit{q_2,2}$,
$U^{(n)}_{1}\qubit{q_1,1} =\qubit{q_1,2}$, and
$U^{(n)}_{1}\qubit{q_2,1} =\qubit{q_3,2}$.
Given a symbol $\sigma\in\{0,1\}$, we set   $U^{(n)}_{\sigma}\qubit{q_3,1}=\qubit{q_0,2}$ and $U^{(n)}_{\sigma}\qubit{q,h} = \qubit{q,h+1\;\mathrm{mod}\;N'}$ for any $q\in Q$ and any $h\in[0,N]_{\integer}-\{1\}$.
These definitions lead to the conclusion that $U^{(n)}_{\cent{0y}\dollar}\qubit{q_0,0} = \qubit{q_1,0}$ and $U^{(n)}_{\cent{1y}\dollar}\qubit{q_0,0} = \qubit{q_2,0}$ for any string $y$.
The desired quantum operation $A^{(n)}_{\sigma}$ is finally set to be     $A^{(n)}_{\sigma}(H) = U^{(n)}_{\sigma} H (U^{(n)}_{\sigma})^{\dagger}$ for any linear operator $H$.
The final Hamiltonian $H^{(x)}_{fin}$ is then defined as  $A^{(n)}_{\cent{x}\dollar}(\Lambda^{(n)}_0)$ so that  $\{H^{(x)}_{fin}\}_{x\in\Sigma^*}$ is generated by $\MM$.
Note that the condition $[H^{(x)}_{ini},H^{(x)}_{fin}]\neq O$ is satisfied.

Next, we wish to verify that $\SSS$ correctly solves $L_a$ with accuracy $1$. Assume that $x=0y$ for a certain string $y$ and set $n=\mu(x)$.
Consider the quantum state $\qubit{\phi_1} = \qubit{q_1,0}$.
By the definition of $A^{(n)}_{\sigma}$'s, it follows that  $A^{(n)}_{\cent{x}\dollar}(\Lambda^{(n)}_0)\qubit{\phi_1} = U^{(n)}_{\cent{x}\dollar} \Lambda^{(n)}_0 (U^{(n)}_{\cent{x}\dollar})^{\dagger} \qubit{\phi_1} = U^{(n)}_{\cent{x}\dollar} \Lambda^{(n)}_0 (U^{(n)}_{\cent{x}\dollar})^{\dagger} U^{(n)}_{\cent{x}\dollar}\qubit{q_0,0} = U^{(n)}_{\cent{x}\dollar}\Lambda^{(n)}_0\qubit{q_0,0} = 0$ since $\qubit{\phi_1} = U^{(n)}_{\cent{x}\dollar}\qubit{q_0,0}$ and $\Lambda^{(n)}_0\qubit{q_0,0}=0$.
Thus, $\qubit{\phi_1}$ is the ground state of $H^{(x)}_{fin}$ with a  ground energy of $0$.
Since $U^{(n)}_{\cent{x}\dollar}$ is unitary and $\Lambda^{(n)}_0=diag(0,1,\ldots,1)$ of rank $2^{m(x)}-1$, all other eigenvalues are $1$, and therefore the spectral gap must be $1$. Obviously, $\qubit{\phi_1}$ belongs to $QS^{(m(x))}_{acc}$.
In a similar way, when $x=1y$, $\qubit{q_2,0}$ is the ground state of $H^{(x)}_{fin}$ with a ground energy of $0$ since  $A^{(n)}_{\cent{x}\dollar}(\Lambda^{(n)}_0)\qubit{q_2,0} =0$.
Moreover, $\qubit{q_2,0}$ falls in $QS^{(m(x))}_{rej}$.
As a result, we conclude that $\SSS$ solves $L_a$ with accuracy $1$, as requested.


\begin{example}
The language
$Equal =\{w \mid \#_a(w)=\#_b(w)\}$ over the binary alphabet $\Sigma=\{a,b\}$ is in $\aeqs{1moqqaf,logsize,0\mbox{-}energy}$.
\end{example}

The language $Equal$ is {reversible context-free},  where  a \emph{reversible context-free language} is recognized by an appropriately chosen \emph{reversible pushdown automaton} \cite{KM12}.
A similar language $L_{eq} =\{a^nb^n\mid n\in\nat\}$, however, is not  reversible context-free \cite{KM12} but it is proven to be recognized by a certain 1.5qfa \cite{KW97}.
Let us design an AEQS for $Equal$ under the desired conditions stated in this example.

With the use of $Q=\{q_1,q_2\}$ and for a parameter $n\in\nat^{+}$, we define $IND_{n} = Q\times [0,N-1]_{\integer}$ and $IND^{(-)}_n = IND_n-\{(q_1,0)\}$, where $N$ indicates $2^{n-1}$. Notice that $|IND_n|=2N$.
Each number in $[0,N-1]_{\integer}$ can be expressed using ${\ilog{N}}$ bits, as noted in Section \ref{sec:QFA}.
Letting $\mu(x)=|x|$ for any $x\in\Sigma^*$, the size $m(x)$ of $\SSS$ is defined to be ${\ilog|IND_{\mu(x)}|}$, which is at most ${\ilog{2(|x|+1)}} +1$,
and thus $m(x)$ is $O(\log{|x|})$.
Moreover, we set $S^{(m(x))}_{acc} = \{(q_1,3l+1)\mid l\in\nat, \mu(x) = 2l+1 \}$ and $S^{(m(x))}_{rej}= IND_{n} - S^{(m(x))}_{acc}$.

Fix an arbitrary input $x\in\Sigma^*$ and set $n=\mu(x)$.
Hereafter, we intend to construct the desired $2N$-dimensional Hamiltonians $H^{(x)}_{ini}$ and $H^{(x)}_{fin}$ of $\SSS$.
We first define $H^{(x)}_{ini}$ as $W^{\otimes m(x)}\Lambda^{(n)}_{0} (W^{\otimes m(x)})^{\dagger}$, where $\Lambda^{(n)}_{0} = \sum_{u\in IND^{(-)}_n} \density{u}{u}$. To construct $H^{(x)}_{fin}$, in contrast, we need to  define
an appropriate 1qqaf $\MM=\{M_n\}_{n\in\nat}$ with $M_n$ having the form $(Q^{(n)},\Sigma, \{\cent,\dollar\}, \{A^{(n)}_{\sigma}\}_{\sigma\in\check{\Sigma}}, \Lambda^{(n)}_{0})$ for  any index $n\in\nat$. Let $Q^{(n)}=IND_n$.
Each inner state of $M_n$ encodes values of both $Q$ and a special \emph{internal clock}. Along one computation path of $M_n$, while reading $a$, this clock moves twice as fast as reading $b$. In another computation path of $M_n$, the clock moves in the other way round. In the end, we quantumly check whether or not the both clocks in two different paths show the same time. We begin with the definition of unitary matrices  $U^{(n)}_{\sigma}$.
A basic idea of constructing $U^{(n)}_{\sigma}$ is to increase the value $i$ in a pair $(q,i)\in Q^{(n)}$ by $2$ and by $1$ when scanning $a$ in inner states $q_1$ and $q_2$, respectively. When scanning $b$, we do the same after exchanging between $2$ and $1$.

Let $U^{(n)}_{\cent}\qubit{q_1,0} = \frac{1}{\sqrt{2}}\qubit{q_1,1}+ \frac{1}{\sqrt{2}}\qubit{q_2,1}$, $U^{(n)}_{\cent}\qubit{q_2,0} = \frac{1}{\sqrt{2}}\qubit{q_1,1} - \frac{1}{\sqrt{2}}\qubit{q_2,1}$,
$U^{(n)}_{a}\qubit{q_1,i} = \qubit{q_1,i+2\;\mbox{mod}\;N}$, $U^{(n)}_{a}\qubit{q_2,i} = \qubit{q_2,i+1\;\mbox{mod}\;N}$,
$U^{(n)}_{b}\qubit{q_1,i} = \qubit{q_1,i+1\;\mbox{mod}\;N}$, and
$U^{(n)}_{b}\qubit{q_2,i} = \qubit{q_2,i+2\;\mbox{mod}\;N}$. Concerning $\dollar$, let $U^{(n)}_{\dollar}\qubit{q_1,i} = \frac{1}{\sqrt{2}}\qubit{q_1,i+1\;\mbox{mod}\;N} + \frac{1}{\sqrt{2}}\qubit{q_2,i+1\;\mbox{mod}\;N}$ and $U^{(n)}_{\dollar}\qubit{q_2,i} = \frac{1}{\sqrt{2}}\qubit{q_1,i+1\;\mbox{mod}\;N} - \frac{1}{\sqrt{2}}\qubit{q_2,i+1\;\mbox{mod}\;N}$.
It then follows that, for any $x$ with $k=\#_a(x)$ and $l=\#_b(x)$, $U^{(n)}_{\cent x\dollar} \qubit{q_1,0}$ equals $\frac{1}{2} (\qubit{q_1,2k+l+1} + \qubit{q_1,k+2l+1}) +  \frac{1}{2} (\qubit{q_2,2k+l+1} - \qubit{q_2,k+2l+1})$ and $U^{(n)}_{\cent x\dollar} \qubit{q_2,0}$ equals $\frac{1}{2} (\qubit{q_1,2k+l+1} - \qubit{q_1,k+2l+1}) +  \frac{1}{2} (\qubit{q_2,2k+l+1} + \qubit{q_2,k+2l+1})$.
For every symbol  $\sigma\in\check{\Sigma}$, the desired quantum operation $A^{(n)}_{\sigma}(H)$ is set to be $U^{(n)}_{\sigma} H (U^{(n)}_{\sigma})^{\dagger}$.
In the end, we define $H^{(x)}_{fin}$ as $A^{(n)}_{\cent x\dollar}(\Lambda^{(n)}_0)$.

We still need to show that our AEQS $\SSS$ correctly recognizes $Equal$ with accuracy $1$.
let us take the quantum state $\qubit{\phi_x} = U^{(n)}_{\cent{x}\dollar}\qubit{q_1,0}$.
In the case of $x\in Equal$ with $\#_a(x)=\#_b(x)=l$, since $\qubit{\phi_x}=\qubit{q_1,3l+1}$, we obtain $\qubit{\phi_x}\in QS^{(m(x))}_{acc}$.
Moreover, since $\Lambda^{(n)}_0\qubit{q_1,0}=0$, we conclude that
$H^{(x)}_{fin}\qubit{\phi_x} = U^{(n)}_{\cent x\dollar} \Lambda^{(n)}_0 (U^{(n)}_{\cent x\dollar})^{\dagger} U^{(n)}_{\cent x \dollar} \qubit{q_1,0} = U^{(n)}_{\cent x\dollar}
\Lambda^{(n)}_0 \qubit{q_1,0} = 0$.
On the contrary, let us consider the case of $x\notin Equal$ with $k=\#_a(x)$ and $l=\#_b(x)$. In this case, $\qubit{\phi_x}$ belongs to $QS^{(m(x))}_{rej}$ because of $\measure{q_1,3l+1}{\phi_x}=0$.
Furthermore, it follows that $U^{(n)}_{\cent x \dollar} \Lambda^{(n)}_0 (U^{(n)}_{\cent x\dollar})^{\dagger} \qubit{\phi_x} = U^{(n)}_{\cent x \dollar} \Lambda^{(n)}_0 \qubit{q_1,0} =0$. This shows that the ground energy is $0$.


\begin{example}\label{ex:Pal}
Consider the set $Pal_{\#}$ of \emph{marked (even-length) palindromes},  that is, $Pal_{\#} =\{w\# w^R\mid w\in\{a,b\}^*\}$ over the ternary alphabet $\Sigma=\{a,b,\#\}$. This language $Pal_{\#}$ is in $\aeqs{ltime\mbox{-}2qqaf,logsize}$.
\end{example}

Similar to $Equal$, the above language $Pal_{\#}$ is also {reversible context-free} \cite{KM12}.
Similar to the language $\{w\in\Sigma^*\mid w=w^R\}$ used in  \cite{AW02}, $Pal_{\#}$ can be recognized by an appropriately chosen  \emph{2-way quantum finite automaton with a classical head}.

Hereafter, we construct the desired AEQS $\SSS$ for $Pal_{\#}$.
We first prepare the index set $IND_n$, which is defined to be $(Q \times [0,2]_{\integer}\times [0,n+1]_{\integer})^2$ with $Q = \{q_1,q_2,q_3,q_4,q_5\}$ for any parameter $n\in\nat$.
For convenience, let $\xi_0=(q_1,1,0)$ and let $IND^{(-)}_n = IND_n -\{(q_0,1,0,\xi_0)\}$.
We then set $S^{(n)}_{acc} = \{(q_1,0,0,\xi_0)\}$ and $S^{(n)}_{rej} =\{(q_i,0,0,\xi_0)\mid i\in\{4,5\}\}$.
The selector $\mu$ is simply defined by $\mu(x)=|x|$ for all $x\in\Sigma^*$.
The size $m(x)$ of $\SSS$ is then defined as ${\ilog|IND_{\mu(x)}|}$, which is at most $2({\ilog|Q|} + 2 + {\ilog(|x|+1)}) = O(\log|x|)$.

Given an arbitrary input $x$, the desired initial Hamiltonian $H^{(x)}_{ini}$ is set to be $\sum_{u\in IND^{(-)}_{\mu(x)}} \density{\hat{u}}{\hat{u}}$ of rank $2^{m(x)}-1$. 
To define the final Hamiltonian $H^{(x)}_{fin}$, however, we need to introduce a  linear-time 2qqaf $\MM=\{M_n\}_{n\in\nat}$, in which each machine $M_n$ works in the configuration space spanned by the elements in $IND_n$.
Let us fix an arbitrary input $x$ and set $n=\mu(x)$. For brevity, we write $N$ for $n+1$.
Recall that $x_{(i)}$ denotes the $i$th symbol of $x$ for every index  $i\in[n]$. We further set $x_{(0)}=\cent$ and $x_{(n+1)}=\dollar$.
For the ease of our description, we assume that $x$ has the form $w_1\# w_2$ for certain strings $w_1,w_2\in\{a,b\}^*$.
The machine $M_n$ operates using six registers of the form $(q,k,h,q',k',h')$, where $q$ is an inner state, $k$ refers to either $0$ or the current ``phase'' number, $h$ refers
to a tape head location, and $(q',k',h')$ is a starting value of $(q,k,h)$.
There are two phases to execute separately.
In the first phase, we try to produce accepting quantum states $\qubit{q,0,0,\xi_0}$ and, in the second phase, our goal is to produce  rejecting quantum states $\qubit{q_i,0,0,\xi_0}$ for any index $i\in\{4,5\}$. For convenience, we write $Q'$ for $\{q_1,q_2,q_3\}$.

We take the first step by applying the set $\{K^{(n)}_{\cent,1},K^{(n)}_{\cent,2}\}$ of Kraus operators defined by $K^{(n)}_{\cent,1}\qubit{q,k,h,q',k',h'} = \measure{\xi_0}{q,k,h}\cdot \qubit{q,k,h,q',k',h'}$ and $K^{(n)}_{\cent,2}\qubit{q,k,h,q',k',h'} = \sum_{z\neq \xi_0}\measure{z}{q,k,h}\cdot \qubit{q',k',h',q,k,h}$. This step leads us to concentrate only on quantum states of the form $\qubit{q,k,h,\xi_0}$ by applying the identity operator whenever these
registers do not contain $\qubit{\xi_0}$ in the following two phases.
Hereafter, the last three registers are assumed to be $\qubit{\xi_0}$.

In the first phase, we start with the quantum state $\qubit{q_1,1,0,\xi_0}$ stored in the six registers, change $\qubit{q_1,1,0,\xi_0}$ to $\qubit{q_1,1,1,\xi_0}$ at scanning $\cent$, and move $M_n$'s tape head to the right. In scanning each input symbol, we apply two unitary operators $\{U_a,U_b\}$ to the first register together with increasing the value in the third register by one. Those two operators are defined as   $U_a\qubit{q_1} =\frac{4}{5}\qubit{q_1}-\frac{3}{5}\qubit{q_2}$, $U_a\qubit{q_2} = \frac{3}{5}\qubit{q_1}+\frac{4}{5}\qubit{q_2}$,  $U_a\qubit{q_3} =\qubit{q_3}$,
$U_b\qubit{q_1} =\frac{4}{5}\qubit{q_1}-\frac{3}{5}\qubit{q_3}$, $U_b\qubit{q_2} =\qubit{q_2}$, and
$U_b\qubit{q_3} = \frac{3}{5}\qubit{q_1}+\frac{4}{5}\qubit{q_3}$.
In addition, let $U_{\#}=U'_{\#}=I$.
After scanning $\#$, we apply $U'_a=U^{-1}_a$ and $U'_b=U^{-1}_b$ instead of $U_a$ and $U_b$, respectively.
Formally, for any $q\in Q'$, if $h$ is in the range $[1,|w_1\#|]_{\integer}$, then we apply  $K_1^{(n,x)}\qubit{q,1,h,\xi_0} = U_{x_{(h)}}\qubit{q}\otimes \qubit{1,h+1,\xi_0}$, and if $h$ is in $[|w_1\#|+1,n]_{\integer}$, then we apply $K_1^{(n,x)}\qubit{q,1,h,\xi_0} = U'_{x_{(h)}}\qubit{q}\otimes \qubit{1,h+1,\xi_0}$.
In scanning $\dollar$,
we change $\qubit{q_1,1,n+1,\xi_0}$  to $\qubit{q_1,0,0,\xi_0}$ and $\qubit{q_i,1,n+1,\xi_0}$ to $\qubit{q_i,2,0}$ for every index  $i\in\{2,3\}$.
Note that, since the input tape of $M_n$ is \emph{circular}, the tape head automatically moves to $\cent$ after processing $\dollar$.
Once we enter an accepting quantum state, we stay in the same state except for an increment of the third register.
Formally, we demand that $\{K^{(n,x)}_1,K^{(n,x)}_2\}$ should satisfy
each of  the following conditions: $K_1^{(n,x)}\qubit{q_1,1,n+1,\xi_0} = \qubit{q_1,0,0,\xi_0}$, $K_1^{(n,x)}\qubit{q_i,1,n+1,\xi_0} = \qubit{q_i,2,0,\xi_0}$,  $K_2^{(n,x)}\qubit{q_1,0,h,\xi_0} = \qubit{q_1,0,h+1,\xi_0}$, and $K_2^{(n,x)}\qubit{q_1,0,n+1,\xi_0}=\qubit{q_1,0,0,\xi_0}$ for any  $i\in\{2,3\}$ and any $h\in[0,n]_{\integer}$.

In the second phase starting with $\qubit{q_i,2,0,\xi_0}$ at $\cent$ for each index $i\in\{2,3\}$, as we scan the input symbols one by one, we randomly choose $q_i$ and $q_{i+2}$ respectively with probabilities  $(\frac{1}{25})^2$ and $(\frac{4\sqrt{39}}{25})^2$ and then move the tape head to the right.
In reading $\dollar$, we change $\qubit{q_i,2,n+1,\xi_0}$ to $\qubit{q_{i+2},0,0,\xi_0}$ and $\qubit{q_{i+2},2,n+1,\xi_0}$ to $\qubit{q_{i+2},2,0,\xi_0}$, and then we move the tape head to $\cent$. For any index  $i\in\{2,3\}$ and any location $h\in[0,n]_{\integer}$,
we formally set
$K_1^{(n,x)}\qubit{q_i,2,h,\xi_0} = \frac{1}{25}\qubit{q_i,2,h+1,\xi_0}$,
$K_1^{(n,x)}\qubit{q_i,2,n+1,\xi_0} =\qubit{q_{i+2},0,0,\xi_0}$,  $K_1^{(n,x)}\qubit{q_{i+2},2,h,\xi_0} = \qubit{q_{i+2},2,h+1,\xi_0}$, and $K_1^{(n,x)}\qubit{q_{i+2},2,n+1,\xi_0} = \qubit{q_{i+2},2,0,\xi_0}$.
As for $K^{(n,x)}_2$, we further set
$K_2^{(n,x)}\qubit{q_i,2,h,\xi_0} = \frac{4\sqrt{39}}{25}\qubit{q_{i+2},2,h+1,\xi_0}$, $K_2^{(n,x)}\qubit{q_{i+2},0,h,\xi_0} = \qubit{q_{i+2},0,h+1,\xi_0}$, and
$K_2^{(n,x)}\qubit{q_{i+2},0,n+1,\xi_0} = \qubit{q_{i+2},0,0,\xi_0}$.

A single step of $M_n$ is made by an application of the quantum operation $A^{(n,x)}$ defined by $A^{(n,x)}(H) = \sum_{i\in\{1,2\}} K_i^{(n,x)} H (K_i^{(n,x)})^{\dagger}$ for any linear operator $H$. We write $(A^{(n,x)})^{k}(H)$ for the $k$ applications of $A^{(n,x)}$ to $H$.
The number of steps taken by $M_n$ except for the first one is exactly   $\ell(n)=2n+3$.
The final Hamiltonian $H^{(x)}_{fin}$ is then defined as $(A^{(n,x)})^{\ell(n)}(\tilde{\Lambda}^{(n)}_0)$, where
$\Lambda^{(n)}_0 = I - \frac{24}{25}\density{\xi_0,\xi_0}{\xi_0,\xi_0}$ and $\tilde{\Lambda}^{(n)}_0 = \sum_{i\in\{1,2\}} K^{(n)}_{\cent,i} \Lambda^{(n)}_0 (K^{(n)}_{\cent,i})^{\dagger}$.

Next, we argue that $\SSS$ correctly solves $Pal_{\#}$
with accuracy $1$.
Let us consider the case of $x\in Pal_{\#}$ and assume that $x$ has the form $w\# w^R$ for a certain nonempty string $w$.
Consider the quantum state $\qubit{\phi_x} = \qubit{\xi_0,\xi_0}$.
After the first phase, the first register of $M_n$ returns to $\qubit{q_1}$; in other words, we obtain $\bra{\phi_x} (K^{(n,x)}_1)^{n+2} \ket{\phi_x} =1$. Since $\bra{\phi_x}\Lambda^{(n)}_0\ket{\phi_x} = \frac{1}{25}$ and the other entries of $\Lambda^{(n)}_0$ are all $1$, $\qubit{\phi_x}$ must be  the ground state of $H^{(x)}_{fin}$. Since $\qubit{\phi_x}$ belongs to $QS^{(m(x))}_{acc}$, $\SSS$ accepts $x$ with accuracy $1$.
In contrast, let us consider the case where $x\notin Pal_{\#}$.
By the choice of amplitudes in $U_{\sigma}$, an analysis similar to \cite{AW02} (also \cite{AY15}) shows that,
after the first phase,  $(\frac{1}{25})^{n+1}\leq \bra{q_1,0,0,\xi_0}(K_1^{(n,x)})^{n+2}\ket{\phi_x} < 1$ and $(\frac{1}{25})^{n+1}\leq \sum_{i\in\{2,3\}} \bra{q_i,2,0,\xi_0}(K_1^{(n,x)})^{n+2}\ket{\phi_x} < 1$.
For convenience, for each index $i\in\{2,3\}$, let $\alpha_i=  \bra{q_i,2,0,\xi_0}(K_1^{(n,x)})^{n+2}\ket{\phi_x}$.
After the second phase, we can observe $(q_{i+2},0,0,\xi_0)$ with probability exactly $|\alpha_i|^2(\frac{1}{25})^{2(n+1)}$, which is much smaller than $(\frac{1}{25})^{n+1}$. Let $\qubit{\psi_x}$ denote the ground state of $H^{(x)}_{fin}$. We then conclude that $\qubit{\psi_x}$ does not include $\qubit{q_1,0,0,\xi_0}$, and thus $\qubit{\psi_x}$ falls in $QS^{(m(x))}_{rej}$. This indicates that $\SSS$ rejects $x$ with accuracy $1$.


\begin{example}
The language $SymCoin = \{x \mid \exists i,j\in[|x|] (x_{(i)}=x_{(j)} \wedge i<j \wedge i+j=|x|+1)\}$ (symmetric coincidence) over the alphabet $\Sigma=\{a,b\}$ is in $\aeqs{1qqaf,logsize,polygap}$.
\end{example}

This language $SymCoin$ is clearly \emph{context-sensitive} and the  well-known pumping lemma for $\cfl$ \cite{BPS61} can prove  that $SymCoin$ is not context-free. It is also known to be {non-stochastic language} \cite{FYS10}, where a \emph{stochastic language} is recognized by a certain one-way probabilistic finite automaton with unbounded-error probability.

Fix an input $x\in\Sigma^*$ arbitrarily and set $n=|x|$.
Let us assume that $n$ is even. The case for an odd $n$ is in essence similarly handled. For this even number $n$, let $C_n$ denote the set $\{(i,j)\mid i,j\in [n],i<j,i+j=n+1\}$ and write $\tilde{C}_n$ for the set $C_n\cup\{(0,0)\}$.
We assume a natural, efficient enumeration of all elements in $C_n$.
We define $IND_n$ to be $\tilde{C}_n \times ( \tilde{\Sigma}\times [0,n+1]_{\integer} )^2$, where $\tilde{\Sigma}$ denotes $\Sigma \cup \{B,acc,rej\}$ and $\{B,acc,rej\}$ is composed of new symbols not in $\Sigma$. Moreover, we set $\mu(x)$ to be $|x|$ for all $x\in\Sigma^*$.
The desired AEQS $\SSS$ has size $m(x) = {\ilog|IND_{\mu(x)}|}$, which is at most ${\ilog|\tilde{C}_{|x|}|} +2({\ilog|\tilde{\Sigma}|} + 2{\ilog(|x|+2)})^2 = O(\log|x|)$.
The desired AEQS $\SSS$ will be designed to encode three pieces of information: the nondeterministic choice of a pair $(i,j)$, the $i$th symbol $\sigma$ of an input $x$, and the current tape head location $h$.

Letting $\xi_0 =(B,0)$, we define
$H^{(x)}_{ini}$ as $\sum_{u\in IND^{(-)}_{\mu(x)}}\density{\hat{u}}{\hat{u}}$ using $IND^{(-)}_{n} =
IND_n -\{(0,0,\xi_0,\xi_0)\}$.
Concerning the final Hamiltonians $\{H^{(x)}_{fin}\}_{x\in\Sigma^*}$, by contrast, we want to design a 1moqqaf $\MM=\{M_n\}_{n\in\nat}$ to generate it. Let $M_n$ have  the form $(Q^{(n)},\Sigma,\{\cent,\dollar\}, \{A^{(n)}_{\sigma}\}_{\sigma\in\check{\Sigma}}, \Lambda^{(n)}_0)$ and assume that $\{A^{(n)}_{\sigma}\}_{\sigma\in\check{\Sigma}}$ is characterized by a set $\KK_n$ of Kraus operators, including  $K^{(n)}_{\sigma,1}$ and $K^{(n)}_{\sigma,2}$ for each endmarker $\sigma\in \{\cent,\dollar\}$ and $U^{(n)}_{\sigma}$ for each symbol $\sigma\in\Sigma$.
The initial mixture $\Lambda^{(n)}_0$ is set to be  $I-\frac{1}{3}\density{0,0,\xi_0,\xi_0}{0,0,\xi_0,\xi_0}$.
During the construction of $H^{(x)}_{fin}$ that follows shortly, we wish to meet the following requirement: if $x_{(i)}=x_{(j)}$ for a certain index pair $(i,j)\in C_n$, then the ground state must have the form $\qubit{i,j,acc,0,\xi_0}$; otherwise, it has the form $\qubit{0,0,B,0,\xi_0}$.
Since there may be multiple witnesses $(i,j)$ in $C_n$ satisfying $x_{(i)}=x_{(j)}$, we need to differentiate all such witnesses by assigning different energy levels to them.

Starting with an arbitrary quantum state, $M_n$ uses the Kraus operators in $\KK_n$ to check whether the $i$th symbol and the $j$th symbol of $x$ are indeed equal.
Let $u=(\sigma,h)$ and $w=(\tau,l)$. In scanning $\cent$, we apply $K^{(n)}_{\cent,1}\qubit{i,j,u,w} = \measure{\xi_0}{u}\cdot \qubit{i,j,u,w}$ and $K_{\cent,2}^{(n)}\qubit{i,j,u,w} = \sum_{s:s\neq\xi_0} \measure{s}{u} \cdot \qubit{i,j,w,u}$. This first step helps us fixate the content of the last five registers to $\qubit{\xi_0}$ for the subsequent argument.
Hereafter, we assume that the last five registers contain only  $\qubit{\xi_0}$.
In the case of $(i,j)\in C_n$, for any two symbols $\sigma\in\Sigma$ and  $\tau\in\Sigma\cup\{B\}$ and for any tape head location $h\in[0,n]_{\integer}$, we set   $U^{(n)}_{\sigma}\qubit{i,j,B,i}\qubit{\xi_0} = \qubit{i,j,\sigma,i+1}\qubit{\xi_0}$,   $U^{(n)}_{\sigma}\qubit{i,j,\sigma,j}\qubit{\xi_0} = \qubit{i,j,acc,j+1}\qubit{\xi_0}$, and  $U^{(n)}_{\sigma}\qubit{i,j,\tau,h}\qubit{\xi_0} = \qubit{i,j,\tau,h+1}\qubit{\xi_0}$ if either $h\notin\{i,j\}$ or $\sigma\neq \tau$.
In scanning $\dollar$, we make transitions   $K^{(n)}_{\dollar,1}\qubit{i,j,acc,n+1}\qubit{\xi_0} = \sqrt{\frac{i}{n+1}} \qubit{i,j,acc,0}\qubit{\xi_0}$ and $K^{(n)}_{\dollar,2}\qubit{i,j,acc,n+1}\qubit{\xi_0} = \sqrt{\frac{n-i+1}{n+1}}\qubit{i,j,rej,0}\qubit{\xi_0}$.
In the case of $(i,j)=(0,0)$, on the contrary, we start with $\qubit{0,0,B,0}\qubit{\xi_0}$ and apply $K^{(n)}_{\cent,1}$,  $U^{(n)}_{\sigma}\qubit{0,0,B,h}\qubit{\xi_0} = \qubit{0,0,B,h+1}\qubit{\xi_0}$, and $K^{(n)}_{\dollar,1}\qubit{0,0,B,n+1}\qubit{\xi_0} = \qubit{0,0,B,0}\qubit{\xi_0}$.

With the use of the Klaus operators in $\KK_n$,
$\{A^{(n)}_{\sigma}\}_{\sigma\in\check{\Sigma}}$ is defined by  $A^{(n)}_{\tau}(H) = \sum_{e\in\{1,2\}} K^{(n)}_{\tau,e} H (K^{(n)}_{\tau,e})^{\dagger}$ for any $\tau\in\{\cent,\dollar\}$ and $A^{(n)}_{\sigma}(H) = U^{(n)}_{\sigma} H (U^{(n)}_{\sigma})^{\dagger}$ for any $\sigma\in\Sigma$.
Given two indices $e_1,e_2\in\{1,2\}$, we succinctly write $V^{(n)}_{\cent{x}\dollar,e_1,e_2}$ for $K^{(n)}_{\dollar,e_2}U^{(n)}_{x}K^{(n)}_{\cent,e_1}$.
It then follows that $A^{(n)}_{\cent{x}\dollar}(H) = \sum_{e_1,e_2\in\{1,2\}} V^{(n)}_{\cent{x}\dollar,e_1,e_2} H (V^{(n)}_{\cent{x}\dollar,e_1,e_2})^{\dagger}$.
The desired $H^{(x)}_{fin}$ is finally defined to be  $A^{(n,x)}_{\cent x\dollar}(\Lambda^{(n)}_0)$.
For the acceptance/rejection criteria pair,  we set  $S^{(n)}_{acc} = \{(i,j,acc,0,\xi_0)\mid (i,j)\in C_n\}$ and
$S^{(n)}_{rej} = \{(0,0,B,0,\xi_0)\}$.

Finally, we intend to prove that $\SSS$ correctly solves $SymCoin$. Let $x$ be any input and set $n=\mu(x)$. If $x$ is in $SymCoin$, then there exists a pair $(i,j)\in C_n$ for which   $x_{(i)}=x_{(j)}$ holds and $i$ is the smallest number. We take the quantum state
$\qubit{\phi_x} = \qubit{i,j,acc,0,\xi_0}$. Since $V^{(n)}_{\cent x\dollar,1,1}\qubit{i,j,\xi_0,\xi_0} = \sqrt{\frac{i}{n+1}}\qubit{\phi_x}$,
it follows that $A^{(n)}_{\cent x\dollar}(\Lambda^{(n)}_0) \qubit{\phi_x} = V_{\cent x\dollar,1,1} \Lambda^{(n)}_0 (V^{(n)}_{\cent x\dollar,1,1})^{\dagger} \qubit{\phi_x} = \sqrt{\frac{i}{n+1}} V^{(n)}_{\cent x\dollar,1,1}\Lambda^{(n)}_0\qubit{i,j,\xi_0,\xi_0} = \frac{i}{n+1}\qubit{\phi_x}$; thus, we obtain $H^{(x)}_{fin}\qubit{\phi_x}= \frac{i}{n+1}\qubit{\phi_x}$.
The minimality of $i$ implies that $\qubit{\phi_x}$ is the ground state and its ground energy is $\frac{i}{n+1}$, which is smaller than $\frac{1}{2}$ since $i\leq \frac{n}{2}$. The spectral gap is therefore at least $\frac{1}{n+1}$.
Obviously, $\qubit{\phi_x}$ is an accepting quantum state in  $QS^{(m(x))}_{acc}$.
In the case of $x\notin SymCoin$, by contrast, let us consider $\qubit{\psi_x} = V^{(n)}_{\cent{x}\dollar,1,1}\qubit{0,0,\xi_0,\xi_0}$. From $\Lambda^{(n)}_0\qubit{0,0,\xi_0,\xi_0} = \frac{2}{3}\qubit{0,0,\xi_0,\xi_0}$, we conclude that
$A^{(n)}_{\cent x\dollar}(\Lambda^{(n)}_0) \qubit{\psi_x} = V^{(n)}_{\cent x\dollar,1} \Lambda^{(n)}_0 (V^{(n)}_{\cent x\dollar,1})^{\dagger} V^{(n)}_{\cent x\dollar,1} \qubit{0,0,\xi_0,\xi_0} =  V^{(n)}_{\cent x\dollar,1}\Lambda^{(n)}_0\qubit{0,0,\xi_0,\xi_0} = \frac{2}{3}\qubit{\psi_x}$; therefore,  $H^{(x)}_{fin}\qubit{\psi_x}=\frac{2}{3}\qubit{\psi_x}$ follows.
Since no computation produces $\qubit{i,j,acc,0,\xi_0}$,
all other eigenstates have eigenvalues of $1$. Therefore,  $\qubit{\psi_x}$ is the ground state with a ground energy of $\frac{2}{3}$.
Notice that $\qubit{\psi_x}$ falls into $QS^{(m(x))}_{rej}$.


\begin{example}
Consider all strings $x$ over the ternary alphabet $\Sigma=\{0,1,\#\}$ that satisfy the following \emph{promise}: $x$ is of the form $0^m \# 1^{n_1}\# 1^{n_2}\# \cdots \# 1^{n_k}$ with $k,m\in\nat^{+}$ and $n_1,\ldots,n_k\in\nat^{+}$ and there exits at most one subset $A$ of $[k]$ for which $m$ equals $\sum_{i\in A}n_i$.
Let $\mathcal{USUBSUM}$ (unary subset sum) denote the promise problem $(USubSum,nonUSubSum)$, where $USubSum$ consists of all promised strings $x$ satisfying $m=\sum_{i\in A}n_i$ for a certain subset $A\subseteq[k]$ and $nonUSubSum$ contains all promised strings not in $USubSum$. This promise problem $\mathcal{USUBSUM}$  belongs to  $\aeqs{1qqaf,linsize,constgap}$.
\end{example}

The set $USubSum$ of accepting instances \emph{with no promise} is a \emph{one-counter context-free language}, which is recognized by an appropriate one-way nondeterministic pushdown automaton using a \emph{unary} stack alphabet.
Here, we intend  to construct a conditional AEQS $\SSS$ for the promise problem $\mathcal{USUBSUM}$.
For any promised input $x$ of the form $0^t \# 1^{n_1}\# 1^{n_2}\# \cdots \# 1^{n_k}$, the segments $0^t$ and $1^{n_i}$ of $x$ are conveniently called \emph{blocks} of $x$. Notice that  $|x|=t+k+\sum_{i\in[k]}n_i$.

Given an arbitrary promised input $x\in\Sigma^*$, let us consider a nondeterministic choice of $d$ blocks, say, $(1^{n_{i_1}},1^{n_{i_2}},\ldots,1^{n_{i_d}})$ with $1\leq i_1<i_2<\cdots <i_d\leq k$.
We associate $x$ with $2^k$ binary strings $s=s_1s_2\cdots s_k$ of length $k$ and define $C_s$ to be $\{i\in[k]\mid s_i=1\}$ so that, if $C_s = \{i_1,i_2,\ldots,i_d\}$, then the tuple  $(1^{n_{i_1}},1^{n_{i_2}},\ldots,1^{n_{i_d}})$ represents a series of nondeterministic choices of blocks.
To blocks $0^t$ and $1^{n_i}$, we respectively assign values $+t$ (positive number) and $-n_i$ (negative number) and try to calculate the sum of the assigned values of the blocks $1^t$ and $1^{n_i}$ for any index $i\in C_s$ step by step. Notice  that this sum equals $0$ exactly when $x$ belongs to $USubSum$.

As for our selector $\mu$, $\mu(x)$ denotes the encoding $\pair{t,k,l}$ of three values $t$, $k$, and $l=\max_{i\in[k]}\{n_i\}$ for any promised string $x$. 
For readability, we hereafter express $\pair{t,k,l}$ as the three separate parameters $t$, $k$, and $l$.
We define the basis of our evolution space, $IND_{t,k,l}$, to be   $\{0,1\}^{k} \times ( [0,k]_{\integer} \times [-kl,t]_{\integer} )^2$ and set $IND_{t,k,l}^{(-)}$ to be $IND_{t,k,l} -\{(0^k,0,0,0,0)\}$.
The size $m(x)$ of $\SSS$ is thus ${\ilog|IND_{t,k,l}|}$, which is at most $k  + 2{\ilog(k+1)}+ 2{\ilog(kl+t+1)} = O(|x|)$.
For each element $(s,i,j,a,b)$ of $IND_{t,k,l}$, the parameter $s$ refers to a series of nondeterministic choices of blocks, $i$ refers to a block number, $j$ refers to the sum of the assigned values of blocks  that have been already calculated, and $(a,b)$ refers to a starting value of the pair $(i,j)$. For convenience,
we write $\xi_0$ for $(0,0)$.
An acceptance/rejection criteria pair $(S^{(m(x))}_{acc},S^{(m(x))}_{rej})$ associated with the input $x$ is defined as $S^{(m(x))}_{acc} = \{(s,k,0,\xi_0)\mid  s\in\{0,1\}^k, s\neq 0^k\}$ and $S^{(m(x))}_{rej} = \{(0^k,k,0,\xi_0)\}$.

The mixture $\sum_{u\in IND_{t,k,l}^{(-)}} \density{\hat{u}}{\hat{u}}$ defines the desired initial Hamiltonian $H^{(x)}_{ini}$ of $\SSS$.
For the desired final Hamiltonian, we need to construct
an appropriate 1moqqaf $\MM=\{M_n\}_{n\in\nat}$ whose elements $M_n$ are  of the form $(Q^{(n)},\Sigma,\{\cent,\dollar\}, \{A^{(n)}_{\sigma}\}_{\sigma\in\check{\Sigma}},  \Lambda^{(n)}_0,Q^{(n)}_0)$, where $\{A^{(n)}_{\sigma}\}_{\sigma\in\check{\Sigma}}$ is characterized by a certain set $\KK_n = \{K^{(n)}_{\cent,1},K^{(n)}_{\cent,2}\}\cup\{U_{\sigma}\mid \sigma\in\Sigma_{\dollar}\}$ of Kraus operators, where  $\Sigma_{\dollar}$ denotes $\Sigma\cup\{\dollar\}$. 
Letting $n=\pair{t,k,l}$, we set $Q^{(n)}=IND_{t,k,l}$ and define $Q^{(n)}_0$ as the set $\{(s,k,0,0,0)\mid s\in \{0,1\}^k\}$.
The initial mixture $\Lambda^{(n)}_0$ is  $I-\frac{1}{2}\density{0^k,0,0,\xi_0}{0^k,0,0,\xi_0}$.
In what follows, we describe how to define the quantum operators in $\KK_n$.

Let $x$ be any promised input of the form $0^t\# 1^{n_1}\# \cdots \# 1^{n_k}$ associated with the parameter triplet $(t,k,l)$. 
Fix an arbitrary string $s\in\{0,1\}^k$ and consider the associated set $C_s$.
We start with an arbitrary quantum state and apply two Kraus operators $K^{(n)}_{\cent,1}$ and $K^{(n,s)}_{\cent,2}$ defined by $K^{(n,s)}_{\cent,1} \qubit{i,j,a,b} = \measure{0,0}{i,j}\cdot \qubit{i,j,a,b}$ and $K^{(n,s)}_{\cent,2}\qubit{i,j,a,b} = \sum_{(f,g)\neq (0,0)} \measure{f,g}{i,j}\cdot \qubit{a,b,i,j}$.
These operators help us fix the starting quantum state to be  $\qubit{0,0,\xi_0}$ in the subsequent computation by applying $U^{(n,s)}_{\sigma}\qubit{i,j,a,b} =  \qubit{i,j,a,b}$ for any symbol $\sigma\in\Sigma_{\dollar}$ and
for any quantum state $\qubit{i,j,a,b}$ with $(a,b)\neq\xi_0$.
We then define the Kraus operator  $K^{(n)}_{\cent,e}$ to be $\sum_{s} (\density{s}{s}\otimes K^{(n,s)}_{\cent,e})$ for each index $e\in\{1,2\}$.
We further define the remaining unitary operators $U^{(n)}_{\sigma}$, which has the form $\sum_{s\in\{0,1\}^n}  (\density{s}{s}\otimes U^{(n,s)}_{\sigma})$ for appropriate unitary operators  $U^{(n,s)}_{\sigma}$.
While reading the first block $0^t$ of $x$, we increase the value of the second  register by making a transition given by
$U^{(n,s)}_{0}\qubit{0,j,\xi_0} = \qubit{0,j+1,\xi_0}$ for any $j\in[0,t-1]_{\integer}$.
To figure out the block number, whenever we read $\#$, we increase the value of the first register by  applying
$U^{(n,s)}_{\#}\qubit{i,j,\xi_0} =  \qubit{i+1,j,\xi_0}$  for any $i\in [0,k-1]_{\integer}$ and move the tape head to the right.
In the case of $i\notin C_s$, we make a transition  $U^{(n,s)}_{1}\qubit{i,j,0,0} = \qubit{i,j,\xi_0}$.
Otherwise, as we read each symbol $1$ of the $(i+1)$th block $1^{n_i}$, we decrease the value of the second register by one.
This process can be done by applying $U^{(n,s)}_{1}\qubit{i,j,\xi_0} = \qubit{i,j-1,\xi_0}$ for any $i\in C_s$ and $j\in[-kl,t]_{\integer}$.
In the end, when we scan the endmarker $\dollar$, we apply
$U^{(n,s)}_{\dollar}\qubit{i,j,a,b} = \qubit{i,j,a,b}$.
For each index $e\in\{1,2\}$, we briefly write  $V^{(n)}_{\cent{x}\dollar,e}$
for $U^{(n)}_{x\dollar}  K^{(n)}_{\cent,e}$.
If $s\neq0^k$, then we obtain $V^{(n)}_{\cent x\dollar,1}\qubit{s,0,0,\xi_0} = \qubit{s,k,0,\xi_0}$.
On the contrary, when $s=0^k$, $V^{(n)}_{\cent x\dollar,1}\qubit{s,0,0,\xi_0} = \qubit{s,k,t,\xi_0}$ follows.
We then observe a positive number in the third register by performing  the projective measurement $\Pi^{(n)}_0$,
which is induced from $Q^{(n)}_0$.

Finally, we define $A^{(n)}_{\cent}(H) = \sum_{e\in\{1,2\}} K^{(n)}_{\cent,e}H(K^{(n)}_{\cent,e})^{\dagger}$ and $A^{(n)}_{\sigma}(H) = U^{(n)}_{\sigma}H(U^{(n)}_{\sigma})^{\dagger}$ for  any symbol $\sigma\in \Sigma\cup\{\dollar\}$, where $H$ is an arbitrary linear operator.
The final Hamiltonian $H^{(x)}_{fin}$ is set to be $\Pi^{(n)}_0A^{(n)}_{\cent{x}\dollar}(\Lambda^{(n)}_0)\Pi^{(n)}_0$.

Let us argue that $\SSS$ correctly solves $\mathcal{USUBSUM}$ with  accuracy $1$.
Given a promised input $x$ of the aforementioned form, if $x\in USubSum$, then there exists a \emph{unique} witness $s\in\{0,1\}^k$ for which $t=\sum_{i\in C_s}n_i$. Let $n=\pair{t,k,l}$. 
Since $|C_s|\geq1$, $s\neq 0^k$ follows.
For this particular string $s$, we
consider the quantum state $\qubit{\phi_s} = \qubit{s,k,0,\xi_0}$, which
equals $V^{(n)}_{\cent{x}\dollar,1}\qubit{s,0,0,\xi_0}$.
Since $\Pi^{(n)}_0\qubit{\phi_s} =  0$, it then follows that $H^{(x)}_{fin}\qubit{\phi_s} = \Pi^{(n)}_0A^{(n)}_{\cent{x}\dollar}(\Lambda^{(n)}_0)\Pi^{(n)}_0 \qubit{\phi_s} = 0$.
Therefore, $\qubit{\phi_s}$ is the ground state and its ground energy is  $0$. The uniqueness of the ground state is guaranteed by the uniqueness of $s$.  Moreover,
$\qubit{\phi_s}$ belongs to $QS^{(m(x))}_{acc}$.
On the contrary, when $x\in nonUSubSum$, we take another quantum state  $\qubit{\psi_x} = \qubit{0^k,0,0,\xi_0}$. Notice that $\qubit{\psi_x} = V^{(n)}_{\cent{x}\dollar,1}\qubit{0^k,0,0,\xi_0}$.
Since  $\Pi^{(n)}_0\qubit{\psi_x} = \qubit{\psi_x}$ and
$\Lambda^{(n)}_0\qubit{0^k,0,0,\xi_0} = \frac{1}{2}\qubit{0^k,0,0,\xi_0}$, we conclude that $H^{(x)}_{fin}\qubit{\psi_x} = \Pi^{(n)}_0 V^{(n)}_{\cent x\dollar,1} \Lambda^{(n)}_0 (V^{(n)}_{\cent x\dollar,1})^{\dagger} \Pi^{(n)}_0 \qubit{\psi_x} = \frac{1}{2} V^{(n)}_{\cent{x}\dollar,1}\qubit{0^k,0,0,\xi_0} = \frac{1}{2}\qubit{\psi_x}$. Since there is no string $s\in\{0,1\}^k$  witnessing $t=\sum_{i\in C_s}n_i$, all other eigenvalues must be $1$; therefore, the ground energy is $\frac{1}{2}$.
Clearly, $\qubit{\psi_x}$ belongs to $QS^{(m(x))}_{rej}$.
By the above argument, we conclude that $\SSS$ correctly solves $\mathcal{USUBSUM}$.


\begin{example}\label{example:MULTIDUP}
Let us consider all strings over the alphabet $\Sigma=\{0,1,\#\}$ that satisfy the following \emph{promise}: (i) $x$ is of the form $w_0\# w_1\# w_2 \cdots \# w_k$ with $k\in\nat^{+}$ and $w_0,w_1,w_2,\ldots, w_k\in\{0,1\}^*$,  (ii) $|w_i|=|w_j|>0$ for any pair $i,j\in[0,k]_{\integer}$, and (iii) there exists at most one index $i\in[k]$ for which  $w_i$ is different from the rest. 
Consider the promise problem $\mathcal{MULTDUP} = (MultDup,nonMultDup)$, where $MultDup$ consists of all promised strings satisfying $w_i=w_j$ for any pair $i,j\in[0,k]_{\integer}$ and $nonMultDup$ contains all promised strings not in $MultDup$.
This promise problem $\mathcal{MULTDUP}$ (multiple duplication) falls into  $\aeqs{1qqaf,logsize,constgap}$.
\end{example}

Notice that the set $nonMultDup$ of rejecting instances consists of strings $w_0\#w_1\#\cdots \# w_k$ satisfying $(w_0)_{(j)} \neq (w_i)_{(j)}$
for two appropriate indices $i\in[k]$ and $j\in[l]$, where $l=|w_0|$.
If there is \emph{no promise}, nonetheless, the set $MultDup$ of accepting instances is context-sensitive but not context-free.
In what follows, we plan to explain how an appropriately chosen AEQS $\SSS$ can solve the promise problem $\mathcal{MULTDUP}$ with accuracy $1$.

For each promised string $x$, we take two parameters $k$ and $l$ explained in the example and set $\mu(x)=\pair{k,l}$. Given such a pair $(k,l)$,
$IND_{\pair{k,l}}$ denotes the index set $[0,k]_{\integer} \times [0,l]_{\integer}  \times ( \Sigma_{B} \times [0,k]_{\integer} \times [0,l]_{\integer} \times [0,N+1]_{\integer} )^2$, where $N=l(k+1)+k$ and  $\Sigma_{B} = \Sigma\cup \{B\}$ with a new symbol $B$ not in $\Sigma$.
For simplicity, we intend to write $IND_{k,l}$ in place of $IND_{\pair{k,l}}$.
An element $(i,j,\sigma,h,r,t,\sigma',h',r',t')$ of $IND_{k,l}$ roughly means that we are scanning the $r$th symbol of $w_i$, $\sigma$ is the $j$th symbol of $w_0$, $h$ is the number of \#s that we have already passed,  $t$ is a tape head location, and $(\sigma',h',r',t')$ refers to a starting value of $(\sigma,h,r,t)$. 
We further set $m(x)$ to be ${\ilog|IND_{\mu(x)}|}$, which is obviously $O(\log|x|)$. 
We succinctly write $C$ for $[k]\times[l]$ and $\xi_0$ for $(B,0,0,0)$.
Moreover, let $S^{(n)}_{acc} =\{(i,j,B,k,0,0,\xi_0)\mid (i,j)\in C\}$ and $S^{(n)}_{rej} = \{(0,0,B,k,0,0,\xi_0)\}$.

Concerning the desired set $\{H^{(x)}_{ini}\}_{x\in\Sigma^*}$  of $\SSS$'s initial Hamiltonians, for each string $x\in\Sigma^*$, we define $H^{(x)}_{ini}$ to be $\sum_{u\in IND_{\mu(x)}^{(-)}} \density{\hat{u}}{\hat{u}}$, where $IND_{k,l}^{(-)}= IND_{k,l} -\{(0,0,\xi_0,\xi_0)\}$. To define the set $\{H^{(x)}_{fin}\}_{x\in\Sigma^*}$  of $\SSS$'s final Hamiltonians, we  hereafter introduce an appropriate 1moqqaf, say, $\MM=\{M_n\}_{n\in\nat}$.

Let $x$ be an arbitrary promised input associated with $(k,l)$ and set $n=\pair{k,l}$ for brevity.
Let $M_n = (Q^{(n)},\Sigma,\{\cent,\dollar\}, \{A^{(n)}_{\sigma}\}_{\sigma\in\check{\Sigma}}, \Lambda^{(n)}_0,Q^{(n)}_0)$.
Under the promise, $x$ has the form $w_0\#w_1\# \cdots \# w_k$ with  $l=|w_i|>0$ for any index $i\in[0,k]_{\integer}$.
We split the computation of $M_n$ into two parts.
Our goal is to make the ground state of $H^{(x)}_{fin}$ have the form $\qubit{i,j,B,k,0,0}\qubit{\xi_0}$ for a certain pair $(i,j)\in C\cup\{(0,0)\}$.
Let us consider an arbitrary pair $(i,j)\in C$ and let
$u$ and $w$ express two tuples $(\sigma,a,b,c)$ and $(\sigma',a',b',c')$, respectively. In reading $\cent$, we apply Kraus operators $K^{(n)}_{\cent,1}$ and $K^{(n)}_{\cent,2}$, which are defined as $K^{(n)}_{\cent,1}\qubit{i,j,u,w} = \measure{\xi_0}{u}\cdot \qubit{i,j,u,w}$ and $K^{(n)}_{\cent,2}\qubit{i,j,u,w} = \sum_{z:z\neq \xi_0} \measure{z}{u}\cdot \qubit{i,j,w,u}$.
We set $A^{(n)}_{\cent}(H) = \sum_{e\in\{1,2\}} K^{(n)}_{\cent,e}H(K^{(n)}_{\cent,e})^{\dagger}$.
This first step helps us  concentrate only on quantum states of the form $\qubit{i,j,u,\xi_0}$. This is possible because, in the case where the last four registers do not contain $\qubit{\xi_0}$, we force the machine to apply the identity operator $I$.

We assume that the computation of $M_n$ begins with the quantum state $\qubit{i,j,\xi_0}\qubit{\xi_0}$. We first move $M$'s tape head to the right, locate the $j$th symbol of $w_0$ using the $5$th register as a counter, and remember it by creating the quantum state  $\qubit{i,j,(w_0)_{(j)},0,j,j+1}\qubit{\xi_0}$.
Using the $4$th register as another counter, we count the number of $\#$s to find the beginning of the $i$th block $w_i$. By counting the number of symbols in this particular block, we further locate the $j$th symbol $(w_i)_{(j)}$ by decreasing the first counter.
Assuming that the current tape head location is $h$, we check whether this symbol $(w_i)_{(j)}$ is different from $(w_0)_{(j)}$. If this is truly the case, then we create the quantum state $\qubit{i,j,B,k,0,h+1}\qubit{\xi_0}$; otherwise, we create  $\qubit{i,j,(w_0)_{(j)},k,0,h+1}\qubit{\xi_0}$ instead.
When we finally reach $\dollar$, we change $\qubit{i,j,\tau,k,0,n+1}\qubit{\xi_0}$ to $\qubit{i,j,\tau,k,0,0}\qubit{\xi_0}$ for any $\tau\in\Sigma_{B}$.
We then observe a non-blank symbol in the third register by performing   the projective measurement $\Pi^{(n)}_0$, which is induced from $Q^{(n)}_0  = \{(i,j,B,k,0,0,\xi_0)\mid (i,j)\in C\cup\{(0,0)\}\}$.

Concerning $A^{(n)}_{\sigma}$ for each symbol $\sigma\in \Sigma\cup\{\dollar\}$, we set $A^{(n)}_{\sigma}(H) = U^{(n)}_{\sigma}H(U^{(n)}_{\sigma})^{\dagger}$.
For readability, we abbreviate $U^{(n)}_{x\dollar} K^{(n)}_{\cent,e}$ as
$V^{(n)}_{\cent{x}\dollar,e}$ for each index $e\in\{1,2\}$. It then follows that $A^{(n)}_{\cent{x}\dollar}(H) = \sum_{e\in\{1,2\}} V^{(n)}_{\cent{x}\dollar,e} H (V^{(n)}_{\cent{x}\dollar,e})^{\dagger}$.
Finally, we set $\Lambda^{(n)}_{0}= I - \frac{1}{2}\density{0,0,\xi_0}{0,0,\xi_0}$ and define $H^{(x)}_{fin}$
to be $\Pi^{(n)}_0A^{(n)}_{\cent x\dollar} (\Lambda^{(n)}_0)\Pi^{(n)}_0$.

We still need to show that $\SSS$ correctly solves $\mathcal{MULTDUP}$ with accuracy $1$. Let us consider the first case where $x$ is in $nonMultDup$ with two parameters $(k,l)$ and choose a \emph{unique} pair $(i,j)\in C$ satisfying $(w_0)_{(j)} \neq (w_i)_{(j)}$. 
In this case, we take the quantum state $\qubit{\phi_x}= \qubit{i,j,B,k,0,0,\xi_0}$, which belongs to $QS^{(m(x))}_{acc}$. Let $n=\pair{k,l}$. Since $\Pi^{(n)}_0\qubit{\phi_x}=0$,
we obtain $\Pi^{(n)}_0 A^{(n)}_{\cent x\dollar}(\Lambda^{(n)}_0)\Pi^{(n)}_0 \qubit{\phi_x} = 0$. 
This fact implies that $H^{(x)}_{fin}\qubit{\phi_x} =0$, and thus $\qubit{\phi_x}$ is the ground state of energy $0$.
On the contrary, let us consider the case where $x$ belongs to $MultDup$.
We then pick $\qubit{\psi_x} =
\qubit{0,0,B,k,0,0,\xi_0}$ in $QS^{(m(x))}_{rej}$.
Note that $V^{(n)}_{\cent x\dollar,1}\qubit{0,0,\xi_0,\xi_0} = \qubit{\psi_x}$ and $\Pi^{(n)}_0\qubit{\psi_x}=\qubit{\psi_x}$.
A simple calculation concludes that
$\Pi^{(n)}_0A^{(n)}_{\cent x\dollar}(\Lambda^{(n)}_0)\Pi^{(n)}_0 \qubit{\psi_x}
= \Pi^{(n)}_0 V^{(n)}_{\cent{x}\dollar,1} \Lambda^{(n)}_0 (V^{(n)}_{\cent{x}\dollar,1})^{\dagger}\Pi^{(n)}_0 \qubit{\psi_x}
= \Pi^{(n)}_0 V^{(n)}_{\cent x\dollar,1}\Lambda^{(n)}_0 \qubit{0,0,\xi_0,\xi_0}
= \frac{1}{2} \Pi^{(n)}_0 V^{(n)}_{\cent{x}\dollar,1} \qubit{0,0,\xi_0,\xi_0}
= \frac{1}{2}\qubit{\psi_x}$; thus, $H^{(x)}_{fin}\qubit{\psi_x}= \frac{1}{2}\qubit{\psi_x}$ follows.
Since all basis states other than $\qubit{\psi_x}$ have value $1$ in $\Lambda^{(n)}_0$, $\qubit{\psi_x}$ must be a unique ground state and its ground energy is $\frac{1}{2}$.
We thus conclude that $\SSS$ solves
$\mathcal{MULTDUP}$ with accuracy $1$.

\subsection{Structural Properties of AEQSs}\label{sec:struct-properties}

We have already exemplified in Section \ref{sec:example} the computational power of conditional AEQSs for six languages.
We further explore the structural properties
of the conditional AEQSs.

The first consideration is the closure properties under unary and binary operations.
Given a binary operation $\circ$ acting on two languages, a language family $\LL$ is said to be \emph{closed under} $\circ$ if, for any two languages $L_1,L_2\in\LL$ over the same alphabet, $L_1\circ L_2$ also belongs to $\LL$. For a unary operator $\circ$, we similarly define the closure property of $\LL$ under this operator $\circ$.
In what follows, we wish to discuss such properties for conditional AEQSs particularly under \emph{XOR} and  \emph{complementation}.
For any given nonempty condition set $\FF$ dictating the behaviors of  AEQSs, we say that $\FF$ \emph{allows a swap of acceptance/rejection criteria} if the new AEQS $\SSS'$ obtained from each AEQS $\SSS$ satisfying $\FF$ by exchanging   $\{S^{(n)}_{acc}\}_{n\in\nat}$ and $\{S^{(n)}_{rej}\}_{n\in\nat}$
in the family $\{(S^{(n)}_{acc},S^{(n)}_{rej})\}_{n\in\nat}$ of acceptance/rejection criteria pairs of $\SSS$ also satisfies $\FF$.
Additionally, we say that $\FF$ \emph{allows accuracy amplification} if, for any AEQS satisfying $\FF$ with accuracy at least $\varepsilon\in(1/2,1)$ and any constant $c$ with $0<c\varepsilon<1$, there always exists another computationally-equivalent AEQS with $\FF$ whose accuracy is at least $c\varepsilon$.

\begin{proposition}\label{properties}
Let $\FF$ be any nonempty set of conditions. Each of the following statements holds.
\renewcommand{\labelitemi}{$\circ$}
\begin{enumerate}\vs{-2}
  \setlength{\topsep}{-2mm}%
  \setlength{\itemsep}{1mm}%
  \setlength{\parskip}{0cm}%

\item $\aeqs{\FF}$ is closed under complementation if $\FF$ allows a swap of acceptance/rejection criteria.

\item $\aeqs{\FF}$ is closed under XOR if $\FF$ allows accuracy amplification.
\end{enumerate}
\end{proposition}

\begin{proof}
(1) Assume that a condition set $\FF$ allows a swap of acceptance/rejection criteria. Given any language $L$ in $\aeqs{\FF}$, take an AEQS $\SSS$ satisfying $\FF$ that solves $L$ with high accuracy. Let $\{(S^{(n)}_{acc},S^{(n)}_{rej})\}_{n\in\nat}$ denote a family of acceptance/rejection criteria pairs of $\SSS$. We define a new AEQS $\SSS'$ from $\SSS$ by exchanging the roles of $S^{(n)}_{acc}$ and $S^{(n)}_{rej}$. By the proposition's assumption, $\SSS'$ also satisfies $\FF$.
Obviously, $\SSS'$ correctly solves the complement $\overline{L}$ with the same accuracy as $\SSS$.
Therefore, $\overline{L}$ belongs to $\aeqs{\FF}$.
This shows the closure of $\aeqs{\FF}$ under complementation.

(2) To express the XOR operation between two languages $L_1$ and $L_2$ over a common alphabet $\Sigma$, we use the notation $L_1\oplus L_2$, which equals $\{x\in\Sigma^*\mid \text{either $x\in \overline{L}_1\cap L_2$ or $x\in L_1\cap\overline{L}_2$} \}$.
Let us consider two AEQSs $\SSS_1$ and $\SSS_2$ satisfying $\FF$ for $L_1$ and $L_2$, respectively.
Assume that, for each index $j\in\{1,2\}$, $\SSS_j$ has the form $(m,\Sigma,\varepsilon_j,\{H_{j,ini}^{(x)}\}_{x\in\Sigma^*}, \{H_{j,fin}^{(x)}\}_{x\in\Sigma^*}, \{S_{j,acc}^{(n)}\}_{n\in\nat},\{S_{j,rej}^{(n)}\}_{n\in\nat})$.
Set $\varepsilon=\min\{\varepsilon_1,\varepsilon_2\}$. Since $\FF$ allows accuracy amplification, without loss of generality, we assume that $\frac{7}{8}<\varepsilon< 1$.

We want to define a new AEQS $\tilde{\SSS} =  (\tilde{m},\Sigma,\tilde{\varepsilon}, \{\tilde{H}^{(x)}_{ini}\}_{x\in\Sigma^*}, \{\tilde{H}^{(x)}_{fin}\}_{x\in\Sigma^*}, \{\tilde{S}^{(n)}_{acc}\}_{n\in\nat}, \{\tilde{S}^{(n)}_{rej}\}_{n\in\nat})$ for $L_1\oplus L_2$ in the following fashion.
Fix a string $x\in\Sigma^*$ arbitrarily. We define $\tilde{H}^{(x)}_{ini} = H^{(x)}_{1,ini}\otimes H^{(x)}_{2,ini}$ and $\tilde{H}^{(x)}_{fin} = H^{(x)}_{1,fin}\otimes H^{(x)}_{2,fin}$.
Furthermore, we define $\tilde{S}^{(n)}_{acc} = (S^{(n)}_{1,acc}\otimes S^{(n)}_{2,rej})\cup (S^{(n)}_{1,rej}\otimes S^{(n)}_{2,acc})$ and $\tilde{S}^{(n)}_{rej} = (S^{(n)}_{1,acc}\otimes S^{(n)}_{2,acc})\cup (S^{(n)}_{1,rej}\otimes S^{(n)}_{2,rej})$.
Associated with $\tilde{S}^{(n)}_{acc}$ and $\tilde{S}^{(n)}_{rej}$,  $\widetilde{QS}^{(n)}_{acc}$ and $\widetilde{QS}^{(n)}_{rej}$ respectively denote the Hilbert spaces spanned by all elements in $\tilde{S}^{(n)}_{acc}$ and those in $\tilde{S}^{(n)}_{acc}$.

For convenience, we write $\eta$ for $\sqrt{2}(1-\varepsilon)$.
Assume that $x\in L_1\oplus L_2$. Consider the case where $x\in L_1$ and $x\notin L_2$. The other case of both $x\notin L_1$ and $x\in L_2$ is symmetrically proven.
Let us assume that $\qubit{\phi_{j,x}}$ is the ground state of $H^{(x)}_{j,fin}$ for each  index $j\in\{1,2\}$. The ground state of $\tilde{H}^{(x)}_{fin}$ is clearly $\qubit{\phi_{1,x}}\otimes \qubit{\phi_{2,x}}$.
We choose two quantum states $\qubit{\psi_{1,x}}\in QS^{(m(x))}_{acc}$ and  $\qubit{\psi_{2,x}}\in QS^{(m(x))}_{rej}$ that are    $\eta$-close to $\qubit{\phi_{1,x}}$ and $\qubit{\phi_{2,x}}$, respectively.
Since $\|\qubit{\phi_{1,x}}\otimes \qubit{\phi_{2,x}} - \qubit{\psi_{1,x}}\otimes \qubit{\psi_{2,x}} \|_2  = \|\qubit{\phi_{1,x}}\otimes(\qubit{\psi_{1,x}}-\qubit{\psi_{2,x}}) + (\qubit{\phi_{1,x}}-\qubit{\phi_{2,x}})\otimes \qubit{\psi_{2,x}}\|_2
\leq \| \qubit{\phi_{1,x}} - \qubit{\phi_{2,x}} \|_2 + \| \qubit{\psi_{1,x}} - \qubit{\psi_{2,x}}\|_2 \leq 2\sqrt{\eta}$, it follows that $\qubit{\phi_{1,x}}\otimes \qubit{\phi_{2,x}}$ is $4\eta$-close to $\qubit{\psi_{1,x}}\otimes \qubit{\psi_{2,x}}$ in $\widetilde{QS}^{(m(x))}_{acc}$.
Similarly, we can deal with the case of $x\notin L_1\oplus L_2$.
The desired accuracy bound $\tilde{\varepsilon}$ is set to be $4\varepsilon-3$.
We then obtain $\frac{1}{2}<\tilde{\varepsilon}< 1$. Since $4\eta=\sqrt{2}(1-\tilde{\varepsilon})$, $\tilde{\SSS}$ solves $L_1\oplus L_2$ with accuracy at least $\tilde{\varepsilon}$.
\end{proof}


Let us consider \emph{inverse images} of functions invariant with system sizes.  Given two functions $f:\Sigma^*\to\Sigma^*$ and $m:\Sigma^*\to\nat$, we say that $f$ is \emph{$m$-preserving} if $m(f(x)) = m(x)$ holds for every string $x\in\Sigma^*$. Let $L_{f^{-1}}$ be the inverse image of $L$ by $f$; that is, $L_{f^{-1}}= \{x\in\Sigma^*\mid f(x)\in L\}$.

\begin{lemma}\label{function-replacement}
Let $L$ be any language over alphabet $\Sigma$, let $m:\Sigma^*\to\nat$ be a function,  and let $f$ be any $m$-preserving function on $\Sigma^*$.
Assume that an AEQS $\SSS = (m,\Sigma,\varepsilon,\{H^{(x)}_{ini}\}_{x\in\Sigma^*}, \{H^{(x)}_{fin}\}_{x\in\Sigma^*}, \{S^{(n)}_{acc}\}_{n\in\nat}, \{S^{(n)}_{rej}\}_{n\in\nat})$ recognizes $L$ with accuracy at least $\varepsilon\in (1/2,1]$. Define $\tilde{H}^{(x)}_{ini}= H^{(f(x))}_{ini}$, $\tilde{H}^{(x)}_{fin}= H^{(f(x))}_{fin}$, $\tilde{S}^{(n)}_{acc} = S^{(n)}_{acc}$, and $\tilde{S}^{(n)}_{rej} = S^{(n)}_{rej}$. Moreover, $\SSS_{f^{-1}}$ denotes the AEQS obtained directly from $\SSS$ by replacing $H^{(x)}_{ini}$,  $H^{(x)}_{fin}$, $S^{(n)}_{acc}$, and $S^{(n)}_{rej}$ respectively with  $\tilde{H}^{(x)}_{ini}$,  $\tilde{H}^{(x)}_{fin}$, $\tilde{S}^{(n)}_{acc}$, and $\tilde{S}^{(n)}_{rej}$. It then follows that $\SSS_{f^{-1}}$ recognizes $L_{f^{-1}}$ with accuracy at least $\varepsilon$.
\end{lemma}

\begin{proof}
As in the premise of the lemma, we take an AEQS $\SSS$ that recognizes  $L$ with accuracy at least $\varepsilon\in(1/2,1]$.
Let us consider $\SSS_{f^{-1}}$ defined from $\SSS$ and $f$.
Our goal is to verify that $S_{f^{-1}}$ indeed recognizes $L_{f^{-1}}$ with accuracy at least $\varepsilon$.
Let $x$ be any input over $\Sigma$. We remark that $\tilde{H}^{(x)}_{ini}$ has a unique ground state because it is also the ground state of $H^{(f(x))}_{ini}$.
when $x$ is in $L_{f^{-1}}$, since $f(x)\in L$, there exists a quantum state $\qubit{\phi_{f(x)}}$ in $QS^{(m(f(x)))}_{acc}$ that is $\sqrt{2}(1-\varepsilon)$-close to the ground state, say,  $\qubit{\psi_{f(x)}}$ of $H^{(f(x))}_{fin}$.
Notice that the Hilbert space $\widetilde{QS}^{(m(x))}_{acc}$ induced from $\tilde{S}^{(m(x))}_{acc}$ matches $QS^{(m(f(x)))}_{acc}$ since $f$ is $\mu$-preserving and $\tilde{S}^{(n)}_{acc} = S^{(n)}_{acc}$ for any $n$.
Therefore, $\qubit{\psi_{f(x)}}$ is also the ground state of $\tilde{H}^{(x)}_{fin}$ and is $\sqrt{2}(1-\varepsilon)$-close to $\qubit{\phi_{f(x)}}$, which also belongs to $\widetilde{QS}^{(m(x))}_{acc}$. As a consequence, $\SSS_{f^{-1}}$ accepts $x$ with accuracy at least $\varepsilon$. The case of $x\notin L_f$ is similarly treated. We thus conclude that $S_{f^{-1}}$ recognizes $L_{f^{-1}}$ with accuracy at least $\varepsilon$.
\end{proof}

Let us consider AEQSs whose Hamiltonians are generated by certain 1moqqaf's, which read extended inputs with the two endmarkers $\cent$ and $\dollar$.   It is, however, possible to remove the right-endmarker $\dollar$ from an input tape of the 1moqqaf's. We call a 1moqqaf with no right-endmarker a \emph{$\dollar$-less 1moqqaf}.

\begin{lemma}\label{remove-dollar}
For any AEQS over alphabet $\Sigma$, if its Hamiltonians are generated by certain 1moqqaf's with the two endmarkers $\cent$ and $\dollar$, then there exists $\dollar$-less 1moqqaf's that also generate the same  Hamiltonians.
\end{lemma}

\begin{proof}
Let $\SSS$ be any AEQS of the form $(m,\Sigma,\varepsilon,\{H^{(x)}_{ini}\}_{x\in\Sigma^*}, \{H^{(x)}_{fin}\}_{x\in\Sigma^*}, \{S^{(n)}_{acc}\}_{n\in\nat}, \{S^{(n)}_{rej}\}_{n\in\nat})$.
Concerning the family $\{H^{(x)}_{fin}\}_{x\in\Sigma^*}$ of final Hamiltonians, we take a polynomially-bounded selector $\mu$ and a 1moqqaf $\MM=\{M_{n}\}_{n\in\nat}$ that generate it.
Fix an arbitrary input $x\in\Sigma^*$ and let $n$ denote $\mu(x)$.
Assume that  $M_{n}$ is of the form $(Q^{(n)}, \Sigma,\{\cent,\dollar\}, \{A^{(n)}_{\sigma}\}_{\sigma\in\check{\Sigma}}, \Lambda^{(n)}_0, Q^{(n)}_0)$ satisfying $H^{(x)}_{fin} = \Pi^{(n)}_0 A^{(n)}_{\cent x\dollar}(\Lambda^{(n)}_0)\Pi^{(n)}_0$, where $\Pi^{(n)}_0$ is induced from $Q^{(n)}_0$. Since $\MM$ is a 1moqqaf, there is a family $\{U^{(n)}_{\sigma}\}_{\sigma\in\check{\Sigma}}$ of unitary operators such that  $A^{(n)}_{\sigma}(H) = U^{(n)}_{\sigma}H(U^{(n)}_{\sigma})^{\dagger}$ for any symbol $\sigma\in\check{\Sigma}$.

Here, we wish to define a new $\dollar$-less 1moqqaf $\tilde{\MM}$ that can generate  $\{H^{(x)}_{fin}\}_{x\in\Sigma^*}$ as well.
For this purpose, we define $\tilde{U}^{(n)}_{\cent} = U^{(n)}_{\dollar}U^{(n)}_{\cent}$ and $\tilde{U}^{(n)}_{\sigma} = U^{(n)}_{\dollar} U^{(n)}_{\sigma} (U^{(n)}_{\dollar})^{\dagger}$ for any symbol $\sigma\in\Sigma$.
Clearly, $\tilde{U}^{(n)}_{\sigma}$ is unitary as well.
We succinctly write $\tilde{M}_n$ to denote the machine $(Q^{(n)},\Sigma, \{\cent\}, \{\tilde{U}^{(n)}_{\sigma}\}_{\sigma\in\Sigma_{\cent}}, \Lambda^{(n)}_0, Q^{(n)}_0)$. Notice that $\tilde{\MM} =\{\tilde{M}_n\}_{n\in\nat}$ is a $\dollar$-less 1moqqaf.
Given a symbol $\sigma\in\Sigma\cup\{\cent\}$, we set  $\tilde{A}^{(n)}_{\sigma}(H) = \tilde{U}^{(n)}_{\sigma}H (\tilde{U}^{(n)}_{\sigma})^{\dagger}$ and define  $\tilde{H}^{(x)}_{fin}$ to be $\Pi^{(n)}_0 \tilde{A}^{(n)}_{\cent{x}\dollar}(\Lambda^{(n)}_0) \Pi^{(n)}_0$. It follows by induction that $A^{(n)}_{\cent{x}\dollar}(H) = U^{(n)}_{\cent{x}\dollar} H (U^{(n)}_{\cent{x}\dollar})^{\dagger} = \tilde{U}^{(n)}_{\cent{x}\dollar} H (\tilde{U}^{(n)}_{\cent{x}\dollar})^{\dagger} = \tilde{A}^{(n)}_{\cent{x}\dollar}(H)$. This guarantees that $\tilde{H}^{(x)}_{fin}$ coincides with $H^{(x)}_{fin}$.

A similar argument works for  $\{H^{(x)}_{ini}\}_{x\in\Sigma^*}$. Therefore, the lemma follows.
\end{proof}

\section{Basic Simulations by AEQSs}\label{sec:simulation}

In Section \ref{sec:example}, we have used various condition sets $\FF$ to construct six examples of conditional AEQSs, $\aeqs{\FF}$. It has become clear that such conditions can be used as a \emph{complexity measure} to classify numerous languages according to how complex to construct AEQSs for the languages.
This section will pursue an idea of making such classification and present sufficient condition sets to characterize four known language families.

\subsection{Simulation of 1moqfa's}

Bounded-error 1moqfa's, which were studied earlier by Moore and Crutchfield \cite{MC00} and by Brodsky and Pippenger \cite{BP02}, may be  considered as the simplest form of quantum finite automata.
Recall from Section \ref{sec:QFA} that $1\mathrm{MOQFA}$ is the language family characterized by bounded-error 1moqfa's.
We give an upper bound on the complexity of $1\mbox{MOQFA}$ in terms of conditional AEQSs. This result may be  compared with Example \ref{example:L_a}, in which
the language $L_a$ not in $1\mathrm{MOQFA}$ falls into  $\aeqs{1moqqaf,logsize,constgap,0\mbox{-}energy}$.

\begin{theorem}\label{upper-lower-bound}
$1\mathrm{MOQFA} \subseteq \aeqs{1moqqaf,constsize,constgap,0\mbox{-}energy}$.
\end{theorem}

\begin{proof}
Let $L$ be an arbitrary language in $1\mathrm{MOQFA}$ over alphabet $\Sigma$ and choose a 1moqfa $M = (Q,\Sigma,\{\cent,\dollar\},\{A_{\sigma}\}_{\sigma\in\check{\Sigma}}, q_0, Q_{acc},Q_{rej})$ recognizing $L$ with error probability at most a certain constant $\varepsilon\in[0,1/2)$.
For simplicity, we assume that $Q$ is conveniently expressed as $\{q_0,q_1,\ldots,q_{2^{k_0}-1}\}$, including the initial inner state $q_0$ of $M$, for a certain constant $k_0\in\nat^{+}$.
Additionally, we take unitary matrices $U_{\sigma}$ satisfying $A_{\sigma}(H) = U_{\sigma}H (U_{\sigma})^{\dagger}$ for any linear operator $H$.

Given an input $x\in\Sigma^*$, we express the extended input $\cent x\dollar$ as $x_0x_1x_2\cdots x_{n+1}$ with $n=|x|$, where $x_0=\cent$, $x_{n+1}=\dollar$, and $x_i\in \Sigma$ for any index $i\in[n]$.
We inductively define $\rho_0=\density{q_0}{q_0}$ and $\rho_{i+1} = U_{x_i}\rho_i U_{x_i}^{\dagger}$ for each index $i\in[0,n+1]_{\integer}$.
Let $\Pi_{acc}$ and $\Pi_{rej}$ be two projections onto the Hilbert spaces $QS_{acc}$ and $QS_{rej}$ spanned by $\{\qubit{q}\mid q\in Q_{acc}\}$ and $\{\qubit{q}\mid q\in Q_{rej}\}$, respectively.
By the choice of $M$ for $L$, it follows that, for any string $x\in L$, $\trace(P_{acc} \rho_{n+2}) \geq 1-\varepsilon$ and, for any $x\in\overline{L}$, $\trace(P_{rej} \rho_{n+2}) \geq 1-\varepsilon$.

To prove that $L$ belongs to $\aeqs{1moqqaf,constsize,constgap,0\mbox{-}energy}$, it suffices to
show how to simulate $M$ by a suitable AEQS, say, $\SSS = (m,\Sigma,\varepsilon,\{H^{(x)}_{ini}\}_{x\in\Sigma^*}, \{H^{(x)}_{fin}\}_{x\in\Sigma^*}, \{S^{(n)}_{acc}\}_{n\in\nat}, \{S^{(n)}_{rej}\}_{n\in\nat})$. For this AEQS $\SSS$,
$H^{(x)}_{ini}$ and $H^{(x)}_{fin}$ are respectively defined to be $W^{\otimes k_0} \Lambda_0 W^{\otimes k_0}$ and $U_{\cent x\dollar} \Lambda_0 U_{\cent x\dollar}^{\dagger}$, where   $\Lambda_0 = \sum_{q\in Q^{(-)}} \density{q}{q}$ with $Q^{(-)}= Q-\{q_0\}$. Notice that $H^{(x)}_{ini}$ and $H^{(x)}_{fin}$ have dimension $2^{k_0}$.

We briefly write $\qubit{\phi_x}$ for $U_{\cent x\dollar} \qubit{q_0}$. Notice that $\density{\phi_x}{\phi_x}$ coincides with $\rho_{n+2}$. By the definition of $H^{(x)}_{fin}$, $\qubit{\phi_x}$ is the ground state of $H^{(x)}_{fin}$ because $H^{(x)}_{fin}\qubit{\phi_x} = U_{\cent x\dollar}\Lambda_0 U^{\dagger}_{\cent x\dollar} U_{\cent x\dollar}\qubit{q_0} = U_{\cent x\dollar}\Lambda_0 \qubit{q_0} =0$.
When $x\in L$, we obtain  $\|\Pi_{acc}\qubit{\phi_x}\|_2^2 = \trace(\Pi_{acc}\rho_{n+2}) \geq 1-\varepsilon$.
This implies that there is a quantum state $\qubit{\phi_{acc}}$ in $QS^{(m(x))}_{acc}$ such that $|\measure{\phi_{acc}}{\phi_x}|^2 \geq 1-\varepsilon$. We then obtain $\|\qubit{\phi_{x}} - \qubit{\phi_{acc}}\|_2^2 \leq
1-|\measure{\phi_{acc}}{\phi_x}|^2\leq \varepsilon$, concluding that  $\qubit{\phi_x}$ is $\sqrt{\varepsilon}$-close to $\qubit{\phi_{acc}}$. Let $\hat{\varepsilon}=1-\sqrt{\frac{\varepsilon}{2}}$. Notice that $1/2<\hat{\varepsilon}\leq 1$ follows from $\varepsilon\in[0,1/2)$.
Since $\sqrt{\varepsilon} =\sqrt{2}(1-\hat{\varepsilon})$, the accuracy of $\SSS$ is at least $\hat{\varepsilon}$.
A similar argument handles the case of $x\in\overline{L}$.
This concludes that $\SSS$ recognizes $L$ with accuracy at least $\hat{\varepsilon}$.

Next, we consider nonzero eigenvalues of $H^{(x)}_{fin}$.
For any inner state $q\in Q$, let $\qubit{\psi_q} = U_{\cent x\dollar}\qubit{q}$. It then follows that $U_{\cent x\dollar}\Lambda_0 U^{\dagger}_{\cent x\dollar} = \sum_{q\in Q^{(-)}} (U_{\cent x\dollar}\density{q}{q}U^{\dagger}_{\cent x\dollar}) = \sum_{q\in Q^{(-)}} \density{\psi_q}{\psi_q}$. Since $U_{\cent x\dollar}$ is unitary, all elements  in $\{\qubit{\psi_q}\}_{q\in Q^{(-)}}$ are nonzero eigenstates of $H_{fin}^{(x)}$ and their eigenvalues are exactly $1$. Since the ground energy is $0$, the spectral gap of $H^{(x)}_{fin}$ must be $1$, as requested.
\end{proof}

\subsection{Simulation of Garbage-Tape 1qfa's}

\emph{Regular languages} are in fact one of the most studied languages in formal language theory. In what follows, we target the class $\reg$ of all regular languages. Since regular languages are known to be recognized by 1qfa's with mixed states and superoperators (e.g., \cite{AY15}), it is easy to show that constant-size 1qqaf's generate Hamiltonians of AEQSs for regular languages. Here, we intend to use 1moqqaf's, in particular, \emph{linear-size 1moqqaf's}, whose machines have $O(n)$ inner states, to generate those Hamiltonians.

\begin{theorem}\label{linsize-regular}
$\reg \subseteq \aeqs{1moqqaf,linsize,constgap,0\mbox{-}energy}$.
\end{theorem}

It is easy to see that regular languages can be recognized by garbage-tape 1qfa's whose garbage-tape head always writes a non-blank symbol and moves to the right \cite[Lemma 3.2(2)]{Yam19b}.
Such a garbage tape is referred to as a \emph{rigid garbage tape} \cite{Yam19b}. The proof of Theorem \ref{linsize-regular} hinges at a critical simulation of garbage-tape 1qfa's by appropriate AEQSs under the desired conditions. For the intended proof, we need the following supportive lemma, by which we can derive the theorem directly. We briefly say that an AEQS $\SSS$ \emph{simulates} a machine $M$ if, for every input, the outcome of $\SSS$ matches that of $M$.

\begin{lemma}\label{oneqfa-to-AEQS}
Any 1qfa $M$ equipped with a rigid garbage tape can be exactly simulated by a certain AEQS whose Hamiltonians are generated by appropriately chosen linear-size 1moqqaf's with spectral gap of $1$ and a ground energy of $0$.
\end{lemma}

\begin{proof}
Let $L$ be any language over alphabet $\Sigma$ and consider a 1qfa $M$  with a garbage alphabet $\Xi$ that recognizes $L$ with error probability at most an appropriate constant $\varepsilon\in[0,1/2)$. For our convenience, similar to 2qfa's, $M$ is assumed to be of the form $(Q,\Sigma,\{\cent,\dollar\}, \Xi, \delta, q_0,Q_{acc},Q_{rej})$, using a (quantum) transition function $\delta$ instead of a family $\{A_{\sigma}\}_{\sigma\in\check{\Sigma}}$ of quantum operations.
Let $x =x_1x_2\cdots x_n$ be an arbitrary input of length $n$. By setting $x_0=\cent$ and $x_{n+1}=\dollar$, we succinctly write
$\tilde{x}$ for its extended input $x_0x_1\cdots x_nx_{n+1}$.

For convenience, let $B$ denote the blank symbol of the garbage tape and let $\Xi_B=\Xi\cup\{B\}$. Since an input-tape head of $M$ moves in one direction until $\dollar$, it suffices to consider the first $n+1$ cells of the garbage tape. Initially, the garbage tape consists of $n+1$ blank cells and, as the input-tape head reads all input symbols one by one, $M$ modifies one new blank cell of the garbage tape by writing a suitable non-blank  symbol from left to right.
The content of the garbage tape at time $t$ is thus of the form $sB^{n+1-t}$ with $|s|=t$.
We define $G_n$ to be $\{w \mid \exists s\in\Xi^* [ |s|\leq n+1 \wedge  w= sB^{n+1-|s|}] \}$. Notice that $|G_n|=O(n2^n)$.

For each symbol $\sigma\in\check{\Sigma}$, we further introduce a unitary operator $U_{\sigma}$ by setting
$U_{\sigma}\qubit{q,s} = \sum_{(p,\xi)\in Q\times \Xi_B} \delta(q,\sigma,p,\xi) \qubit{p, s\xi}$.
A computation of $M$ on the input $x$ results in a final quantum state $U_{\cent x\dollar}\qubit{q_0,B^{n+1}}$.
In the end of the computation of $M$, we apply two projections of the form $\Pi_{acc}\otimes I$ and $\Pi_{rej}\otimes I$, where $I$ acts on the Hilbert space spanned by the elements in $G_n$.
We write $\qubit{r_q}$ to denote $\qubit{q}\qubit{B^{n+1}}$ for each inner state $q\in Q$.
If $x\in L$ (resp., $x\notin L$), then $M$ satisfies $\trace ((\Pi_{acc}\otimes I) U_{\cent x\dollar}\density{r_{q_0}}{r_{q_0}} (U_{\cent x\dollar})^{\dagger}) \geq 1-\varepsilon$ (resp., $\trace( (\Pi_{rej}\otimes I) U_{\cent x\dollar}\density{r_{q_0}}{r_{q_0}} (U_{\cent x\dollar})^{\dagger}) \geq 1-\varepsilon$).

The desired selector $\mu$ is defined by $\mu(x)=|x|$ for any $x\in\Sigma^*$.
We define the desired AEQS $\SSS$ of the form $(m,\Sigma,\varepsilon,\{H^{(x)}_{ini}\}_{x\in\Sigma^*}, \{H^{(x)}_{fin}\}_{x\in\Sigma^*}, \{S^{(n)}_{acc}\}_{n\in\nat}, \{S^{(n)}_{rej}\}_{n\in\nat})$ as follows.
For any $n\in\nat$ and any input $x$, let $IND_n = Q\times G_n$ and let $m(x) = {\ilog|IND_{\mu(x)}|}$, which is at most ${\ilog|Q|} + {\ilog|G_{|x|}|} = O(|x|)$. Let $n=\mu(x)$. 
We further define $S^{(m(x))}_{acc} = Q_{acc}\times G_{n}$ and $S^{(m(x))}_{rej} = Q_{rej}\times G_{n}$.
We set $H^{(x)}_{ini}$ to be $\sum_{u\in IND_{n}^{(-)}} \density{\hat{u}}{\hat{u}}$, where $IND^{(-)}_n= IND_n -\{(q_0,B^{n+1})\}$.
For any symbol $\sigma\in\check{\Sigma}$, the quantum operation $A^{(n)}_{\sigma}$  is defined as $A^{(n)}_{\sigma}(H) = U_{\sigma}H (U_{\sigma})^{\dagger}$ for any linear operator $H$.
Moreover, we set $\Lambda^{(n)}_0 = \sum_{u\in IND^{(-)}_n} \density{u}{u}$.
It then follows that $A^{(n)}_{\cent{x}\dollar}(\Lambda^{(n)}_0) = U_{\cent x\dollar}\Lambda^{(n)}_0 (U_{\cent x \dollar})^{\dagger} = \sum_{(q,w)\in IND_n^{(-)}} \density{\psi_{q,w}}{\psi_{q,w}}$, where  $\qubit{\psi_{q,w}}$ is shorthand for $U_{\cent x\dollar}\qubit{q}\qubit{w}$.
In the end, $H^{(x)}_{fin}$ is set to be $A^{(n)}_{\cent x\dollar}( \Lambda^{(n)}_0)$.

The unitarity of $U_{\sigma}$ concludes that each vector  $\qubit{\psi_{q,w}}$ is an eigenstate of $H^{(x)}_{fin}$, whose eigenvalue is $1$. In particular, since $A^{(n)}_{\cent{x}\dollar}(\Lambda^{(n)}_0)\qubit{\psi_{q_0,B^{n+1}}}=0$, $\qubit{\psi_{q_0,B^{n+1}}}$ must be a unique ground state of $H^{(x)}_{fin}$, and thus the spectral gap is $1$.

Let us show that $\SSS$ correctly recognizes $L$.
Assume that $x\in L$. Note that $\qubit{\psi_{q_0,B^{n+1}}}$ equals $U_{\cent{x}\dollar}\qubit{r_{q_0}}$ for the value $n=\mu(x)$.
Since $\trace ((\Pi_{acc}\otimes I) U_{\cent x\dollar}\density{r_{q_0}}{r_{q_0}} (U_{\cent x\dollar})^{\dagger}) \geq 1-\varepsilon$, we obtain $\|(\Pi_{acc}\otimes I) \qubit{\psi_{q_0,B^{n+1}}}\|_2^2\geq 1-\varepsilon$.
Since $\Pi_{acc}$ equals $\sum_{(q,w)\in Q_{acc}\times G_n} \density{q,w}{q,w}$, there exists a quantum state $\qubit{\xi_{acc}}$ for which  $|\measure{\xi_{acc}}{\psi_{q_0,B^{n+1}}}|^2 \geq 1-\varepsilon$.
The ground state $\qubit{\psi_{q_0,B^{n+1}}}$ of $H^{(x)}_{fin}$ satisfies $\|\qubit{\psi_{q_0,B^{n+1}}} - \qubit{\xi_{acc}}\|_2^2 \leq 1-|\measure{\xi_{acc}}{\psi_{q_0,B^{n+1}}}|^2 \leq \varepsilon$.
From this fact, if we take $\hat{\varepsilon} = 1-\sqrt{\frac{\varepsilon}{2}}$, then $\qubit{\psi_{q_0,B^{n+1}}}$ is $\sqrt{2}(1-\hat{\varepsilon})$-close to $\qubit{\xi_{acc}}$. Notice that $\frac{1}{2}<\hat{\varepsilon}\leq1$. The case of $x\notin L$ is similar.
Therefore, $\SSS$ recognizes $L$ with accuracy at least $\hat{\varepsilon}$.
\end{proof}

Finally, we return to the proof of Theorem \ref{linsize-regular}.

\begin{proofof}{Theorem \ref{linsize-regular}}
Our goal is to show that $\reg\subseteq \aeqs{1moqqaf,linsize,constgap,0\mbox{-}energy}$.
Let $L$ be any regular language and take a 1qfa $M$ with a rigid garbage tape that recognizes $L$ with bounded-error probability. By Lemma \ref{oneqfa-to-AEQS}, there exists an AEQS $\SSS$ that exactly simulates $M$. Therefore, $\SSS$ recognizes $L$ with high accuracy. Lemma \ref{oneqfa-to-AEQS} ensures that $\SSS$'s Hamiltonians are generated by 1moqqaf's with spectral gap $1$ and a ground energy of $0$. This implies that $L$ belongs to $\aeqs{1moqqaf,linsize,constgap,0\mbox{-}energy}$.
\end{proofof}

\subsection{Simulation of Unambiguous Pushdown Automata}

Let us consider the family $\ucfl$ of all \emph{unambiguous context-free languages}. These languages are recognized by appropriate 1npda's, each of which has at most one accepting computation path on each input and runs in linear time. Such 1npa's are called \emph{one-way unambiguous pushdown automata} (or 1upda's). We wish to show that all 1upda's can be simulated on AEQSs under appropriate conditions, shown in Theorem \ref{CFL-vs-AEQS}.

\begin{theorem}\label{CFL-vs-AEQS}
$\ucfl \subseteq \aeqs{ltime\mbox{-}1.5qqaf,linsize,constgap}$.
\end{theorem}

To make the proof of the theorem readable, we assume that any 1upda is already \emph{in an ideal shape}\footnote{This \emph{ideal-shape property} also holds for various types of pushdown automata \cite{Yam19a}.} \cite{Yam19a}, in which  the pop operations  always take place by first reading an input symbol $\sigma$ and then making  a series (one or more) of the pop operations without reading any further input symbol. More precisely, a 1upda in an ideal shape with an input alphabet $\Sigma$ and a stack alphabet $\Gamma$ takes only the following five actions.
Let $\Gamma^{(-)}$ denote $\Gamma$ except for the bottom marker $\bot$.
(1) Scanning an input symbol $\sigma\in\Sigma$, preserve the topmost stack symbol (called a \emph{stationary operation}).  (2) Scanning  $\sigma\in\Sigma$,  push a new symbol $u$ ($\in\Gamma^{(-)}$) without changing any other symbol in the stack. (3) Scanning  $\sigma\in\Sigma$, pop the topmost stack symbol. (4) Without scanning an input symbol (i.e., $\lambda$-move), pop the topmost stack symbol. (5) The stack operation (4) comes only after either (3) or (4). The content of the stack is expressed in order as $a_1a_2\cdots a_k$ in such a way that $a_1$ is the bottom marker $\bot$ and $a_k$ is the topmost stack symbol. The bottom marker $\bot$ is assumed to be neither rewritable nor popped.

\begin{proofof}{Theorem \ref{CFL-vs-AEQS}}
Given any unambiguous context-free language $L$ over alphabet $\Sigma$, we pick a 1upda $N$ in an ideal shape that correctly recognizes $L$.
Additionally, $N$ is assumed to have the form $(Q,\Sigma,\{\cent,\dollar\}, \Gamma, \delta,q_0,\bot,Q_{acc},Q_{rej})$, where $\Gamma$ is a stack alphabet with the bottom marker $\bot$.
The transition function $\delta$ maps $Q\times \check{\Sigma}_{\lambda}\times\Gamma$ to
$\PP(Q\times \Gamma^*)$, where $\check{\Sigma}_{\lambda} = \check{\Sigma}\cup\{\lambda\}$.
Let us simulate this 1upda $N$ on an appropriately chosen AEQS $\SSS=(m,\Sigma,\varepsilon,  \{H^{(x)}_{ini}\}_{x\in\Sigma^*}, \{H^{(x)}_{fin}\}_{x\in\Sigma^*}, \{S^{(n)}_{acc}\}_{n\in\nat}, \{S^{(n)}_{rej}\}_{n\in\nat} )$, where Hamiltonians $H^{(x)}_{ini}$ and $H^{(x)}_{fin}$ are generated by certain linear-time 1.5qqaf's.
For simplicity, we further assume that, at every step, $N$ makes \emph{exactly} $k$ nondeterministic choices for a certain constant $k\in\nat^{+}$ independent of inputs.
If $\delta$ satisfies  $\delta(q,\sigma,a)=\{(p_1,z_1),(p_2,z_2),\ldots,(p_k,z_k)\}$ for a certain tuple  $(q,\sigma,a)\in Q\times\check{\Sigma}_{\lambda}\times\Gamma$, then we define a ``determinization'' of $\delta$ by setting  $\delta_d(i,q,\sigma,a)=(p_i,z_i)$ for every choice $i\in[k]$.
Although there may be infinite computation paths of $N$ on certain inputs $x$, whenever $x\in L$, we can find an accepting computation path of $N$ whose length is $O(|x|)$, more precisely, at most $c|x|+c$ for a suitable constant $c\in\nat^{+}$ independent of $x$.
We succinctly write $\ell_{n}$ for $cn+c$.

We can force $N$ to stop its computation after $\ell_{|x|}$ steps without changing the outcome of $N$. Thus, length-$\ell_{n}$ series of nondeterministic choices made by $N$ essentially contribute to the outcome of $N$. We call such series \emph{decision series}.
By the definition of 1upda's, since $N$ may make a certain number of consecutive $\lambda$ transitions, each nondeterministic choice also needs to specify  $N$'s current selection between a $\lambda$-move and a non-$\lambda$-move.
Because of the presence of $\lambda$-transitions, nonetheless, we not only trace the tape head location of $N$ but also implement an \emph{internal clock} to keep track of time to describe the progress of $N$'s computation. 
A basic idea of our simulation is, therefore, to keep track of the current tape head location, the current  content of the stack, and a clock time as well as the current content of an extra write-once tape used as a garbage tape. This last tape will be used to make our simulation reversible.

In what follows, a sextuple $(s,q,i,t,r,g)$ expresses the current circumstance in which $N$ is in inner state $q$ at time $t$, scanning the $i$th tape cell together with stack content $r$ and garbage-tape content $g$, provided that $s=s_1s_2\cdots s_{\ell_n}$ is a decision series, where each $s_i$ takes a value from  $\{0,1\}\times [k]$ for any index $i\in[\ell_n]$.
From this circumstance, we determine the next move in the following fashion.
If $s_t$ has the form $(0,j)$ for a certain number $j\in[k]$, then we make a $\lambda$-move by applying $\delta_d(j,q,\lambda,a)$, where $a$ is the topmost symbol of the stack content $r$; on the contrary, if $s_t$ is of the form $(1,j)$, then we apply $\delta_d(j,q,x_{(i)},a)$ instead.

We use a track notation $\track{q}{r}$ (cf. \cite{TYL10}) and prepare the set  $\Xi = \{\track{q}{r}\mid q\in Q,r\in \Gamma\}$.
The selector $\mu$ is set to be $\mu(x)=|x|$ for all $x\in\Sigma^*$. Additionally, given a parameter $n$, we define
$IND_n$ to be $((\{0,1\}\times[k])^{\ell_{n}}\cup\{0^{\ell_n}\}) \times ( Q \times [0,n+1]_{\integer}\times [0,\ell_{n}]_{\integer}\times (\Gamma^{(-)})^{\leq \ell_{n}}\bot \times \Xi^{\leq \ell_{n}} )^2$.
Finally, $m(x)$ is set to be ${\ilog|IND_{\mu(x)}|}$, which is clearly bounded by $O(|x|)$. For later use, we abbreviate $(q_0,0,0,\bot,\lambda)$ as $\xi_0$.

Let us define the desired AEQS $\SSS$ as follows.
The initial Hamiltonian $H^{(x)}_{ini}$ is simply set to be $\sum_{u\in IND_{\mu(x)}^{(-)}}\density{\hat{u}}{\hat{u}}$, where $IND_n^{(-)} = IND_n -\{(0^{\ell_n},\xi_0,\xi_0)\}$. To construct $H^{(x)}_{fin}$, however, we wish to describe an appropriate linear-time 1.5qqaf $\MM= \{M_n\}_{n\in\nat}$.
Let us fix $x\in\Sigma^*$ arbitrarily and set $n=\mu(x)$.
Hereafter, we are focused on $M_n$, which is assumed to have the form
$(Q^{(n)}, \Sigma,\{\cent,\dollar\}, \{1,2\}, \KK_n,  \KK_{n,\cent},  \Lambda^{(n)}_0, Q^{(n)}_0)$, where  $Q^{(n)} = IND_n$, $\KK_n= \{K^{(n)}_e\}_{e\in\{1,2\}}$, and  $\KK_{n,\cent}= \{K^{(n)}_{\cent,e}\}_{e\in\{1,2\}}$.

Firstly, we intend to define a restriction of $K^{(n)}_e$ to $x$, denoted by $K^{(n,x)}_e$, which satisfies $K^{(n)}_e = \sum_{x\in\Sigma^*} (\density{x}{x}\otimes K^{(n,x)}_e)$.
Let $s=s_1s_2\cdots s_{\ell_n}$ denote a decision series in $(\{0,1\}\times[k])^{\ell_n}$ that uniquely specifies a corresponding computation path of $M_n$ on $x$.
The machine $M_n$ works with $11$ registers that store a tuple of the form $(s,q,i,t,r,g,q',i',t',r',g')$.
If the last five registers do not contain $\xi_0 = (q_0,0,0,\bot,\lambda)$, then we simply perform the identity operator except for an increment of both the current clock time and the current tape head location.
To fixate the last five registers to $\xi_0$, we use the following special Kraus operators $\{K^{(n)}_{\cent,1},K^{(n)}_{\cent,2}\}$ at the first step.
In scanning $\cent$, by letting $u=(q,i,t,r,g)$ and $w=(q',i',t',r',g')$, similarly to Examples \ref{ex:Pal}--\ref{example:MULTIDUP}, we make the following transitions: $K^{(n)}_{\cent,1} \qubit{s,u,w} = \measure{\xi_0}{u}\cdot \qubit{s,u,w}$ and $K^{(n)}_{\cent,2} \qubit{s,u,w} = \sum_{z:z\neq \xi_0} \measure{z}{u}\cdot \qubit{s,w,u}$.

In what follows, we assume that the last five registers of $M_n$ contain only $\qubit{\xi_0}$. Let us consider the case where $N$ pushes a stack symbol $a$. Assume that, at time $t$, and the nondeterministic choice $s_t$ is $(1,h)$ and $N$ changes its current surface configuration $(q,i,r)$ to another surface configuration $(p,j,ra)$ in a single step by applying a transition $\delta_d(h,q,x_{(i)},b_r)=(p,b_ra)$, where $x_{(0)}=\cent$, $x_{(|x|+1)}=\dollar$, and $b_r$ is the topmost stack symbol of $r$.
In this case, if $i\leq n$, then we apply   $K_1^{(n,x)}\qubit{s,q,i,t,r,g}\qubit{\xi_0} = \qubit{s,p,j,t+1,ra,g\track{q}{b_r}}\qubit{\xi_0}$; moreover,
if $i=n+1$, then we set $K_1^{(n,x)}\qubit{s,q,n+1,t,r,g}\qubit{\xi_0} = \qubit{s,p,0,t+1,ra,g\track{q}{b_r}}\qubit{\xi_0}$ since $M_n$'s input tape is circular. If $s_t=(0,h)$, then $\delta_d$ takes a similar transition except for $\lambda$ replacing $x_{(i)}$.
In the case where $s_t= (1,h)$ and $N$ pops a stack symbol, we assume that $N$ changes $(q,i,r)$ to $(p,j,r')$ with $r=r'b_r$ by making a transition  $\delta_d(h,q,x_{(i)},b_r) = (p,\lambda)$.
We then apply $K_1^{(n,x)}\qubit{s,q,i,t,r,g}\qubit{\xi_0} = \qubit{s,p,j,t+1,r',g\track{q}{b_r}}\qubit{\xi_0}$. The case of  $s_t=(0,h)$ can be handled similarly.
Next, let us consider the missing case of $s=0^{\ell_n}$.
Starting with the quantum state $\qubit{0^{\ell_n},\xi_0,\xi_0}$,
we simply make the transition $K^{(n,x)}_1 \qubit{0^{\ell_n},q_0,0,t,\bot,\lambda}\qubit{\xi_0} = \qubit{0^{\ell_n},q_0,0,t+1,\bot,\lambda}\qubit{\xi_0}$.

The final Hamiltonian $H^{(x)}_{fin}$ is then defined in the following way. Let $\Lambda^{(n)}_0 = I - \frac{1}{2} \density{0^{\ell_{n}},\xi_0,\xi_0}{0^{\ell_{n}},\xi_0,\xi_0}$, let $Q^{(n)}_0 = \{(s,q,i,\ell_n,r,g,\xi_0)\in IND_n \mid s\neq 0^{\ell_n}, q\in Q_{acc} \}$,
and let $\Pi^{(n)}_0$ be the projective measurement induced from $Q^{(n)}_0$.
For any linear operator $H$, we set $A^{(n,x)}(H)$ to be $\sum_{e\in\{1,2\}} K^{(n,x)}_e H (K^{(n,x)}_e)^{\dagger}$.
The desired Hamiltonian $H^{(x)}_{fin}$ is then defined as $\Pi^{(n)}_0 (A^{(n,x)})^{\ell_{|x|}} (\tilde{\Lambda}^{(n)}_0) \Pi^{(n)}_0$, provided that  $\tilde{\Lambda}^{(n)}_0 = \sum_{i\in\{1,2\}} K^{(n)}_{\cent,i} \Lambda^{(n)}_0 (K^{(n)}_{\cent,i})^{\dagger}$.
Moreover, let $S^{(m(x))}_{acc} = Q^{(n)}_0$ and $S^{(m(x))}_{rej} = \{(0^{\ell_n},q_0,0,\ell_n,\bot,\lambda,\xi_0)\}$.

To complete the proof, we need to prove that $\SSS$ recognizes $L$ correctly.  For any string $x$ in $L$, there exists an accepting computation path  of $N$ on $x$, and this derives the existence of an accepting decision series, say, $s$ in $(\{0,1\}\times[k])^{\ell_n}$, where $n=\mu(x)$. 
We take the quantum state $\qubit{\phi_x}= (K_1^{(n,x)})^{\ell_{|x|}}\qubit{s,q_0,0,0,\bot,\lambda} \qubit{\xi_0}$.
Since  $\Pi^{(n)}_0\qubit{\phi_x} =0$,
it instantly follows that $H^{(x)}_{fin}\qubit{\phi_x} = \Pi^{(n)}_0 (A^{(n,x)})^{\ell_n}(\tilde{\Lambda}^{(n)}_0)\Pi^{(n)}_0 \qubit{\phi_x} = 0$.
This implies that $\qubit{\phi_x}$ is the ground state, whose ground energy is $0$. Since there is exactly one accepting computation path of $N$ on $x$, the ground state must be unique. 
Moreover, $\qubit{\phi_x}$ falls in $QS^{(m(x))}_{acc}$.
On the contrary, we assume $x\notin L$ and take $\qubit{\psi_x} = (K_1^{(n,x)})^{\ell_n}\qubit{0^{\ell_{n}},\xi_0,\xi_0}$, which belongs to  $QS^{(m(x))}_{rej}$.
It then follows that $\Pi^{(n)}_0 (A^{(n,x)})^{\ell_n}(\tilde{\Lambda}^{(n,x)}_0) \Pi^{(n,x)}_0 \qubit{\psi_x} = \Pi^{(n)}_0 (K_1^{(n,x)})^{\ell_{n}}\tilde{\Lambda}^{(n)}_0 ((K_1^{(n,x)})^{\ell_n})^{\dagger} \qubit{\psi_x} = \frac{1}{2}\qubit{\psi_x}$
since $\Pi^{(n)}_0 \qubit{\psi_x} = \qubit{\psi_x}$.
We thus conclude that $H^{(x)}_{fin}\qubit{\psi_x} = \frac{1}{2}\qubit{\psi_x}$. The quantum state $\qubit{\psi_x}$ is therefore the  ground state with a ground energy of $\frac{1}{2}$.
\end{proofof}

\subsection{Simulation of Polynomial-Size 2qfa's}

A \emph{quantum Turing machine} (or a QTM)  \emph{with a (flexible) garbage tape} is similar to a (flexible) garbage-tape 2qfa but it can rewrite any tape cell using symbols from a \emph{tape alphabet} (which may be different from an input alphabet).
In particular, we are focused on logarithmic-space QTMs equipped with garbage tapes running in expected\footnote{This expectation is taken over all the lengths of the computation paths of a QTM on each fixed input.} polynomial time.
The use of a garbage tape helps us postpone any intermediate projective measurement until the end of a computation \cite{Yam19b}, and thus all transitions can be described unitarily.
Those machines form the complexity class  $\mathrm{ptime\mbox{-}BQL}$, composed of all languages recognized with bounded-error
probability by garbage-tape QTMs using $O(\log|x|)$ work-tape cells in expected $|x|^{O(1)}$ steps \cite{Yam19b,Yam19c}.

\begin{theorem}\label{BQL-vs-2qqaf}
$\mathrm{ptime\mbox{-}BQL} \subseteq \aeqs{ptime\mbox{-}2qqaf,polysize,0\mbox{-}energy}$.
\end{theorem}

By \cite[Lemma 5.2]{Yam19b}, a bounded-error QTM $M$ equipped with a garbage tape running in expected polynomial time using logarithmic space (even if the QTM has advice) can be simulated by an appropriate family $\{N_{n}\}_{n\in\nat}$ of garbage-tape 2qfa's with $n^{O(1)}$ inner states and expected polynomial runtime in such a way that, for any input $x$, $M$ accepts (resp., rejects) $x$ with bounded-error probability iff $N_{|x|}$ accepts (resp., rejects) $x$ with bounded-error probability.
In the following proof of Theorem \ref{BQL-vs-2qqaf}, we will utilize this 2qfa-characterization of resource-bounded QTMs with garbage tapes.

\begin{proofof}{Theorem \ref{BQL-vs-2qqaf}}
Given an arbitrary language $L$ in $\mathrm{ptime}\mbox{-}\bql$ over alphabet $\Sigma$, as noted above, we take an appropriate family $\NN= \{N_{n}\}_{n\in\nat}$ of garbage-tape 2qfa's that recognizes $L$ with bounded-error probability using $n^{O(1)}$ inner states.
Let us assume that $N_{n}$ has the form $(Q_{n},\Sigma,\{\cent,\dollar\}, \Xi_{n},\delta_{n}, q_{0,n}, Q_{acc,n},Q_{rej,n})$, where $\Xi_{n}$ is a garbage alphabet.
For each input $x$ of length $n$, the (quantum) transition function $\delta_{n}$ induces a time-evolution operator $U^{(x)}_{n}$ defined by   $U^{(x)}_{n}\qubit{q,i,z} =  \sum_{p,d,\xi}\delta_{n}(q,x_{(i)},p,d,\xi) \qubit{p,i+d\:(\mathrm{mod}\:n+2),z\xi}$ for any $(q,i,z)\in Q_{n}\times [0,n+1]_{\integer}\times (\Xi_{n})^*$, where $(p,d,\xi)$ in the summation ranges over $Q_{n}\times \{-1,0,+1\} \times\Xi_{n}$.

Take an appropriate polynomial $p$ such that, on every input $x$,  $N_{|x|}$ halts within expected $p(|x|)$ time. By the property of expectation, we can take two special constants $c\geq1$ and $\varepsilon\in(0,1/2]$ such that, even if we force $N_{|x|}$ to stop its operation after $cp(|x|)$ steps, the probability of acceptance/rejection of $N_{|x|}$ remains at least $\frac{1}{2}+\varepsilon$. For convenience, we set $\bar{p}(n)=cp(n)$ for any $n\in\nat$.
Similarly to the proof of Lemma \ref{oneqfa-to-AEQS} with a slight modification of
garbage-tape contents, we introduce the set $G_n$ of such contents as $\{w\mid \exists s\in\Xi_n^*[|s|\leq \bar{p}(n)\wedge w=sB^{\bar{p}(n)-|s|}]\}$, where $B$ is a special symbol indicating the blank cell.
It then follows that (i) $\|\Pi_{acc,n} (U^{(x)}_{|x|})^{\bar{p}(|x|)}\qubit{q_{0,n},0,B^{\bar{p}(n)}}\|_2^2 \geq \frac{1}{2}+\varepsilon$  for any $x\in L$ and
(ii) $\|\Pi_{rej,n} (U^{(x)}_{|x|})^{\bar{p}(|x|)}\qubit{q_{0,n},0,B^{\bar{p}(n)}}\|_2^2 \geq \frac{1}{2}+\varepsilon$ for any $x\in \overline{L}$, where $\Pi_{acc,n}$ and $\Pi_{rej,n}$ are projective measurements associated with $Q_{acc,n}$ and $Q_{rej,n}$, respectively.
It thus suffices to simulate $\NN$ on an appropriate AEQS in order to recognize $L$.

From $\NN$, we intend to define the target AEQS $\SSS$ of the form $(m,\Sigma,\varepsilon,  \{H^{(x)}_{ini}\}_{x\in\Sigma^*}, \{H^{(x)}_{fin}\}_{x\in\Sigma^*}, \{S^{(n)}_{acc}\}_{n\in\nat}, \{S^{(n)}_{rej}\}_{n\in\nat} )$.
For this purpose, we need to define a polynomial-time 2qqaf $\MM=\{M_n\}_{n\in\nat}$ that produces, in particular, the final Hamiltonians of $\SSS$.
Firstly, we define $IND_n$ to be $Q_{n}\times [0,n+1]_{\integer}\times G_n$ and we set $\Lambda^{(n)}_0$ to be $\sum_{u\in IND_n^{(-)}} \density{u}{u}$, where  $IND_n^{(-)} = IND_n - \{(q_{0,n},0,B^{\bar{p}(n)})\}$.
As for the desired selector $\mu$, we define $\mu(x)=|x|$ for any $x\in\Sigma^*$.
Since $H^{(x)}_{fin}$ has dimension $|IND_{\mu(x)}|$, the size $m(x)$ of $\SSS$ is ${\ilog|IND_{\mu(x)}|}$, which is clearly upper-bounded by a certain polynomial in $|x|$.
By setting $A^{(n,x)}(H)$ to be $U^{(x)}_{n} H (U^{(x)}_{n})^{\dagger}$,  the desired $H^{(x)}_{fin}$ is defined to be $(A^{(n,x)})^{\bar{p}(n)} (\Lambda^{(n)}_0)$.
The acceptance/rejection criteria pair $(S^{(m(x))}_{acc},S^{(m(x))}_{rej})$ is defined by $S^{(m(x))}_{acc} = \{(q,i,z)\in IND_{n} \mid q\in Q_{acc,n}\}$ and $S^{(m(x))}_{rej}
=\{(q,i,z)\in IND_{n} \mid q\in Q_{rej,n}\}$.

Finally, we wish to verify that $\SSS$ correctly solves $L$. Let $x$ be any string in $\Sigma^*$ and set $n=\mu(x)$.
We first assume that $x$ is in $L$. Let us consider the quantum state $\qubit{\phi_x} = (U^{(x)}_{n})^{\bar{p}(n)} \qubit{\xi_n}$, where $\qubit{\xi_n}= \qubit{q_{0,n},0,B^{\bar{p}(n)}}$. It then follows that $H^{(x)}_{fin}\qubit{\phi_x} = (A^{(n,x)})^{\bar{p}(n)}(\Lambda^{(n)}_0)\qubit{\phi_x} = (U^{(x)}_{n})^{\bar{p}(n)} \Lambda^{(n)}_0 ((U^{(x)}_{n})^{\dagger})^{\bar{p}(n)} (U^{(x)}_{n})^{\bar{p}(n)} \qubit{\xi_n} = U^{\bar{p}(|x|)}_{n,|x|}\Lambda^{(n)}_0 \qubit{\xi_n} =0$ because of  $\Lambda^{(n)}_0\qubit{\xi_n}=0$.
This implies that $\qubit{\phi_x}$ is the ground state and its corresponding ground energy is $0$.
Note that the value $\min_{\qubit{\psi}}\|\qubit{\phi_x}-\qubit{\psi}\|_2^2$ over all vectors $\qubit{\psi}$ in $QS^{(m(x))}_{acc}$ is upper-bounded by $1-\max_{\qubit{\psi}}|\measure{\psi}{\phi_x}|^2$.
Since $\max_{\qubit{\phi}} |\measure{\psi}{\phi_x}|^2\geq |\bra{\phi_x}\Pi_{acc,n}\ket{\phi_x}| = \|\Pi_{acc,n}\qubit{\phi_x}\|_2^2\geq \frac{1}{2}+\varepsilon$, it follows that $\min_{\qubit{\psi}}\|\qubit{\phi_x}-\qubit{\psi}\|_2^2 \leq 1- (\frac{1}{2}+\varepsilon) = \frac{1}{2}-\varepsilon$.
We conveniently set $\hat{\varepsilon}=1-\frac{\sqrt{1-2\varepsilon}}{2}$. Since $0<\varepsilon\leq \frac{1}{2}$, we obtain  $\frac{1}{2}<\hat{\varepsilon}\leq 1$. Because of $\min_{\qubit{\psi}}\|\qubit{\phi_x}-\qubit{\psi}\|_2\leq \sqrt{2}(1-\hat{\varepsilon})$, we conclude that the accuracy rate is at most $\hat{\varepsilon}$.

The case of $x\notin L$ is similarly treated. Therefore, $\SSS$ correctly solves $L$, as requested.
\end{proofof}

\section{Computational Complexity of AEQS}\label{sec:complexity-AEQS}

We have demonstrated the power of conditional AEQSs for various languages in Sections \ref{sec:behaviors}--\ref{sec:simulation}.
Nevertheless, we still wonder how powerful the AEQSs can be.
To answer this question, we want to present a simple upper bound on the computational complexity of conditional AEQSs.
In Lemma \ref{time-evolution}, we have already shown how to approximate each AEQS  and this approximation process can be carried out by certain forms of  resource-bounded QTMs with \emph{advice}\footnote{In general,  \emph{advice} is an external source that can provide necessary information to enhance the computational power of underlying machines \cite{NY04,Yam14b,Yam19b}.} of polynomial size.
Such advised QTMs naturally induce a nonuniform analogue of ptime-BQL, denoted  $\mathrm{ptime}\mbox{-}\bql/\poly$ \cite{Yam19b}, which is composed of  all languages recognized with bounded-error probability by garbage-tape QTMs with polynomial-size advice strings (one per every input size) using logarithmic space in expected polynomial runtime.

Unlike the previous sections, we need to restrict the behaviors of selectors. A selector $\mu:\Sigma^*\to\nat$ for an alphabet $\Sigma$ is said to be \emph{logarithmic-space computable} if there exists a deterministic Turing machine (DTM) equipped with a read-only input tape, a rewritable work tape, and a write-once output tape such that, for each input $x\in\Sigma^*$, $M$ outputs $1^{\mu(x)}$ using $O(\log|x|)$ work-tape cells in $|x|^{O(1)}$ steps.
We introduce the notation $\FF=$``Lsel'' to refer to the logarithmic-space computable selectors.
Similarly, we restrict the use of transitions of 1moqqaf's so that their transition amplitudes are limited to  \emph{logarithmic-space approximable}\footnote{A real number $\alpha$ is \emph{logarithmic-space approximable} if there are a polynomial $p$ and a DTM with a read-only input tape, a rewritable work tape, and a write-only output tape that, on each input $1^n$,  produces a number, which is $2^{-p(n)}$-close to $\alpha$ in $n^{O(1)}$  time using only $O(\log{n})$ space. A complex number is \emph{logarithmic-space  approximable} if its real and imaginary parts are both logarithmic-space approximable.} \emph{complex numbers}.
To express such a special restriction imposed on 1moqqaf's, we use a new notation   $\FF=$``1moqqaf($\tilde{\complex})$'' in place of $\FF=$``1moqqaf''.

\begin{proposition}
$\aeqs{1moqqaf\mbox{$(\tilde{\complex})$}, Lsel,constsize,polygap} \subseteq \mathrm{ptime}\mbox{-}\bql/\poly$.
\end{proposition}

\begin{proof}
The following argument is motivated by \cite{DMV01,FGG+01} and heavily depends on Lemma \ref{time-evolution}.
Let $L$ be any language in $\aeqs{1moqqaf\mbox{$(\tilde{\complex})$}, Lsel, constsize, polygap}$ and take an AEQS $\SSS$ that recognizes $L$ with accuracy at least $\varepsilon\in(1/2,1]$.
Assume that $\SSS$ is of the form $(m,\Sigma,\varepsilon,  \{H^{(x)}_{ini}\}_{x\in\Sigma^*}, \{H^{(x)}_{fin}\}_{x\in\Sigma^*}, \{S^{(n)}_{acc}\}_{n\in\nat}, \{S^{(n)}_{rej}\}_{n\in\nat} )$. For each input $x\in\Sigma^*$, let us consider the minimal evolution time $T_x$ of the system.
By the adiabatic theorem (stated in Section \ref{sec:system-evolution}), there exits an appropriate polynomial $p$ such that $T_x\geq p(|x|)$ holds for any $x\in\Sigma^*$.

Since $m(x)$ is upper-bounded by a certain constant, we set $k=2^{m(x)}$.
Let $\qubit{\psi_g(0)}$ and $\qubit{\psi_g(T_x)}$ respectively express the ground states of $H^{(x)}_{ini}$ and $H^{(x)}_{fin}$.
Moreover, let $\qubit{\psi(t)}$ denote the quantum state defined in Section  \ref{sec:system-evolution} together with its associated unitary matrix $U_{T_x}$.

Assume that, with the help of an appropriate logarithmic-space computable selector $\mu$, $H^{(x)}_{fin}$ is generated by a certain 1moqqaf $\MM=\{M_n\}_{n\in\nat}$, where each $M_n$ has the form $(Q^{(n)},\Sigma,\{\cent,\dollar\}, \{A^{(n)}_{\sigma}\}_{\sigma\in\check{\Sigma}}, \Lambda^{(n)}_0)$ with $Q^{(n)}=\{q_1,q_2,\ldots,q_k\}$. Since $\mu$ is polynomially-bounded, we take a polynomial $r$ satisfying $\mu(x)\leq r(|x|)$ for all $x\in\Sigma^*$.
It thus possible to restrict the range of the parameter $n$ appearing in $M_n$ within the integer interval $[0,r(|x|)]_{\integer}$.
Since $A^{(n)}_{\sigma} = U^{(n)}_{\sigma} H(U^{(n)}_{\sigma})^{\dagger}$ for an appropriately chosen unitary operator $U^{(n)}_{\sigma}$, $H^{(x)}_{fin}$ can be  written as  $U^{(n)}_{\cent x\dollar}\Lambda^{(n)}_0 (U^{(n)}_{\cent x\dollar})^{\dagger}$, where $n=\mu(x)$.

We choose an appropriate unitary matrix $P_{n,x}$ that helps us diagonalize $\Lambda^{(n)}_0$ as $\sum_{q\in Q^{(n)}} \xi_{n,q} P_{n,x} \density{q}{q} P_{n,x}^{\dagger}$ for certain real numbers $\xi_{n,q}$'s; namely, $\Lambda^{(n)}_0 = P_{n,x} diag(\xi_{n,q_1},\xi_{n,q_2},\ldots,\xi_{n,q_k}) P_{n,x}^{\dagger}$.
The set $\{\xi_{n,q} \mid q\in Q\}$ consists of
all eigenvalues of $H^{(x)}_{fin}$.
To make the following argument readable, we further assume that $H^{(x)}_{ini}$ has eigenstates $\{\qubit{\hat{q}}\mid q\in Q^{(n)}\}$ and their associated eigenvalues $\{\lambda_{n,q}\mid q\in Q^{(n)}\}$; in other words, $H^{(x)}_{ini}$ is of the form $\sum_{q\in Q^{(n)}} \lambda_{n,q}\density{\hat{q}}{\hat{q}}$.

Since $H^{(x)}_{ini}$ and $H^{(x)}_{fin}$ have dimension $k$, without loss of generality, we assume that $\max\{\|H^{(x)}_{ini}\|, \|H^{(x)}_{fin}\|\} \leq ck$ for a certain constant $c>0$. Take a large enough number $R>0$, which will be determined later.
As in the proof of Lemma \ref{time-evolution}, for any index $j\in[0,R-1]_{\integer}$, we set
$\alpha_j=\frac{1}{\hbar}\frac{T_x}{R}\left( 1- \frac{2j+1}{2R} \right)$ and $\beta_j = \frac{1}{\hbar}\frac{T_x}{R}\frac{2j+1}{2R}$. We then define $U^{(n)}(j+1,j) = e^{-\imath\alpha_j H^{(x)}_{ini} - \imath \beta_j H^{(x)}_{fin}}$ and $V^{(n)}(j+1,j) = e^{-\imath\alpha_j H^{(x)}_{ini}}\cdot e^{ - \imath \beta_j H^{(x)}_{fin}}$. By Lemma \ref{time-evolution}, $U^{(n)}(j+1,j)$ can be approximated by $V^{(n)}(j+1,j)$ to within $O(\frac{c^2k^2 T_x^2}{R^2})$.
It thus follows that the matrix $U^{(n)}_{R} = U^{(n)}(R,R-1)\cdots U^{(n)}(2,1) U^{(n)}(1,0)$ can be approximated by $V^{(n)}_{R} = V^{(n)}(R,R-1)\cdots V^{(n)}(2,1) V^{(n)}(1,0)$ to within $O(\frac{ck^2 T_x^2}{R})$. Notice that $U_{T_x}$ coincides with $U^{(n)}_R$.

By the nature of a matrix exponential, it follows that
$e^{-\imath \alpha_j H^{(x)}_{ini}} = \sum_{q\in Q^{(n)}} e^{-\imath \alpha_j}\density{\hat{q}}{\hat{q}}$ and $e^{-\imath \beta_j H^{(x)}_{fin}} = \sum_{q\in Q^{(n)}} e^{-\imath \beta_j\xi_{n,q}} P_{n,x}\density{q}{q} (P_{n,x})^{\dagger}$.
By setting $\gamma_n= \frac{1}{\hbar}\frac{T_x}{R}\frac{1}{2R}$,
we obtain $\alpha_j = (2R-2j-1)\gamma_n$ and $\beta_j = (2j+1)\gamma_n$.
With the use of $\gamma_n$, we prepare two special phase shift operators,
$PS^{(x)}_{ini}  = \sum_{q\in Q^{(n)}} e^{-\imath \lambda_{n,q}\gamma_n}\density{\hat{q}}{\hat{q}}$ and $PS^{(x)}_{fin}= \sum_{q\in Q^{(n)}} e^{-\imath \xi_{n,q} \gamma_n} P_{n,x}\density{q}{q}(P_{n,x})^{\dagger}$, and set $Z^{(x)}(j+1,j)$ to be $W^{\otimes k}(PS^{(x)}_{ini})^{2R-2j-1}W^{\otimes k} (PS^{(x)}_{fin})^{2j+1}$.
We then conclude that $V^{(n)}(j+1,j)$ coincides with $Z^{(x)}(j+1,j)$. Therefore, the matrix $U_{T_x}$ is approximable by $Z^{(x)}(R,R-1)\cdots Z^{(x)}(2,1)$ to within $O(\frac{c^2k^2 T_x^2}{R})$. Note that, since $T_x \geq p(|x|)$, if we set $R = c^2k^2 T_x^3$, then we obtain $\frac{c^2k^2 T_x^2}{R} = O(\frac{1}{p(n)})$.

Let us consider the following quantum algorithm $\AAA$. In a preprocessing stage, we encode the information on good approximations of the values $\xi_{n,q}\gamma_n$'s for all $q\in Q^{(n)}$ and all $n\in[0,r(|x|)]_{\integer}$ into a single advice string of polynomial length in $|x|$.
Let $x$ be any input in $\Sigma^*$. We first compute the number $n=\mu(x)$ using logarithmic space.
We then retrieve the encoded approximation of $\{\xi_{n,q}\gamma_n \mid q\in Q\}$ from the advice bit by bit, prepare $PS^{(x)}_{ini}$ and $PS^{(x)}_{fin}$, and construct  $Z^{(x)}(\cdot,\cdot)$. We sequentially apply $Z^{(x)}(j+1,j)$ by changing $j$ from $0$ to $R-1$ and obtain an approximation of $U_{T_x}\qubit{\psi_g(0)}$, which further approximates $\qubit{\psi_g(T_x)}$.
By measuring the obtained quantum state in the bases of $QS^{(m(x))}_{acc}$ and $QS^{(m(x))}_{rej}$, we can determine the acceptance and the rejection of the input $x$.
Note that $\qubit{\psi_g(T_x)}$ is $\sqrt{2}(1-\varepsilon)$-close to a certain quantum state, say, $\qubit{\phi}$  in either $QS^{(m(x))}_{acc}$ or  $QS^{(m(x))}_{rej}$.
Therefore, we correctly accept or reject $x$ with probability at least $|\measure{\psi_g(T_x)}{\phi}|^2$. Since $\|\qubit{\psi_g(T_x)}-\qubit{\phi}\|_2^2 = 2(1-|\measure{\psi_g(T_x)}{\phi}|) \leq 2(1-\varepsilon)^2$, the value  $|\measure{\psi_g(T_x)}{\phi}|^2$ is lower-bounded by $(1-(1-\varepsilon)^2)^2$, which is larger than $\frac{1}{2}$. Since $\SSS$ solves $L$ with high accuracy, our algorithm $\AAA$ correctly recognizes $L$ as well.

The space usage of the above calculation is upper-bounded by a suitable logarithmic function in $n$ since the multiplication can be carried out using logarithmic space.
\end{proof}

\section{A Brief Discussion on Future Research}

\emph{Quantum annealing} has been sought as a quantum analogue of \emph{simulated annealing} (or thermal annealing) and it has been implemented to solve computational problems by quantum-physical systems, which evolve according to an appropriately defined Schr\"{o}dinger equation.
This exposition has sought for discrete-time adiabatic evolutions of such quantum systems and it has proposed a practical model of adiabatic quantum computation under the name of \emph{adiabatic quantum evolutionary system} (AEQS). An actual evolution of an AEQS is dictated by two essential Hamiltonians of the AEQS, and therefore we need to develop a simple, practically easy method of constructing these Hamiltonians.

The past literature has discussed a close connection  between adiabatic quantum computing and quantum circuits as well as quantum Turing machines. In sharp contrast, we have taken a bold step toward establishing a new connection to \emph{quantum automata theory}.
Because of their quite simple features, quantum finite automata are more suited to describe the behaviors of Hamiltonians of AEQSs and thus to solve simple decision problems (or equivalently, languages).
Throughout this exposition, we have demonstrated that numerous variants of quantum finite automata proposed in the past literature help us measure the computational complexity of languages on AEQSs.

Nonetheless, our preliminary results of this exposition are not fully satisfactory and this fact motivates us to look for further investigations in a hope of promoting a far better understanding of the essence of adiabatic evolutions of underlying quantum systems. Toward full-fledged future research on AEQSs, we wish to present a short list of challenging open problems. We strongly hope that this work will pave a road to a fruitful research field of \emph{algorithmic adiabatic quantum computing}.

\begin{itemize}
  \setlength{\topsep}{-2mm}%
  \setlength{\itemsep}{1mm}%
  \setlength{\parskip}{0cm}%

\item {\bf [Searching for Better Condition Sets]}
In Section \ref{sec:behaviors}, we have presented various condition sets $\FF$, which are sufficient for solving given target languages, such as $Equal$ and $SymCoin$. However, we suspect that the choices of these condition sets $\FF$ might not be the best ones for most of the languages.
For example, we have obtained $Equal \in\aeqs{1moqqaf,logsize,0\mbox{-}energy}$. Is it possible to use ``constsize'' instead of ``logsize'' even if we need to remove the condition of ``$0$-energy''?
It appears to be challenging, in general, to determine the precise complexity of AEQSs for various sets $\FF$ of conditions.
For example, what is the exact computational complexity of $\aeqs{1qqaf,logsize,polygap,0\mbox{-}energy}$?

\item {\bf [Proving Closure Properties]}
In automata theory, various closure properties have been discussed for numerous language families. It is important to find useful closure properties for conditional AEQSs as well.
In Section \ref{sec:struct-properties}, we have briefly touched an initial investigation on such closure properties; however, it is still far from settling the question of whether numerous  closure properties hold for even simple conditional AEQSs.

\item {\bf [Expanding Machine Models]}
In Sections \ref{sec:behaviors}--\ref{sec:simulation}, we have used various quantum quasi-automata families, including 1moqqaf, 1qqaf, linear-time 1.5qqaf, and polynomial-time 2qqaf, in order to construct Hamiltonians of several conditional AEQSs. These machine families are introduced directly from well-known models of 1moqfa's, 1qfa's, 1.5qfa's, and 2qfa's with slight modifications. In the past literature,  many more machine models have been proposed to expand the early models of quantum finite automata. It also seems desirable to make  such machine models fit into  our AEQS platform to enhance the efficiency of the construction of the Hamiltonians of AEQSs.

\item {\bf [Separations among AEQSs]}
In Section \ref{sec:simulation}, we have demonstrated simple lower-bounds on the complexity of four language families in terms of conditional AEQSs.
In Section \ref{sec:complexity-AEQS}, on the contrary, we have briefly discussed an upper bound on the computational complexity of certain conditional AEQSs.
Since the computational complexity of conditional AEQSs hinges at the choice of their underlying conditional sets $\FF$, it may be beneficial to ``compare'' the power of two condition sets $\FF_1$ and $\FF_2$ by showing that $\aeqs{\FF_2}$ properly contains $\aeqs{\FF_1}$; namely, $\aeqs{\FF_1}\subsetneqq \aeqs{\FF_2}$.

\item {\bf [Allowing Low Accuracy]}
Another (theoretically important) consideration of AEQSs is to allow  \emph{low accuracy} in place of high accuracy when recognizing  various languages.
More precisely, a \emph{lowly-accurate AEQS} accepts an input with accuracy more than $1/2$ and rejects the input with accuracy at least $1/2$. It seems natural to ask what languages we can solve on AEQSs with low accuracy.  To emphasize the use of low accuracy, we wish to use a special notation of $\mathrm{la}\mbox{-}\aeqs{\FF}$ for the low accuracy version of $\aeqs{\FF}$. Following an idea of \cite{AW02}, for instance, we can show that the language $Equal$ belongs to  $\mathrm{la}\mbox{-}\aeqs{1qqaf,constsize}$. It is, however, not entirely clear if there is a close relationship between accuracy rate and system size.

\item {\bf [Application to Search Problems]}
Our main target of this exposition has been decision problems (or equivalently, languages). It is, however, possible to expand the scope of our research setting to \emph{search problems} (as well as \emph{optimization problems}, \emph{counting problems}, etc.) within the current framework of AEQSs based on quantum quasi-automata families.
For this purpose, we first need to replace an acceptance/rejection criteria pair $(S^{(n)}_{acc},S^{(n)}_{rej})$ by a solution set $Sol^{(x)}$ of a target search problem for an input $x$.
We then say that an AEQS $\SSS$ \emph{solves} a search problem $P$ if, for any ``accessible'' instance $x$ given to the problem $P$, the ground state of $H^{(x)}_{fin}$ is $\sqrt{2}(1-\varepsilon)$-close to a certain quantum state in $QSol^{(x)}$, which is the Hilbert space spanned by $\{\qubit{u}\mid u\in Sol^{(x)}\}$. How can we develop a coherent theory based on this new setting?

\item {\bf [Based on a Uniform Model]}
Our 1qqaf's and 2qqaf's are, by their definitions, constructed in a nonuniform fashion; therefore, Hamiltonians generated by these machines are also nonuniform in nature.
Let $\TT$ denote a machine ``type'', such as logarithmic-space Turing machines. We say that a family $\GG=\{M_n\}_{n\in\nat}$ of quantum quasi-automata is \emph{$\TT$-uniform} if there exists a machine $N$ of type $\TT$ with an extra write-once output tape such that, for any input $x\in\Sigma^*$, $N$ on $x$ produces a ``description'' of $M_n$ on its output tape. The $\TT$-uniformity of the family $\GG$ leads to the uniformity of Hamiltonians that are generated by $\GG$. Since the uniformity notion is essential in computational complexity theory, it is fruitful to investigate a unform setting of the results stated in this exposition.

\item {\bf [Use of the Quantum Ising Model]}
In this exposition, we have used quantum quasi-automata families to generate Hamiltonians of AEQSs. In the literature (e.g., \cite{SMTC02}), Hamiltonians have also been described by a \emph{quantum Ising model}. The construction of final Hamiltonians in \cite{FGGS00} is in fact founded on the quantum Ising model. In certain cases, the use of such a model makes it easier to construct the desired Hamiltonians. It is therefore interesting to see how we apply quantum finite automata to introduce a simple, practical quantum Ising model.
\end{itemize}


\let\oldbibliography\thebibliography
\renewcommand{\thebibliography}[1]{%
  \oldbibliography{#1}%
  \setlength{\itemsep}{-2pt}%
}
\bibliographystyle{alpha}

\end{document}